\documentclass{amsart}
\usepackage{graphicx, xcolor} %
\usepackage{amssymb, amsfonts, amsthm, mathrsfs, hyperref, amsmath}
\usepackage{enumitem}
\usepackage{dsfont}
\usepackage{accents}
\usepackage{mathtools}
\usepackage{comment}
\allowdisplaybreaks

\numberwithin{equation}{section}
\newtheorem{theorem}{Theorem}[section]
\newtheorem{lemma}[theorem]{Lemma}
\newtheorem{proposition}[theorem]{Proposition}
\newtheorem{corollary}[theorem]{Corollary}
\newtheorem{definition}[theorem]{Definition}
\newtheorem{assumption}[theorem]{Assumption}

\newtheorem{remark}[theorem]{Remark}
\newtheorem{example}[theorem]{Example}

\newcommand{\llangle}{\langle\!\!\!\langle}
\newcommand{\rrangle}{\rangle\!\!\!\rangle}
\newcommand{\cA}{\mathcal{A}}
\newcommand{\cC}{\mathcal{C}}
\newcommand{\cD}{\mathcal{D}}
\newcommand{\cN}{\mathcal{N}}
\newcommand{\cX}{\mathcal{X}}
\newcommand{\R}{\mathbb{R}}

\newcommand{\dd}{\mathrm{d}}

\newcommand{\bE}{\mathbb{E}}
\newcommand{\ba}{\mathbf{a}}
\newcommand{\bb}{\mathbf{b}}

\newcommand{\be}{\mathbf{e}}
\newcommand{\bm}{\mathbf{m}}
\newcommand{\bh}{\mathbf{h}}
\newcommand{\bt}{\mathbf{t}}
\newcommand{\bg}{\mathbf{g}}

\newcommand{\bx}{\mathbf{x}}
\newcommand{\by}{\mathbf{y}}
\newcommand{\bz}{\mathbf{z}}
\newcommand{\bp}{\mathbf{p}}
\newcommand{\bP}{\mathbb{P}}
\newcommand{\bq}{\mathbf{q}}
\newcommand{\br}{\mathbf{r}}

\newcommand{\bv}{\mathbf{v}}
\newcommand{\bQ}{\mathbf{Q}}
\newcommand{\bR}{\mathbf{R}}

\newcommand{\bX}{\mathbf{X}}
\newcommand{\bY}{\mathbf{Y}}
\newcommand{\bZ}{\mathbf{Z}}
\newcommand{\sG}{\mathsf{G}}

\newcommand{\brho}{\boldsymbol{\rho}}
\newcommand{\bnu}{\boldsymbol{\nu}}
\newcommand{\bmu}{\boldsymbol{\mu}}
\newcommand{\bpi}{\boldsymbol{\pi}}
\newcommand{\bxi}{\boldsymbol{\xi}}

\newcommand{\barE}{\bar{\mathbf{e}}}
\DeclareMathOperator{\supp}{supp}

\DeclarePairedDelimiterX{\divg}[2]{(}{)}{%
  #1\;\delimsize\|\;#2%
}

\title{A Mathematical Study of the Excess Growth Rate}

\author{Steven~Campbell}
\address{Department of Statistics, Columbia University, New York, NY, USA}
\email{sc5314@columbia.edu}
\author{Ting-Kam~Leonard~Wong}
\address{Department of Statistical Sciences, University of Toronto, Toronto, ON, Canada}
\email{tkl.wong@utoronto.ca}


\begin{document}

\begin{abstract}
The excess growth rate, defined as the gap in Jensen's inequality for the logarithm, is a fundamental functional in portfolio theory. In this paper, we present a mathematical study motivated by information theory. We begin by establishing its properties and showing that it has rich connections with information-theoretic concepts such as the Helmholtz free energy, L. Campbell's measure of average code length and large deviations. Our main results consist of three axiomatic characterization theorems of the excess growth rate, in terms of (i) the relative entropy, (ii) the gap in Jensen's inequality, and (iii) the logarithmic divergence that generalizes the Bregman divergence. Furthermore, we study maximization of the excess growth rate and compare it with Kelly's growth optimal portfolio. Our results not only provide theoretical justifications of the significance of the excess growth rate, but also establish new connections between information theory and quantitative finance.
\end{abstract}

\maketitle

\section{Introduction} \label{sec:intro}
The excess growth rate is a fundamental logarithmic functional arising in portfolio theory. In this paper, we undertake a mathematical study of this quantity from the perspective of information theory. Specifically, we:
\begin{itemize}
\item demonstrate that the excess growth rate---analogous to the relative entropy---has rich connections with information theory and geometry, probability, and statistical physics;
\item formulate and prove three novel axiomatic characterization theorems of the excess growth rate; and
\item study maximization of the (expected) excess growth rate and compare this with Kelly's growth optimal portfolio.
\end{itemize}

We start with the definition of the excess growth rate. Throughout this paper, we denote the (closed) {\it unit simplex} in $\R^n$, $n \geq 1$, by
\[
\Delta_n \coloneqq\left\{\bx = (x_1, \ldots, x_n) \in[0,1]^n:\sum_{i = 1}^n x_i=1\right\}.\footnote{We adopt the convention that $\Delta_1 = \Delta_1^\circ\coloneqq\{1\}$.}
\]
Its relative interior is the {\it open simplex}
\[
\Delta_n^\circ \coloneqq \left\{\bx 
\in \Delta_n : x_i > 0 \text{ for } i \in [n]\right\},
\]
where $[n] \coloneqq \{1, \ldots, n\}$. The {\it support} of $\bx\in[0,\infty)^n$ is defined by
\begin{equation*} %
\supp(\bx)\coloneqq\{i\in[n]: x_i>0\}.
\end{equation*}
Define the {\it domain}
\begin{equation} \label{eqn:set.Dn}
\mathcal{D}_n \coloneqq \{(\bpi,\bR)\in\Delta_n\times [0,\infty)^n: \supp(\bpi)\subset \supp(\bR)\}
\end{equation}
as well as the {\it slice}
\begin{equation} \label{eqn:slices}
\begin{split}
\mathcal{D}_n(\bpi \mid\cdot) &\coloneqq \{\bR\in[0,\infty)^n: (\bpi,\bR)\in\mathcal{D}_n\}, \quad \bpi \in \Delta_n.
\end{split}
\end{equation}

\begin{definition} [Excess growth rate] \label{def:egr} { \ }
\begin{enumerate}
\item[(i)] For $n \geq 1$ and $(\bpi, \bR) \in \cD_n$, we define the excess growth rate of $\bR$ weighted by $\bpi$ by
\begin{equation} \label{eqn:EGR.financial}
\Gamma(\bpi, \bR) \coloneqq \log \left( \sum_{i \in \supp(\bpi)} \pi_i R_i\right) - \sum_{i \in \supp(\bpi)} \pi_i \log R_i.
\end{equation}
\item[(ii)] For $(\bpi, \bR) \in \cD_n$, define $\br \coloneqq \log \bR \coloneqq (r_i \coloneqq \log R_i)_{i \in \supp(\bpi)} \in \R^{\supp(\bpi)}$. The excess growth rate can be equivalently defined in terms of $\br$ by
\begin{equation} \label{eqn:egr.small.r}
\gamma(\bpi, \br) \coloneqq \Gamma(\bpi, \bR) = \log \left( \sum_{i \in \supp(\bpi)} \pi_i e^{r_i} \right) - \sum_{i \in \supp(\bpi)} \pi_i r_i.
\end{equation}
\end{enumerate}
\end{definition}

Note that we use the same symbol $\Gamma$ for each of the functions $\Gamma = \Gamma_n : \cD_n \rightarrow \R_+ \coloneqq [0, \infty)$, $n \geq 1$, and similarly for $\gamma$. We write $\Gamma_n$ and $\gamma_n$ (and similarly for other quantities) when it is helpful to emphasize the dimension.

To motivate Definition \ref{def:egr} financially, consider $n$ assets, such as stocks, whose prices are strictly positive.  For a given holding period like a month, let $\bpi = (\pi_1, \ldots, \pi_n) \in \Delta_n$ be the vector of initial portfolio weights, so that $\pi_i \geq 0$ is the initial proportion of wealth invested in asset $i$. By construction, we have $\sum_{i = 1}^n \pi_i = 1$. Suppose $R_i \in (0, \infty)$ is the {\it gross return} of asset $i$ over the holding period. That is, an investment of one dollar yields $R_i$ dollars at the end of the holding period. Then $r_i \coloneqq \log R_i$ is its {\it log return}. The gross return of the portfolio is the weighted sum $\sum_{i = 1}^n \pi_iR_i$. By Jensen's inequality, the portfolio's log return $\log \left(\sum_{i = 1}^n \pi_i R_i\right)$ is greater than or equal to $\sum_{i = 1}^n \pi_i \log R_i$, the weighted average log return of the assets. The excess growth rate is defined as the {\it gap} in Jensen's inequality. For technical purposes, we define the excess growth rate on the set $\cD_n$, so $R_i$ is  allowed to be $0$ whenever $\pi_i = 0$. When $n = 1$, we have $\Gamma_1 \equiv 0$ since $\bpi$ reduces to a point mass. This is both mathematically and financially trivial but serves as a base case. %

It is also useful to think of the excess growth rate as a {\it divergence} between the initial prices $\bX$ and the final prices $\bY$ of the assets regarded as elements of $[0, \infty)^n$. This is analogous to the relative entropy which is a divergence between a pair of probability distributions. See Lemma \ref{lem:link.rel.entr} below for the precise relationship between the two quantities.

\begin{definition} [Excess growth rate as a divergence] \label{def:egr.divergence}
For $n \geq 1$ and $\bpi \in \Delta_n$, we define $\Gamma_{\bpi}\divg{\cdot}{\cdot} : \cD_n(\bpi \mid \cdot) \times \cD_n(\bpi \mid \cdot) \rightarrow \R_+$ by
\begin{equation} \label{eqn:excess.growth.divergence}
\Gamma_{\bpi}\divg{\bY}{\bX} \coloneqq  \log \left( \sum_{i \in \supp(\bpi)} \pi_i \frac{Y_i}{X_i} \right) - \sum_{i \in \supp(\bpi)} \pi_i \log \frac{Y_i}{X_i},
\end{equation}
where $\bX = (X_1, \ldots, X_n)$ and $\bY = (Y_1, \ldots, Y_n)$.
\end{definition}

Clearly, we have $\Gamma_{\bpi}\divg{a\bY}{a \bX} = \Gamma_{\bpi}\divg{\bY}{\bX}$ for $a > 0$. This is a special case of {\it num\'{e}raire invariance} which will be formulated in Proposition \ref{prop:numeraire.invariance} below. It is also clear that generally $\Gamma_{\bpi}\divg{\bY}{\bX} \neq \Gamma_{\bpi}\divg{\bX}{\bY}$. Financially, this means that the excess growth rate is not invariant under time reversal, as expected.

To the best of our knowledge, the concept of excess growth in finance was first introduced in \cite{FS82}. Later, it became an essential concept in {\it stochastic portfolio theory} \cite{F02, FK09}. Independently, the authors of \cite{BF92} introduced the same quantity and called it the {\it diversification return}. Our definition follows that of \cite{PW13}. In Section \ref{sec:egr.properties}, we establish mathematical properties of the excess growth rate that are used in the axiomatic characterizations. Further discussion of financial applications and the related
literature is given in Appendix \ref{sec:egr.applications} which may be skipped without loss of continuity. 

\begin{remark}
The concept of excess growth rate can be extended to a general measure-theoretic framework. Specifically, suppose $\Omega$ is a measurable space. If $\pi$ is a probability
measure on $\Omega$  and $R$ is a non-negative random variable on $\Omega$ which is $\pi$-almost surely positive, we may define
\begin{equation} \label{eqn:EGR.general}
\Gamma(\pi, R) \coloneqq \log \left(\int_{\Omega} R \dd \pi\right) - \int_{\Omega} \log R \dd \pi.
\end{equation}
That is, we replace the weighted sums in \eqref{eqn:EGR.financial} by integrals. For concreteness and the financial applications we have in mind, we focus on the discrete setting in this paper. Nevertheless, we believe many results, including some axiomatic characterizations, can be extended to the context of \eqref{eqn:EGR.general}.
\end{remark}

\subsection{Summary of results and organization}
Our first contribution in this paper is to show that the excess growth rate, similar to the relative entropy, is deeply connected to familiar concepts in information theory, statistical physics and probability. In particular, we show that:\footnote{These connections are not used in the axiomatic characterizations.}
\begin{itemize}
\item the excess growth rate can be interpreted in terms of the {\it Helmholtz free energy}, and has a {\it variational representation} (Section \ref{sec:variational});
\item the difference between L.~Campbell's measure of average code length \cite{campbell1965coding} and Shannon's one can be expressed in terms of the excess growth rate (Section \ref{sec:Campbell}); 
\item the excess growth rate emerges in a large deviation principle of the {\it scaled Dirichlet distribution} (Definition \ref{def:scaled.Dirichlet}), analogous to how the relative entropy features in Sanov's theorem (Section \ref{sec:probabilistic}).
\end{itemize}
Several more connections, including correspondences with the {\it R\'{e}nyi divergence} and {\it cross-entropy}, and the {\it logarithmic divergence} \cite{PW16, PW18, W18} in information geometry \cite{A16}, can also be found in the paper. In fact, the excess growth rate can be expressed directly in terms of the relative entropy using algebraic operations on the simplex in compositional data analysis (Lemma \ref{lem:link.rel.entr}). Nevertheless, our body of results goes well beyond this identity.

Our main contribution, presented in Section \ref{sec:characterization}, is a collection of three {\it axiomatic characterization theorems} that uniquely determine the excess growth rate, possibly up to a multiplicative constant, based on natural invariance and algebraic properties. Axiomatic characterizations of various information-theoretic quantities have been studied by many researchers, beginning with Shannon himself~\cite[Theorem 2]{shannon1948mathematical} (another classic is R\'{e}nyi's paper \cite{R61}). To give a flavor of some of the ideas involved, consider the fundamental {\it additive property} of the Shannon entropy:
\[
H( \bp \otimes \bq) = H(\bp) + H(\bq),
\]
where $\bp \otimes \bq$ denotes the product distribution. This property is closely related to the {\it functional equation} $f(xy) = f(x) + f(y)$, $x, y > 0$, whose general solution (assuming only that $f$ is Lebesgue measurable) is $f(x) = c \log x$, $c \in \R$.\footnote{This functional equation is equivalent to {\it Cauchy's equation} \eqref{eqn:Cauchy} which plays an important role in the proof of our second characterization theorem.} A comprehensive mathematical study of axiomatic characterizations of information measures and related quantities, as well as detailed historical discussions, can be found in Leinster's book \cite{leinster2021entropy} which is primarily motivated by diversity measures in biology. In fact, Leinster's book provided the initial impetus for our work.\footnote{We thank Martin Larsson for bringing this reference to our attention.} 

Our three characterization theorems highlight different aspects of the excess growth rate and further reinforce its importance:%
\begin{itemize}
\item Our most significant characterization result (Theorem \ref{thm:characterization.rel.entr}), proved in Section \ref{sec:relative entropy}, shows that the excess growth rate is completely determined by several natural financial properties, including num\'{e}raire invariance and a {\it chain rule} (Proposition \ref{prop:chain.rule.1}). Our proof is based on a characterization of relative entropy, its relation with the excess growth rate (Lemma \ref{lem:link.rel.entr}), as well as a delicate analysis of boundary values.

\item In Section \ref{sec:Jensen.gap}, we characterize axiomatically the {\it gap} in Jensen's inequality, for a general ``generating function'', and show in this setting that the logarithmic case, which leads to the excess growth rate, is characterized by num\'{e}raire invariance (Theorem \ref{thm:characterization.Jensen}). 
\item In Section \ref{sec:log.divergence}, we exploit the fact that the excess growth rate is a member of the family of {\it logarithmic divergences} introduced by Pal and the second author \cite{PW16}. The logarithmic divergence generalizes the Bregman divergence, which includes the squared Euclidean distance and relative entropy on the simplex as special cases. We show in Theorem \ref{thm:egr.L.divergence} that the excess growth rate is the unique logarithmic divergence which is {\it perturbation invariant}; this is closely related to num\'{e}raire invariance. A by-product is a new characterization of the (negative) {\it cross-entropy} within the family of exponentially concave functions on the open simplex.
\end{itemize}

The significance of the excess growth rate in portfolio theory leads naturally to maximization of this quantity. In Section \ref{sec:optimization} we study two versions of this problem, first in a deterministic setting ($\max_{\bpi \in \Delta_n} \gamma(\bpi,\br)$ where $\br$ is fixed), then in a probabilistic setting where we maximize the expected excess growth rate $\bE [ \gamma(\bpi,\br)] $ assuming $\br$ is a random vector. In the deterministic case, we derive an explicit characterization of the solution and, via a variational representation, link it with the perspective function in convex analysis. In the probabilistic case, we derive a first-order condition for the optimizer and compare this problem with the classical {\it growth optimal portfolio} \cite[Chapter 16]{CT06}. 

As discussed in Section \ref{sec:IT.literature} below, information theory and quantitative finance share deep connections. In this paper, we show that the excess growth rate fosters new synergies between the two fields. Our results suggest many directions for future research, some of which are discussed in Section \ref{sec:conclusion}.

\subsection{Information theory and quantitative finance} \label{sec:IT.literature}
Interactions between information theory and quantitative finance began soon after Shannon's seminal paper \cite{shannon1948mathematical}. In \cite{kelly1956new}, Kelly showed that in repeated investment or gambling situations, the value of side information can be quantified by mutual information, a fundamental information-theoretic quantity that arises in the definition of channel capacity. Kelly's work (and that of Breiman \cite{breiman1961optimal}, among others) led to the concept of growth optimal portfolio, also called the num\'{e}raire portfolio, which has profound implications in finance. See \cite{maclean2011kelly} for a volume of collected papers. Intuitively, optimal investment and information theory are fundamentally related because successful investment and efficient data transmission/extraction both hinge on prediction (asset returns or source alphabets). For example, in \cite{algoet1988asymptotic} an asymptotic equipartition property was established in the context of growth optimal investment. Among the many subsequent works, Cover's paper \cite{C91}, published in the inaugural issue of the journal \emph{Mathematical Finance}, was particularly impactful. This paper introduced the universal portfolio, which serves as the financial analogue of universal coding and as a benchmark for online portfolio selection algorithms \cite{LH14}. In \cite{OJ23}, it was shown that regret guarantees of universal portfolio algorithms imply time-uniform concentration inequalities for bounded random variables. For further details and other classical connections, we refer the reader to Chapters 6 and 16 of \cite{CT06}. Recently, the financial perspective on information theory has been fruitfully extended to optimal hypothesis testing using $e$-values, see \cite{larsson2025numeraire, ramdas2024hypothesis}.

\section{Excess growth rate: properties and interconnections} \label{sec:2}
In Section \ref{sec:egr.properties}, we establish some mathematical properties of the excess growth rate. The underlying financial intuition will be carefully explained but is not required to follow the mathematical development. In Section \ref{sec:relative entropy}, we show that these properties (as well as Lebesgue measurability) uniquely characterize the excess growth rate up to a multiplicative constant. Then, in Sections \ref{sec:variational}--\ref{sec:probabilistic}, we show that the excess growth rate arises naturally not only in finance but also in statistical physics, information theory, and probability theory.

\subsection{Mathematical properties}  \label{sec:egr.properties}
We begin with two properties that are immediate from the definition. Given $\bx \in \R^n$ and a permutation $\sigma$ of $[n]$, we define 
\[
\bx\sigma \coloneqq (x_{\sigma(1)}, \ldots, x_{\sigma(n)}) \in \R^n.
\]

\begin{proposition}[Permutation invariance] \label{prop:permutation.invariance}
For any $(\bpi, \bR) \in \cD_n$ and permutation $\sigma$ of $[n]$, we have $\Gamma(\bpi\sigma, \bR\sigma) = \Gamma
(\bpi, \bR)$.
\end{proposition}

\begin{proposition}[Dependence on the support] \label{prop:support}
For $\bpi \in \Delta_n$ and $\bR, \bR' \in \cD_n(\bpi \mid \cdot)$, we have $\Gamma(\bpi, \bR) = \Gamma(\bpi, \bR')$ if $R_i = R_i'$ for $i \in \supp(\bpi)$. In particular, $\Gamma(\bpi, \bR) = 0$ if $\bR$ is constant on $\supp(\bpi)$.
\end{proposition}

Together, Propositions \ref{prop:permutation.invariance} and \ref{prop:support} state that the excess growth rate is invariant under relabeling the assets (and their returns), and depends only on the assets that are held in the portfolio. 

\begin{proposition}[Num\'{e}raire invariance] \label{prop:numeraire.invariance}
For $(\bpi, \bR) \in \cD_n$ and $a > 0$, we have 
$\Gamma(\bpi, a\bR) = \Gamma
(\bpi, \bR)$. Equivalently, $\Gamma_{\bpi}\divg{b\bY}{a \bX} = \Gamma_{\bpi}\divg{\bY}{\bX}$ for any $\bX, \bY \in \cD_n(\bpi \mid \cdot)$ and $a, b > 0$.
\end{proposition}
\begin{proof}
From \eqref{eqn:EGR.financial} and the additive property of the logarithm, we have
\begin{align*}
\Gamma(\bpi, a\bR) &= \log \left( \sum_{i \in \supp(\bpi)} \pi_i (aR_i) \right) - \sum_{i \in \supp(\bpi)} \pi_i \log (aR_i) \\
  &= \log \left( \sum_{i \in \supp(\bpi)} \pi_i R_i \right) + \log a - \sum_{i \in \supp(\bpi)} \pi_i \log R_i - \log a\\
  &= \Gamma(\bpi, \bR).\qedhere
\end{align*}
\end{proof}

Here is the financial interpretation. Suppose that we express the gross return $R_i$ of asset $i \in \supp(\bpi)$ as $Y_i/X_i$, where $X_i$ and $Y_i$ are, respectively, the initial and final prices. For concreteness, let $X_i$ and $Y_i$ be the {\it dollar values}. In financial terms, we say that the {\it num\'{e}raire} is cash (with respect to a fixed currency). Suppose that we measure prices in terms of another positive quantity, such as the value of the S\&P 500 Index or a risk-free bank account, which moves from $Q$ to $Q'$. That is, we define the {\it relative prices} of asset $i$ by $\tilde{X}_i = X_i/Q$ and $\tilde{Y}_i = Y_i/Q'$; these are the prices under the new num\'{e}raire. Then, the {\it relative gross return} is given by
\begin{equation} \label{eqn:relative.return}
\tilde{R}_i \coloneqq \frac{\tilde{Y}_i}{\tilde{X}_i} = \frac{Y_i/Q'}{X_i/Q} = \frac{Q}{Q'} \frac{Y_i}{X_i} = \frac{Q}{Q'} R_i.
\end{equation}
Thus, we have $\tilde{\bR} \coloneqq (\tilde{R}_1, \ldots, \tilde{R}_n) = a\bR$, where $a = \frac{Q}{Q'} > 0$. Thus, the excess growth rate is independent of the choice of the num\'{e}raire (and, in particular, the risk-free interest rate).

By num\'{e}raire invariance, for each $n$, the function $\Gamma : \cD_n \rightarrow \R_+$ is determined by its restriction to the set
\begin{equation} \label{eqn:An}
\cA_n \coloneqq \cD_n \cap (\Delta_n \times \Delta_n).
\end{equation}
We define the slice $\cA_n(\bpi \mid \cdot)$ %
analogously (see \eqref{eqn:slices}). Specifically, for $(\bpi, \bR) \in \cD_n$, we have
\begin{equation} \label{eqn:egr.renormalize}
\Gamma(\bpi, \bR) = \Gamma(\bpi, \cC_{\bpi}[\bR]),
\end{equation}
where $\cC_{\bpi}: \cD_n(\bpi \mid \cdot) \rightarrow \cA_n(\bpi \mid \cdot)$ is the {\it closure} with respect to (the support of) $\bpi$, defined by
\begin{equation} \label{eqn:closure}
\left( \cC_{\bpi}[\bx] \right)_i \coloneqq 
\left\{\begin{array}{ll}
        x_i / \sum_{j \in \supp(\bpi)} x_j, & \text{if } i \in \supp(\bpi),\\
        0, & \text{otherwise.}\\
        \end{array}\right.
\end{equation}
If the relevant support is $[n]$ (so that $\bx \in (0, \infty)^n$), we simply write $\cC[\bx]$ which is an element of $\Delta_n^{\circ}$. We introduce several related algebraic operations for later use:
\begin{itemize}
\item {\it Hadamard (componentwise) product}:
\begin{equation} \label{eqn:Hadamard}
(\bx \by)_i \coloneqq x_iy_i, \quad \bx, \by \in \R^n.
\end{equation}
\item {\it Componentwise inverse}:
\[
(\bx^{-1})_i \coloneqq \frac{1}{x_i}, \quad \bx \in (0, \infty)^n.
\]
\item {\it Perturbation operation} with respect to $\bpi \in \Delta_n$:
\[
\bx \oplus_{\bpi} \by \coloneqq \cC_{\bpi}[\bx \by], \quad \bx, \by \in \cA_n(\bpi \mid \cdot).
\]
We write $\bx \oplus \by$ when the support is $[n]$. 
\item {\it Powering operation:}
\[
\alpha \odot \bx \coloneqq \cC[(x_i^{\alpha})_{1 \leq i \leq n}], \quad (\bx, \alpha) \in \Delta_n^{\circ} \times \R.
\]
\end{itemize}

It is well known in {\it compositional data analysis} \cite{A94, egozcue2003isometric, EA21} that the open simplex $\Delta_n^{\circ}$ becomes an $(n - 1)$-dimensional real vector space if we regard $\oplus$ as vector addition and $\odot$ as scalar multiplication. The additive identity (zero element) is the {\it barycenter} $\barE = \barE_n \coloneqq (1/n, \ldots, 1/n)$, and vector subtraction is given by
\[
\bx \ominus \by \coloneqq \cC[\bx \by^{-1}], \quad \bx, \by \in \Delta_n^{\circ}.
\]
For a general $\bpi \in \Delta_n$ (whose support may be a strict subset of $[n]$), and for $\bx, \by \in \cA_n(\bpi \mid \cdot)$, we define the generalized difference $\bx \ominus_{\bpi} \by \in \cA_n(\bpi \mid \cdot)$ by
\[
(\bx \ominus_{\bpi} \by)_i \coloneqq 
\left\{\begin{array}{ll}
        (x_i/y_i) / \sum_{j \in \supp(\bpi)} (x_j/y_j), & \text{if } i \in \supp(\bpi),\\
        0, & \text{otherwise.}\\
        \end{array}\right.
\]

The {\it chain rules}, to be stated next, tell us how to decompose the excess growth rate of a {\it composite portfolio}, i.e., a portfolio of portfolios. Let $n, k_1, \ldots, k_n \geq 1$ be integers and let
\[
\bpi \in \Delta_n, \quad \bp^1 \in \Delta_{k_1}, \ldots, \bp^n \in \Delta_{k_n}.
\]
Write $\bp^i = (p_1^i, \ldots, p_{k_i}^i)$ and $\bp = (\bp^1, \ldots, \bp^n)$. The {\it composite distribution} $\bpi \circ \bp$ is
\begin{equation} \label{eqn:composite.distribution}
\begin{split}
\bpi \circ \bp &\coloneqq (\pi_1 p_1^1, \ldots, \pi_1 p_{k_1}^1, \ldots,  \pi_n p_1^n, \ldots, \pi_n p_{k_n}^n)\\
  &= (\pi_1 \bp^1, \ldots, \pi_n \bp^n) \in \Delta_{k_1 + \cdots + k_n}.
\end{split}
\end{equation}
We index its components by $(\bpi\circ \bp)_{i,j} = \pi_i p_j^i$, where $(i, j) \in [n] \times [k_i]$. Financially, we may think of the $i$-th conditional distribution $\bp^i$ as a portfolio consisting of $k_i$ assets, and the composite portfolio holds the $n$ portfolios as individual assets. We allow some of the assets to overlap. For example, the capital corresponding to the weights $p_1^1$ and $p_1^2$ can be invested in the same asset. This can be enforced by letting the gross returns $R_1^1$ and $R_1^2$ be equal. Probabilistically, ${\bf p}$ represents a collection of conditional distributions.

The following version of the chain rule was formulated in \cite{PW13}. Throughout the paper, we denote the Euclidean inner product by $\langle \bx, \by \rangle$.

\begin{proposition}[Chain rule (first version)] \label{prop:chain.rule.1}
Let $n, k_1, \ldots, k_n \geq 1$,
\[
\bpi \in \Delta_n, \quad \bp = (\bp^1, \ldots, \bp^n) \in \Delta_{k_1} \times \cdots \times \Delta_{k_n},
\]
and let
\[
\bR = (\bR^1, \ldots, \bR^n) \in \prod_{i = 1}^n \cD_{k_i}(\bp^i \mid \cdot).
\]
Denote
\[
\llangle \bp, \bR\rrangle \coloneqq (\langle \bp^1, \bR^1 \rangle, \ldots, \langle \bp^n, \bR^n\rangle) \in (0, \infty)^n.
\]
Then
\begin{equation} \label{eqn:egr.chain.rule.1}
\begin{split}
&\Gamma(\bpi \circ \bp, \bR) 
= \Gamma(\bpi, \llangle \bp, \bR \rrangle) + \sum_{i \in \supp(\bpi)} \pi_i \Gamma (\bp^i, \bR^i).
\end{split}
\end{equation}
\end{proposition}

Consider a portfolio of $k_1 + \cdots + k_n$ assets with weights $\bpi \circ \bp$. Its excess growth rate is equal to $\Gamma(\bpi \circ \bp, \bR)$. The chain rule states that we may decompose it as a sum of two parts. First, consider the {\it portfolio of portfolios}, where ``asset $i$'' is the $i$-th portfolio with gross return $\langle \bp^i, \bR^i \rangle$. This gives the excess growth rate $\Gamma(\bpi, \llangle \bp, \bR \rrangle)$. Note that the gross returns $R_j^i$ of the individual ``atomic'' assets enter indirectly through $\llangle \bp, \bR \rrangle$ (in \cite{leinster2021entropy} this property is called {\it modularity}). The second term is the weighted sum of the excess growth rates $\Gamma (\bp^i, \bR^i)$ of the individual portfolios. The reader should note that this property is reminiscent of the chain rule of relative entropy (see \eqref{eqn:relative.entropy.chain.rule}). The precise algebraic relationships between the relative entropy and excess growth rate are given in Lemma \ref{lem:link.rel.entr}.

\begin{proof}[Proof of Proposition \ref{prop:chain.rule.1}]
By definition of $\Gamma(\bpi \circ \bp, \bR)$, we have
\begin{align*}
&\Gamma(\bpi \circ \bp, \bR) \\
&= \log \left( \sum_{i \in \supp(\bpi)} \sum_{j \in \supp(\bp^i)} \pi_i p_j^i R_j^i\right) - \sum_{i \in \supp(\bpi)} \sum_{j \in \supp(\bp^i)} \pi_i p_j^i \log R_j^i\\
&=\log \left( \sum_{i \in \supp(\bpi)} \pi_i \langle \bp^i, \bR^i \rangle \right) - \sum_{i \in \supp(\bpi)} \pi_i \left( \sum_{j \in \supp(\bp^i)} p_j^i \log R_j^i \right) \\
&= \log \left( \sum_{i \in \supp(\bpi)} \pi_i \langle \bp^i, \bR^i \rangle \right) - \sum_{i \in \supp(\bpi)} \pi_i \log \langle \bp^i, \bR^i \rangle \\
&\quad + \sum_{i \in \supp(\bpi)} \pi_i \left( \log \langle \bp^i, \bR^i \rangle - \sum_{j \in \supp(\bp^i)} p_j^i \log R_j^i \right) \\
&= \Gamma(\bpi, \llangle \bp, \bR \rrangle) + \sum_{i \in \supp(\bpi)} \pi_i \Gamma(\bp^i, \bR^i).\qedhere
\end{align*}
\end{proof}

The chain rule can be stated in a slightly more general form. To motivate it, suppose that the $n$ portfolios hold assets in $n$ countries with {\it different currencies} (num\'{e}raires). To compute the excess growth rate of the composite portfolio, we need to express all returns in terms of a common currency. In the statement below, we think of $a_i \geq 0$ as the conversion factor for the assets in the $i$-th portfolio; when $a_i > 0$, it plays the role of the factor $Q/Q'$ in \eqref{eqn:relative.return}. 

\begin{proposition}[Chain rule (general)] \label{prop:chain.rule.2}
Let $n, k_1, \ldots, k_n \geq 1$,
\[
\bpi \in \Delta_n, \quad \bp = (\bp^1, \ldots, \bp^n) \in \Delta_{k_1} \times \cdots \times \Delta_{k_n},
\]
\[
\bR = (\bR^1, \ldots, \bR^n) \in \prod_{i = 1}^n \cD_{k_i}(\bp^i \mid \cdot) \quad \text{and} \quad \ba = (a_1, \ldots, a_n) \in \cD_n(\bpi \mid \cdot).
\]
Define
\[
\ba \circ \bR \coloneqq (a_1 \bR^1, \ldots, a_n \bR^n) \in \cD_{k_1 + \cdots + k_n}(\bpi \circ \bp \mid \cdot).
\]
Then, we have
\begin{equation} \label{eqn:egr.chain.rule.2}
\Gamma(\bpi \circ \bp, \ba \circ \bR) = \Gamma(\bpi, \ba  \llangle \bp, \bR \rrangle) + \sum_{i \in \supp(\bpi)} \pi_i \Gamma (\bp^i, \bR^i).
\end{equation}
Here $\ba  \llangle \bp, \bR \rrangle$ is the componentwise product as defined in \eqref{eqn:Hadamard}.
\end{proposition}
\begin{proof}
While \eqref{eqn:egr.chain.rule.2} can be proved directly from the definition of $\Gamma$, we show that it is a consequence of the support properties in Proposition \ref{prop:support}, num\'{e}raire invariance (Proposition \ref{prop:numeraire.invariance}) and the first chain rule (Proposition \ref{prop:chain.rule.1}). In Lemma \ref{lem:A6.A7} below, we show that given the support properties, \eqref{eqn:egr.chain.rule.2} is equivalent to num\'{e}raire invariance and the first chain rule.

Let $(\bpi,\ba)\in\cD_n$ and let
$(\bp^i,\bR^i)\in\cD_{k_i}$, $i=1,\ldots,n$. Define
$\widetilde{\ba}=(\widetilde{a}_1,\dots,\widetilde{a}_n)$ by $\widetilde a_i=a_i$ on $\supp(\bpi)$ and $\widetilde a_i=1$ otherwise. Then $\widetilde a_i>0$ for all $i$. Moreover,
$\ba\circ\bR$ and $\widetilde{\ba}\circ\bR$ agree on
$\supp(\bpi\circ\bp)$, while
$\widetilde{\ba}\llangle\bp,\bR\rrangle$ and
$\ba\llangle\bp,\bR\rrangle$ agree on $\supp(\bpi)$. By the
support condition,
\begin{equation}\label{eqn:tilde.identity.1}
\Gamma(\bpi\circ\bp,\ba\circ\bR)
=
\Gamma(\bpi\circ\bp,\widetilde{\ba}\circ\bR),
\quad
\Gamma(\bpi,\widetilde{\ba}\llangle\bp,\bR\rrangle)
=
\Gamma(\bpi,\ba\llangle\bp,\bR\rrangle).
\end{equation}
Set $\widetilde{\bR}^i=\widetilde a_i\bR^i$. Since $\widetilde a_i>0$,
the first chain rule applies to
$\widetilde{\bR}=(\widetilde{\bR}^1,\ldots,\widetilde{\bR}^n)$ and gives
\[
\Gamma(\bpi \circ \bp, \widetilde{\bR}) 
= \Gamma(\bpi, \llangle \bp, \widetilde{\bR} \rrangle) + \sum_{i \in \supp(\bpi)} \pi_i \Gamma(\bp^i, \widetilde{\bR}^i).
\]
Equivalently,
\begin{equation}\label{eqn:tilde.identity.2}
\Gamma(\bpi\circ\bp,\widetilde{\ba}\circ\bR)
=
\Gamma(\bpi,\widetilde{\ba}\llangle\bp,\bR\rrangle)
+
\sum_{i \in \supp(\bpi)}\pi_i\Gamma(\bp^i,\widetilde a_i\bR^i).
\end{equation}
By num\'eraire invariance,
$\Gamma(\bp^i,\widetilde a_i\bR^i)=\Gamma(\bp^i,\bR^i)$ for each $i$. Combining this observation with \eqref{eqn:tilde.identity.1} and \eqref{eqn:tilde.identity.2} proves \eqref{eqn:egr.chain.rule.2}.    
\end{proof}

Further insight about the chain rule can be obtained by restating it in probabilistic language.\footnote{This discussion, including Example \ref{eg:chain.rule} below, is inspired by Ruodu Wang.} Consider the functional
\begin{equation} \label{eqn:Gamma.probabilistic}
\tilde{\Gamma}(R) \coloneqq \log \bE[R] - \bE [ \log R],
\end{equation}
which is defined for any random variable $R$ that is almost surely positive and for which
$\bE[\log R]$ is finite. From \eqref{eqn:EGR.general}, this is equal to $\Gamma(\bpi, \bR)$ if $\bP(R = R_i) = \pi_i$. Given $(\bpi, \bp, \ba, \bR)$ in the context of Proposition \ref{prop:chain.rule.2}, consider random variables $X$, $a$ and $R$ whose joint distribution is specified by
\begin{equation*}
\begin{split}
&\bP(X = i) = \pi_i, \quad i \in [n],\\
&\bP(a = a_i \mid X = i) = 1, \quad i \in [n],\\
&\bP(R = R_j^i \mid X = i) = p_j^i, \quad j \in [k_i].
\end{split}
\end{equation*}
Note that $a$ is a function of $X$. The terms $\Gamma (\bp^i, \bR^i)$ correspond to 
\begin{equation} \label{eqn:Gamma.probabilistic.conditional}
\tilde{\Gamma}(R \mid X) \coloneqq \log \bE[R \mid X] - \bE[\log R \mid X],
\end{equation}
which is the conditional version of \eqref{eqn:Gamma.probabilistic}. The general chain rule \eqref{eqn:egr.chain.rule.2} is equivalent to the identity
\begin{equation} \label{eqn:chain.rule.probabilistic}
\tilde{\Gamma} (aR) = \tilde{\Gamma} (a\bE[R \mid X]) + \bE[ \tilde{\Gamma}( R \mid X)],
\end{equation}
and \eqref{eqn:egr.chain.rule.1} is equivalent to the special case $a \equiv 1$. 

\begin{example} \label{eg:chain.rule}
Consider the excess growth rate $\gamma(\bpi, \br)$ in terms of the log return, see \eqref{eqn:egr.small.r}. Taylor expanding about $\br = 0$ shows that
\begin{equation} \label{eqn:Taylor}
\gamma(\bpi, \br) = \frac{1}{2}\mathrm{Var}_{\bpi}(\br) + o\Big(\sum_{i \in \supp(\bpi)} r_i^2\Big),
\end{equation}
where $\mathrm{Var}_{\bpi}(\br)$ is the variance of $\br$ weighted by $\bpi$:
\[
\mathrm{Var}_{\bpi}(\br) \coloneqq \sum_{i \in \supp(\bpi)} \pi_i r_i^2 - \Big(\sum_{i \in \supp(\bpi)} \pi_i r_i \Big)^2.
\]

Consider the functional $\check{\Gamma}(\bpi, \bR) \coloneqq \mathrm{Var}_{\bpi}(\br)$ which, in analogy to \eqref{eqn:Gamma.probabilistic.conditional}, corresponds to the functionals
\[
\check{\Gamma}(R) \coloneqq \mathrm{Var}(\log R), \quad \check{\Gamma}(R \mid X) \coloneqq \mathrm{Var}(\log R \mid X).
\]
The following statements hold:
\begin{itemize}
\item $\check{\Gamma}$ is num\'{e}raire-invariant: if $a > 0$ is a constant, then
\[
\check{\Gamma}(aR) = \mathrm{Var}(\log R + \log a) = \mathrm{Var}(\log R) = \check{\Gamma}(R).
\]
\item $\check{\Gamma}$ does not satisfy the first chain rule. By the law of total variance, we have
\begin{equation*}
\begin{split}
\check{\Gamma}(R) &= \mathrm{Var}(\log R) \\
&= \mathrm{Var}(\bE[\log R \mid X]) + \bE[ \mathrm{Var}(\log R \mid X)] \\
&= \check{\Gamma}(\exp(\bE[\log R \mid X])) + \bE[\check{\Gamma}( R \mid X)],
\end{split}
\end{equation*}
which may be interpreted as another chain rule. But \eqref{eqn:chain.rule.probabilistic} requires that the first term is $\check{\Gamma}( \bE[R \mid X])$. 
\end{itemize}

We note that the quadratic term in \eqref{eqn:Taylor} is also a special case of Rao's quadratic entropy with \emph{dissimilarity matrix} $D_{ij}(\br)=\frac12(r_i-r_j)^2$ (cf.~\cite[Section~1.2]{qi2023geometric}). Indeed, in the notation of \cite{qi2023geometric},
\[ q_{D(\br)}(\bpi) \coloneqq \frac12\bpi^\top D(\br)\bpi
= \frac12\mathrm{Var}_{\bpi}(\br).\]
Here $D(\br)$ is a Euclidean distance matrix whose entries are the squared distances between the points $r_i/\sqrt{2}$.
\end{example}

\begin{remark}
Recall that the {\it relative entropy} 
$H\divg{\cdot}{\cdot}$ is given on $\Delta_n \times \Delta_n$, $n \geq 1$, by
\begin{equation}\label{eqn:rel.ent.and.lse}
H\divg{\bp}{\bq} = 
\left\{\begin{array}{ll}
        \sum_{i = 1}^n p_i \log \frac{p_i}{q_i}, & \text{if } \supp(\bp) \subset \supp(\bq);\\
        +\infty, & \text{otherwise.}
        \end{array}\right.
\end{equation} 
In Lemma \ref{lem:link.rel.entr} below, we will show that the excess growth rate can be expressed in terms of the relative entropy via $\Gamma(\bpi, \br) = H\divg{ \bpi}{ \bpi \oplus_{\bpi} \br}$ for $(\bpi,\br)\in\cA_n$. The properties of the excess growth rate discussed above can also be derived using this identity and familiar properties of relative entropy. In fact, this correspondence underlies the proof of our first characterization theorem (Theorem \ref{thm:characterization.rel.entr}).
\end{remark}

\subsection{Free energy and variational representation} \label{sec:variational}
We relate the excess growth rate with the Helmholtz free energy and state a variational representation. Consider a physical system with $n$ possible states and let $\bpi \in \Delta_n$ be a reference distribution that represents the multiplicities of states. Let $\mathbf{E} = (E_1, \ldots, E_n) \in \R^n$ represent the energies of the states and $\beta > 0$ be the inverse temperature. Consider the (weighted) {\it Gibbs distribution} $\bp^{\star} = \bp^{\star}(\bpi, \mathbf{E}, \beta) \in \Delta_n$ given by
\begin{equation} \label{eqn:Gibbs}
p_i^{\star} \coloneqq 
\left\{\begin{array}{ll}
        \frac{1}{Z(\bpi, \mathbf{E}, \beta)} \pi_i e^{-\beta E_i}, & \text{if } i \in \supp(\bpi);\\
        0, & \text{otherwise,}\\
        \end{array}\right.
\end{equation}
where $Z(\bpi, \mathbf{E}, \beta)$ is the {\it partition function} given by
\[
Z(\bpi, \mathbf{E}, \beta) \coloneqq \sum_{j \in \supp(\bpi)} \pi_j e^{-\beta E_j}.
\]
By construction, we have $\supp(\bp^{\star}) \subset \supp(\bpi)$. In this context, the {\it Helmholtz free energy} is the quantity
\begin{equation} \label{eqn:Helmholtz}
A(\bpi, \mathbf{E}, \beta) \coloneqq - \frac{1}{\beta} \log Z(\bpi, \mathbf{E}, \beta) = \frac{-1}{\beta} \log \left( \sum_{j \in \supp(\bpi)} \pi_j e^{-\beta E_j} \right).
\end{equation}
(See, for example, \cite[Chapter 3]{PB21} for the physical background.) On the other hand, the average energy of the system with respect to the reference distribution $\bpi$ is given by
\[
U(\bpi, \mathbf{E}) \coloneqq \sum_{j \in \supp(\bpi)} \pi_j E_j.
\]
Letting $\bR = \exp(-\beta \mathbf{E}) \in (0, \infty)^n$, we have the identity
\begin{equation} \label{eqn:free.energy.identity}
\Gamma(\bpi, \bR) = \beta ( U(\bpi, \mathbf{E}) - A(\bpi, \mathbf{E}, \beta)).
\end{equation}
That is, the excess growth rate is, up to a multiplicative constant, the difference between the reference average energy and the Helmholtz free energy.

The distribution $\bp^{\star}$ given by $\eqref{eqn:Gibbs}$ can be justified by {\it Gibbs' variational principle} \cite[Proposition 4.7]{PW25} of the free energy (or equivalently the log--sum--exp in \eqref{eqn:egr.small.r}):
\begin{equation} \label{eqn:free.energy.variational}
A(\bpi, \mathbf{E}, \beta) = \inf_{\bp \in \Delta_n} \left\{ \langle \bp, \mathbf{E} \rangle + \frac{1}{\beta} H\divg{\bp}{\bpi} \right\},
\end{equation}
and the infimum is attained uniquely by $\bp = \bp^{\star}$.  From this and \eqref{eqn:free.energy.identity}, we immediately obtain a variational representation of the excess growth rate which will be further explored in Section \ref{sec:variational.interpretation}. 

\begin{proposition}[Variational representation] \label{prop:egr.variational}
For $\bpi \in \Delta_n$ and $\br \in \R^n$, we have 
\begin{equation}\label{eqn:var.rep.egr}
\gamma(\bpi, \br) 
=\sup_{\bp\in\Delta_n}\Bigl\{\langle \bp-\bpi , \br \rangle-H \divg{\bp}{\bpi}\Bigr\}.
\end{equation}
Moreover, the unique maximizer of \eqref{eqn:var.rep.egr} is $\bp^\star = \bpi \oplus_{\bpi} \cC[e^{\br}]$.
\end{proposition}
\begin{proof}
We omit the proof as this result is classical. We only note that the support of the optimal $\bp$ must be contained in that of $\bpi$ as otherwise $H\divg{\bp}{\bpi} = \infty$. Also, the optimizer is unique since $H\divg{\cdot}{\bpi}$ is strictly convex on the set $\{ \bp  \in \Delta_n: \supp(\bp) \subset \supp(\bpi)\}$.
\end{proof}

\subsection{Information-theoretic interpretation} \label{sec:Campbell}
We show that the excess growth rate can be expressed in terms of L.~Campbell's measure of average code length \cite{campbell1965coding}. Mathematically, the relation is essentially the same as the one in \eqref{eqn:free.energy.identity}.

Fix $n \geq 1$. Let $\bpi = (\pi_1, \ldots, \pi_n) \in \Delta_n^{\circ}$ be a probability distribution and $\boldsymbol{\ell} = (\ell_1, \ldots, \ell_n) \in \mathbb{Z}_{> 0}^n$ be a set of codeword lengths over an alphabet $\mathcal{X}$ of size $D \geq 2$. 

\begin{definition}[Campbell's measure of expected code length]
Consider the distribution $\bpi$ and the vector $\boldsymbol{\ell}$ of codeword lengths as described above. For $\rho > 0$, we define
\begin{equation} \label{eqn:Campbell}
L_\rho(\bpi, \boldsymbol{\ell}) \coloneqq \frac{1}{\rho} \log_D \left( \sum_{i=1}^n \pi_i D^{\rho \ell_i} \right).
\end{equation}
\end{definition}

The idea is to consider a cost which is exponential in the length of the codeword. Campbell's measure is obtained by normalizing the expected value of $D^{\rho \ell}$ by a logarithmic transformation. This can be contrasted with Shannon's expected code length $S(\bpi, \boldsymbol{\ell}) \coloneqq \sum_{i=1}^n \pi_i \ell_i$, which is recovered in the limit $\lim_{\rho \rightarrow 0} L_{\rho}(\bpi, \boldsymbol{\ell}) = S(\bpi, \boldsymbol{\ell})$. At the other extreme, we have $\lim_{\rho \rightarrow \infty} L_{\rho}(\bpi, \boldsymbol{\ell}) = \max_{1 \leq i \leq n} \ell_i$. In \cite{campbell1965coding}, Campbell established source coding theorems under which the asymptotic optimal value of $L_\rho(\bpi, \boldsymbol{\ell})$---for long sequences of input symbols---is the R\'{e}nyi entropy. Also see \cite{bercher2009source} and the references therein for other extensions and applications.

We observe that the difference $L_{\rho}(\bpi, \boldsymbol{\ell}) - S(\bpi, \boldsymbol{\ell})$ between Campbell's and Shannon's expected lengths is, up to a multiplicative constant, an excess growth rate.

\begin{proposition}[Excess growth rate in Campbell's measure] \label{prop:Campbell}
Let $n \geq 1$, $\bpi \in \Delta_n^{\circ}$, $\boldsymbol{\ell} \in \mathbb{Z}_{>0}^n$ and $\rho > 0$. Define $\bR = (R_1, \ldots, R_n)$, where $R_i = D^{\rho \ell_i}$, is the vector of exponentiated code lengths. Then, we have
\begin{equation} \label{eqn:Campbell.identity}
L_{\rho}(\bpi, \boldsymbol{\ell}) - S(\bpi, \boldsymbol{\ell}) = \frac{1}{\rho \log D} \Gamma(\bpi, \bR).
\end{equation}
\end{proposition}
\begin{proof} 
We have
\[
L_{\rho}(\bpi, \boldsymbol{\ell}) = \frac{1}{\rho} \log_D \left(\sum_{i = 1}^n \pi_i R_i\right) = \frac{1}{\rho \log D} \log \left( \sum_{i  = 1}^n \pi_i R_i \right).
\]
On the other hand, we have
\[
S(\bpi, \boldsymbol{\ell}) = \sum_{i = 1}^n \pi_i \ell_i = \frac{1}{\rho} \sum_{i = 1}^n \pi_i \log_D R_i = \frac{1}{\rho \log D} \sum_{i = 1}^n \pi_i \log R_i.
\]
We obtain \eqref{eqn:Campbell.identity} by taking the difference. 
\end{proof}

\subsection{Probabilistic interpretations} \label{sec:probabilistic}
We prove two results that provide probabilistic interpretations of the excess growth rate in terms of the {\it scaled Dirichlet distribution}. We fix $n \geq 2$.

To motivate the first result, recall that the relative entropy arises in {\it Sanov's theorem} in the theory of large deviations. Let $\bp \in \Delta_n^{\circ}$. For $N \geq 1$, let 
\[
\bQ_N \sim \mathrm{Multinomial}(N, \bp)
\]
be an $n$-dimensional multinomial random vector. Let $\mu_{\bp, N}$ be the law of $\frac{1}{N} \bQ_N$. In other words, $\mu_{\bp, N}$ represents the law of the empirical distribution of $N$ independent samples from the categorical distribution on the state space $[n]$ with weights $\bp$. Then, it can be shown that the family $(\mu_{\bp, N})_{N \geq 1}$ satisfies the {\it large deviation principle} (LDP) with rate $N$ and rate function 
\begin{equation} \label{eqn:LDP.relative.entropy}
I(\bq) = H\divg{\bq}{\bp}, \quad \bq \in \Delta_n^{\circ}.
\end{equation}
In particular, we have
\begin{equation} \label{eqn:LDP}
\lim_{N \rightarrow \infty} \frac{1}{N} \log \mu_{\bp, N}(S) = -\inf_{\bq \in S} H\divg{\bq}{\bp}
\end{equation}
for sufficiently regular Borel subsets $S$ of $\Delta_n^{\circ}$. In Theorem \ref{thm:LDP}(ii) below, we recall the definition of LDP in the context of the excess growth rate. We refer the reader to \cite{D09} for a comprehensive treatment of large deviation theory, and \cite[Chapter 11]{CT06} for an introduction with a focus on information-theoretic concepts.

Analogously, we show that the excess growth rate arises in the large deviation principle of another stochastic model involving the scaled Dirichlet distribution \cite{MME21, MMPE11}. This extends the formulation in \cite[Section 3.1]{pal2020multiplicative} and \cite[Example III.18]{wong2022tsallis} which is restricted to the case $\bpi = \barE = (1/n, \ldots, 1/n)$.  For $\alpha, \beta \in (0, \infty)$, we let $\mathrm{Gamma}(\alpha, \beta)$ be the gamma distribution on $\R_+$ with shape parameter $\alpha$ and rate parameter $\beta$.

\begin{definition}[Scaled Dirichlet distribution] \label{def:scaled.Dirichlet}
The scaled Dirichlet distribution $\mathcal{SD}(\boldsymbol{\alpha}, \boldsymbol{\beta})$, with parameters $\boldsymbol{\alpha} = (\alpha_1, \ldots, \alpha_n)$ and $ \boldsymbol{\beta} = (\beta_1, \ldots, \beta_n)$ in $(0, \infty)^n$, is the distribution of the $\Delta_n^{\circ}$-valued random vector
\begin{equation} \label{eqn:scaled.Dirichlet}
\bY = \mathcal{C}[\bX], \quad \bX = (X_1, \ldots, X_n),
\end{equation}
where $X_1, \ldots, X_n$ are jointly independent with $X_i \sim \mathrm{Gamma}(\alpha_i, \beta_i)$.
\end{definition}

When $\boldsymbol{\beta} = (\beta, \ldots, \beta)$ is a constant vector, then $\mathcal{SD}(\boldsymbol{\alpha}, \boldsymbol{\beta})$ reduces to the Dirichlet distribution $\mathcal{D}(\boldsymbol{\alpha})$ with parameter $\boldsymbol{\alpha}$. The scaled Dirichlet distribution can be traced to the works of Savage and Dickey \cite{D68} in the 1960s. It was studied in \cite{MME21, MMPE11} as a more flexible version of the Dirichlet distribution to model simplex-valued (or {\it compositional}) data.

The following lemma shows that the scaled Dirichlet distribution can be expressed in terms of the usual Dirichlet distribution and the simplicial operations introduced in Section \ref{sec:egr.properties}.

\begin{lemma} \label{lem:scaled.Dirichlet.simplex}
Let $\boldsymbol{\alpha}, \boldsymbol{\beta} \in (0, \infty)^n$. Then $\mathcal{SD}(\boldsymbol{\alpha}, \boldsymbol{\beta})$ is equal to the distribution of
\begin{equation}
\bY = \cC[\boldsymbol{\beta}^{-1}] \oplus \bZ,
\end{equation}
where $\bZ \sim \mathcal{D}(\boldsymbol{\alpha})$. In particular, we have $\mathcal{SD}(\boldsymbol{\alpha}, \boldsymbol{\beta}) = \mathcal{SD}(\boldsymbol{\alpha}, c\boldsymbol{\beta})$ for any $c > 0$.
\end{lemma}
\begin{proof}
Suppose that $\bY \sim \mathcal{SD}(\boldsymbol{\alpha}, \boldsymbol{\beta})$ is expressed as \eqref{eqn:scaled.Dirichlet}. Recall that if $\tilde{X} \sim \mathrm{Gamma}(a, b)$, then $c\tilde{X} \sim \mathrm{Gamma}(a, b/c)$. So $\tilde{Z}_i \coloneqq \beta_i X_i \sim \mathrm{Gamma}(\alpha_i, 1)$, and $\bZ \coloneqq \cC[\tilde{\bZ}] \sim \mathcal{D}(\boldsymbol{\alpha})$. It follows that
\[
\bY = \mathcal{C}[\bX] = \mathcal{C}[\boldsymbol{\beta}^{-1} \tilde{\bZ}] = \cC[\boldsymbol{\beta}^{-1}] \oplus \bZ.
\]
\end{proof}

On $\Delta_n^{\circ}$, we take as reference measure the {\it Aitchison measure} $\lambda_n$ defined by
\[
\dd \lambda_n(\by) \coloneqq \frac{1}{\sqrt{n} \prod_{i = 1}^n y_i} \dd y_1 \cdots \dd y_{n-1}, \quad \by \in \Delta_n^{\circ}.
\]
More precisely, we regard the Lebesgue measure $\dd y_1 \cdots \dd y_{n-1}$ on
\[
\{(y_1, \ldots, y_{n-1}) \in (0, 1)^{n-1} : y_1 + \cdots + y_{n-1} < 1\}
\]
as a measure on $\Delta_n^{\circ}$ through the measurable bijection
\[
(y_1, \ldots, y_{n-1}) \leftrightarrow \left(y_1, \ldots, y_{n-1}, 1 - \sum_{i = 1}^{n-1} y_i\right).
\]
By \cite[(11)]{MME21}, the density of $\mathcal{SD}(\boldsymbol{\alpha}, \boldsymbol{\beta})$ with respect to the Aitchison measure $\lambda_n$ is given by
\begin{equation} \label{eqn:scaled.Dirichlet.density}
\frac{\Gamma\left( \sum_{i = 1}^n \alpha_i \right) \sqrt{n} }{\prod_{i = 1}^n \Gamma(\alpha_i)} \frac{\prod_{i = 1}^n (\beta_i y_i)^{\alpha_i}}{\left( \sum_{i = 1}^n \beta_i y_i \right)^{\sum_{i = 1}^n \alpha_i}}, \quad \by \in \Delta_n^{\circ},
\end{equation}
where $\Gamma(\cdot)$ is the gamma function (not to be confused with the excess growth rate).

\begin{remark} \label{rmk:Aitchison}
The Aitchison measure is the Haar measure on $(\Delta_n^{\circ}, \oplus)$ (as a topological commutative group) which is unique up to a multiplicative constant.
\end{remark}

To formulate an LDP involving the excess growth rate, we parameterize $\boldsymbol{\alpha}$ and $\boldsymbol{\beta}$ as follows. Let $\bpi \in \Delta_n^{\circ}$ and let $\sigma > 0$ be a noise parameter whose role is analogous to that of $1/N$ in \eqref{eqn:LDP}. For $\bx \in (0, \infty)^n$, we define 
\begin{equation} \label{eqn:mu.pi.x.sigma}
\mu_{\bpi, \bx, \sigma} \coloneqq \mathcal{SD}\left(\boldsymbol{\alpha} = \frac{1}{\sigma} \bpi, \boldsymbol{\beta} = \bpi \bx^{-1}\right).
\end{equation}
By Lemma \ref{lem:scaled.Dirichlet.simplex}, $\mu_{\bpi, \bx, \sigma}$ is the distribution of the random vector
\begin{equation} \label{eqn:Dirichlet.perturbation}
\bY = (\cC[\bx] \ominus \bpi) \oplus \bZ, \quad \bZ \sim \mathcal{D}\left(\frac{1}{\sigma} \bpi \right).
\end{equation}
Since $\bY$ depends on $\bx$ only via $\cC[\bx]$, we may assume without loss of generality that $\bx \in \Delta_n^{\circ}$. We obtain the density of $\mu_{\bpi, \bx, \sigma}$ by plugging \eqref{eqn:mu.pi.x.sigma} into \eqref{eqn:scaled.Dirichlet.density}.

\begin{lemma}[Density of $\mu_{\bpi, \bx, \sigma}$] \label{lem:Dirichlet.density}
For $\bpi, \bx \in \Delta_n^{\circ}$ and $\sigma > 0$, the density of $\mu_{\bpi, \bx, \sigma}$ with respect to $\lambda_n$ is given by
\begin{equation} \label{eqn:Y.density}
f(\by \mid \bpi, \bx, \sigma) \coloneqq \frac{\Gamma(1/\sigma) \sqrt{n}}{\prod_{i = 1}^n \Gamma(\pi_i/\sigma)} e^{\frac{-1}{\sigma} H(\bpi)} e^{ \frac{-1}{\sigma} \Gamma_{\bpi}\divg{\by}{\bx}}, \quad \by \in \Delta_n^{\circ},
\end{equation}
where $H(\bpi) = -\sum_{i = 1}^n \pi_i \log \pi_i$ is the Shannon entropy of $\bpi$.
\end{lemma}

\begin{remark}[The case of equal weights] \label{rmk:equal.weights}
If $\bpi = \barE$ is the barycenter of $\Delta_n^{\circ}$, then in \eqref{eqn:Dirichlet.perturbation} (with $\bx \in \Delta_n^{\circ}$) we have $\bY = \bx \oplus \bZ$, where $\bZ \sim \mathcal{D}(1/(n\sigma), \ldots, 1/(n\sigma))$. This recovers the {\it Dirichlet perturbation model} studied in \cite{pal2020multiplicative, TWYZ25, wong2022tsallis}.
\end{remark}

\begin{theorem}[Excess growth rate as rate function] \label{thm:LDP}
Let $n \geq 2$ and $\bpi \in \Delta_n^{\circ}$.
\begin{itemize}
\item[(i)] We have
\begin{equation} \label{eqn:density.LDP}
\lim_{\sigma \downarrow 0} \sup_{\bx, \by \in \Delta_n^{\circ}} \left| -\sigma \log f(\by \mid \bpi, \bx, \sigma) - \Gamma_{\bpi}\divg{\by}{\bx} \right| = 0.
\end{equation}
\item[(ii)] For $\bx \in \Delta_n^{\circ}$, the family $\left( \mu_{\bpi, \bx, \sigma} \right)_{\sigma > 0}$ of probability distributions on $\Delta_n^{\circ}$ satisfies the large deviation principle with rate $1/\sigma$ and rate function $I(\by) = \Gamma_{\bpi}\divg{\by}{\bx}$. By definition, this means that for every closed subset $F$ and every open subset $G$ of $\Delta_n^{\circ}$, we have 
\begin{equation} \label{eqn:LDP.def.1}
\begin{split}
\limsup_{\sigma \downarrow 0} \sigma \log \mu_{\bpi, \bx, \sigma}(F) &\leq -\inf_{\by \in F} \Gamma_{\bpi}\divg{\by}{\bx}, \\
\liminf_{\sigma \downarrow 0} \sigma \log \mu_{\bpi, \bx, \sigma}(G) &\geq -\inf_{\by \in G} \Gamma_{\bpi}\divg{\by}{\bx}.
\end{split}
\end{equation}
\end{itemize}
\end{theorem}
\begin{proof}
(i) From \eqref{eqn:Y.density}, we have
\begin{equation} \label{eqn:log.difference}
\begin{split}
&-\sigma \log f(\by \mid \bpi, \bx, \sigma) - \Gamma_{\bpi}(\by \mid \bx) \\
&= -\sigma \log \Gamma\left(\frac{1}{\sigma}\right) - \sigma \log \sqrt{n} + \sigma \sum_{i = 1}^n \log \Gamma\left( \frac{\pi_i}{\sigma} \right) + H(\bpi),
\end{split}
\end{equation}
which is independent of $\bx$ and $\by$ (this gives the $\sup$ in \eqref{eqn:density.LDP}).

By Stirling's approximation
\[
\log \Gamma(z) = z \log z - z + O(\log z), \quad \text{as } z \rightarrow \infty,
\]
the last expression in \eqref{eqn:log.difference} is equal to
\begin{equation*}
\begin{split}
&- \log \frac{1}{\sigma} + 1 - \sigma \log \sqrt{n} + \sum_{i = 1}^n \left( \pi_i \log \frac{\pi_i}{\sigma} - \pi_i \right) + H(\bpi) + o(1) \\
&= -H(\bpi) + H(\bpi) - \sigma \log \sqrt{n} + o(1) = o(1) \rightarrow 0, \quad \text{as } \sigma \downarrow 0.
\end{split}
\end{equation*}
This gives the desired result \eqref{eqn:density.LDP}. In particular, from the computation above we have the limit
\[
\lim_{\sigma \downarrow 0} \sigma \log \left( \frac{\Gamma\left(\frac{1}{\sigma}\right)}{\prod_{i = 1}^n \Gamma\left( \frac{\pi_i}{\sigma} \right)} \right) = H(\bpi).
\]

(ii) This is an immediate consequence of the uniform limit in (i). Since the argument is standard in the theory of large deviations and is not needed in the rest of the paper, we omit the details.
\end{proof}

Our second result expresses the excess growth rate as a {\it R\'{e}nyi divergence} between members of the family $(\mu_{\bpi, \bx, \sigma})_{\bx \in \Delta_n^{\circ}}$. Here, $\bpi \in \Delta_n^{\circ}$ and $\sigma > 0$ are fixed. We may regard $\bx$ as a location parameter. To give a classical analogue, consider the squared {\it Mahalanobis distance} \cite{M2018} on $\R^n$ defined by
\begin{equation} \label{eqn:Mahalanobis}
d_{\mathrm{M}}^2(\bx, \by) \coloneqq (\bx - \by)^{\top} \Sigma^{-1} (\bx - \by), 
\end{equation}
where $\bx, \by \in \R^n$ are considered column vectors and $\Sigma \in \R^{n \times n}$ is a given strictly positive definite matrix. It is well known that $d_{\mathrm{M}}^2$ expresses, up to a constant, the relative entropy between members of the {\it normal location family} $\{ \cN(\bx, \Sigma) \}_{\bx \in \R^n}$ (note that the covariance matrix $\Sigma$ is kept fixed):
\begin{equation} \label{eqn:Mahalanobis.normal.location}
H\divg{ \cN(\bx, \Sigma) }{ \cN(\by, \Sigma) } = \frac{1}{2} d_{\mathrm{M}}^2(\bx, \by), \quad \bx, \by \in \R^n.
\end{equation}
An axiomatic characterization of the Mahalanobis distance is given in Theorem \ref{thm:Mahalanobis} below.

\begin{remark} \label{rmk:exp.family}
The identity \eqref{eqn:Mahalanobis.normal.location} is a special case of the general result that the relative entropy between members of an exponential family of probability distributions can be expressed as a Bregman divergence \cite{A16}.
\end{remark}

Let $\dd \mu_1 = f_1 \dd \nu$ and $\dd \mu_2 = f_2 \dd \nu$ be probability measures on a (measurable) state space $\cX$, where $\nu$ is a common dominating measure. For $\alpha \in (0, 1) \cup (1, \infty)$, the {\it R\'{e}nyi divergence} of order $\alpha$ is defined by
\begin{equation} \label{eqn:Renyi.divergence.densities}
H_{\alpha}\divg{\mu_1}{\mu_2} \coloneqq \frac{1}{\alpha - 1} \log \left( \int_{\cX} f_1^{\alpha} f_2^{1 - \alpha} \dd \nu \right).
\end{equation}
where we read $f_1^{\alpha} f_2^{1 - \alpha} = f_1^{\alpha}/f_2^{\alpha - 1}$ for $\alpha > 1$ and we use the conventions $0/0 = 0$ and $x/0 = \infty$ for $x > 0$; see \cite[Definition 2]{VH14}. We refer the reader to \cite{VH14} for properties of the R\'{e}nyi divergence. In particular, the relative entropy is recovered in the limit $\alpha \uparrow 1$.

\begin{theorem}[Excess growth rate as R\'{e}nyi divergence] \label{thm:Renyi}
For $n \geq 2$, $\bpi \in \Delta_n^{\circ}$, $\sigma > 0$ and $\bx, \by \in \Delta_n^{\circ}$, we have
\begin{equation} \label{eqn:Dirichlet.Renyi}
H_{1 + \sigma} \divg{\mu_{\bpi, \by, \sigma} }{\mu_{\bpi, \bx, \sigma}} = \frac{1}{\sigma} \Gamma_{\bpi}\divg{\by }{\bx}.
\end{equation}
\end{theorem}
\begin{proof}
Fix $\bpi \in \Delta_n^{\circ}$ and $\sigma > 0$. By Lemma \ref{lem:Dirichlet.density}, we have
\[
f(\bz \mid \bpi, \bx, \sigma) = C e^{\frac{-1}{\sigma} \Gamma\divg{\bz}{\bx}},
\]
where $\Gamma\divg{\bz}{\bx} \coloneqq \Gamma_{\bpi}\divg{\bz}{\bx}$ and
\[
C = C_{\bpi, \sigma} \coloneqq \frac{\Gamma(\frac{1}{\sigma}) \sqrt{n}}{\prod_{i = 1}^n \Gamma(\frac{\pi_i}{\sigma})} e^{\frac{-1}{\sigma} H(\bpi)}.
\]
Let $\bx, \by \in \Delta_n^{\circ}$ be given, and consider 
\begin{equation*}
\begin{split}
&\sigma H_{1 + \sigma}\divg{\mu_{\bpi, \by, \sigma}}{ \mu_{\bpi, \bx, \sigma}} \\
&= 
\sigma \frac{1}{(1 + \sigma) - 1} \log \left( \int_{\Delta_n^{\circ}} f(\bz \mid \bpi, \by, \sigma)^{1 + \sigma} f(\bz \mid \bpi, \bx, \sigma)^{-\sigma}  \dd \lambda_n(\bz) \right) \\
&= \log \left(  C\int_{\Delta_n^{\circ}} e^{-\frac{1 + \sigma}{\sigma} \Gamma\divg{\bz}{\by} + \Gamma\divg{\bz}{\bx} } \dd \lambda_n(\bz) \right).
\end{split}
\end{equation*}

To evaluate the integral, consider the identity
\begin{equation} \label{eqn:integral.identity}
\int_{\Delta_n^{\circ}}  C e^{\frac{-1}{\sigma} \Gamma\divg{\bz}{\by}} \dd \lambda_n(\bz) \equiv 1, \quad \by \in \Delta_n^{\circ}.
\end{equation}
In fact, by num\'{e}raire invariance (Proposition \ref{prop:numeraire.invariance}), the identity holds for $\by \in (0, \infty)^n$. Let $\nabla_{\bv}$ be the directional derivative, with respect to $\by$, in the direction 
\[
\bv = \left( \frac{y_i^2}{x_i} \right)_{1 \leq i \leq n} \in \R^n.
\]
Differentiating under the integral sign in \eqref{eqn:integral.identity},\footnote{This can be justified using Remark \ref{rmk:Aitchison} and standard estimates.} we have
\[
C \int_{\Delta_n^{\circ}} e^{\frac{-1}{\sigma} \Gamma\divg{\bz}{\by}}  \nabla_{\bv} \Gamma\divg{\bz}{\by} \dd \lambda_n(\bz) = 0.
\] 
Noting that 
\[
\nabla_{\bv} \Gamma\divg{\bz}{\by} = \sum_i \frac{-\pi_i \frac{z_i}{y_i^2}}{\sum_j \pi_j \frac{z_j}{y_j}} \frac{y_i^2}{x_i} + \sum_i \frac{\pi_i}{y_i} \frac{y_i^2}{x_i} = - \frac{\sum_i \pi_i \frac{z_i}{x_i}}{\sum_i \pi_i \frac{z_i}{y_i}} + \sum_i \pi_i \frac{y_i}{x_i},
\]
we have
\[
\sum_i \pi_i \frac{y_i}{x_i} = C \int_{\Delta_n^{\circ}} e^{\frac{-1}{\sigma} \Gamma\divg{\bz}{\by}}  \frac{\sum_i \pi_i \frac{z_i}{x_i}}{\sum_i \pi_i \frac{z_i}{y_i}} \dd \lambda_n(\bz).
\]
Observe that we may rearrange the above as
\begin{equation*}
\begin{split}
&e^{\Gamma\divg{\by}{\bx} + \sum_i \pi_i \log \frac{y_i}{x_i}} = C \int_{\Delta_n^{\circ}} \frac{e^{\frac{-1}{\sigma} \Gamma\divg{\bz}{\by}} e^{\Gamma \divg{\bz}{\bx} + \sum_i \pi_i \log \frac{z_i}{x_i}}}{e^{\Gamma\divg{\bz}{\by} + \sum_i \pi_i \log \frac{z_i}{y_i}}} \dd \lambda_n(\bz) \\
&\Rightarrow e^{\Gamma\divg{\by}{\bx}} = C \int_{\Delta_n^{\circ}} e^{-\frac{1+\sigma}{\sigma} \Gamma\divg{\bz}{\by}} e^{\Gamma\divg{\bz}{\bx}} \dd \lambda_n(\bz).
\end{split}
\end{equation*}
We obtain \eqref{eqn:Dirichlet.Renyi} by taking logarithms on both sides.
\end{proof}

\begin{remark} \label{rmk:Renyi.discussion}
The identity \eqref{eqn:Dirichlet.Renyi} can also be derived by showing that the family $\{\mu_{\bpi, \bx, \sigma}\}_{\bx \in \Delta_n^{\circ}}$ can be reparameterized as a \emph{$\lambda$-exponential family} in the sense of \cite[Definition III.1]{wong2022tsallis}, and using the fact that the logarithmic divergence of a suitable potential function is the R\'{e}nyi divergence \cite[Theorem III.14]{wong2022tsallis}. This result extends the relation in Remark \ref{rmk:exp.family} beyond the standard exponential family. In fact, the direct proof given above, which is shorter but may appear to be tricky, is motivated by this general theory. %

\end{remark}

\section{Characterization theorems} \label{sec:characterization}
In this section, we state and prove three characterization theorems for the excess growth rate that highlight different aspects of this quantity. 

\subsection{Via relative entropy} \label{sec:relative entropy}
Our first characterization theorem shows that the properties of the excess growth rate discussed in Section \ref{sec:egr.properties}, together with Lebesgue measurability, uniquely characterize it (as a family $(\Gamma : \cD_n \rightarrow \R_+)_{n \geq 1}$) up to a proportional constant. Our proof makes use of an algebraic relation between the excess growth rate and the relative entropy (Lemma \ref{lem:link.rel.entr}), as well as an axiomatic characterization of the latter. 

For easy reference, we gather here the relevant properties which are stated in terms of a family of functions $(G : \cD_n \rightarrow \R)_{n \geq 1}$.
\begin{enumerate}[label=\textup{(A\arabic*)},leftmargin=*]
\item\label{A1} $G(\bpi,\bR)$ is (jointly) Lebesgue measurable.\footnote{Lebesgue measurability is understood with respect to the 
Lebesgue measure on the affine hull of $\cD_n \subset \mathbb{R}^n \times \mathbb{R}^n$. Joint Lebesgue measurability of a function on $\cD_n$ (resp.~$\cA_n$)
is equivalent to measurability of its restriction to $\Delta_n^\circ\times(0,\infty)^n$ (resp.~$\Delta_n^\circ\times\Delta_n^\circ$) by the
completeness of Lebesgue measure
\cite[Proposition~2.11(a)]{folland1999real} since $\cD_n\setminus(\Delta_n^\circ\times(0,\infty)^n)$ (resp.~$\cA_n\setminus (\Delta_n^\circ\times\Delta_n^\circ)$) is Lebesgue-null.}
\item\label{A3} 
$G(\bpi\sigma,\bR\sigma)=G(\bpi,\bR)$ for every permutation $\sigma$.
\item\label{A6} $G(\bpi,\bR)=G(\bpi,\bR')$ if $R_i=R_i'$ for $i\in \supp(\bpi)$.
\item\label{A4} 
$G(\bpi,\bR)=0$ if $\bR$ is constant on $\supp(\bpi)$. 
\item\label{A5}
For $(\bpi,\ba)\in \cD_n$, $k_i\ge1$, $(\bp^i,\bR^i)\in \cD_{k_i}$, $\bp=(\bp^1,\dots,\bp^n)$ and $\bR=(\bR^1,\dots,\bR^n)$, the general chain rule holds:
\begin{equation} \label{eqn:A5.chain.rule}
G(\bpi \circ \bp, \ba \circ \bR) = G(\bpi, \ba \llangle \bp, \bR \rrangle) + \sum_{i = 1}^n \pi_i G (\bp^i, \bR^i).
\end{equation}
\end{enumerate}

Assumption \ref{A1} asks for minimal regularity to rule out pathological functions. \ref{A3} is the permutation invariance of Proposition \ref{prop:permutation.invariance}. \ref{A6} and \ref{A4} correspond to Proposition \ref{prop:support}. \ref{A5} is the general chain rule in Proposition \ref{prop:chain.rule.2}. The following lemma shows that given \ref{A6}--\ref{A4}, we may equivalently replace \ref{A5} by num\'{e}raire invariance and the first chain rule in Proposition \ref{prop:chain.rule.1}.

\begin{lemma} \label{lem:A6.A7}
Let $(G: \cD_n \rightarrow \R)_{n \geq 1}$  be a family of functions satisfying \ref{A6}--\ref{A4}. Then it satisfies \ref{A5} if and only if it satisfies the following two conditions
\begin{enumerate}[label=\textup{(A\arabic*)},leftmargin=*,start =6]
\item\label{A6.} $G(\bpi, a \bR) = G(\bpi, \bR)$ for $a > 0$.
\item\label{A7.} The first chain rule holds:
\[
G(\bpi \circ \bp, \bR) = G(\bpi, \llangle \bp, \bR \rrangle) + \sum_{i = 1}^n \pi_i G (\bp^i, \bR^i).
\]
\end{enumerate}
\end{lemma}
\begin{proof}
Let $G$ satisfy \ref{A6}--\ref{A4}. From the proof of Proposition \ref{prop:chain.rule.2}, we have that if $G$ satisfies \ref{A6.}--\ref{A7.}, then it satisfies \ref{A5}.

Conversely, suppose that $G$ satisfies \ref{A5}. Letting $a_1 = \cdots = a_n = 1$ in \eqref{eqn:A5.chain.rule} recovers \ref{A7.}. To show \ref{A6.}, apply \ref{A5} with a one-dimensional outer block: take $n=1$, $k_1=m\geq1$, $\bpi=(1)\in\Delta_1$, $\ba=(\alpha)$ for $\alpha>0$, $\bp=(\bp^1)$ for $\bp^1\in\Delta_m$, and $\bR=(\bR^1)$ for $\bR^1\in \cD_m(\bp^1 \mid \cdot)$. We get
\begin{align*}
G_{m}( \bp^1, \alpha\bR^1) &=G_{m}( \bpi\circ \bp, \ba\circ \bR) \\
&= G_{1}( (1), \alpha \langle \bp^1,\bR^1\rangle )+ 1 \cdot G_{m}( \bp^1, \bR^1) \\
&=G_{m}( \bp^1, \bR^1),
\end{align*}
since \ref{A4} ensures $G_{1}( (1), \alpha \langle \bp^1,\bR^1\rangle )=0$.
\end{proof}

\begin{theorem}[Characterization I] \label{thm:characterization.rel.entr}
Let $(G : \cD_n \rightarrow \R)_{n \geq 1}$ be a family of functions. The following are equivalent:
\begin{itemize}
    \item[(i)] The family satisfies \ref{A1}--\ref{A4} and either \ref{A5} or \ref{A6.}--\ref{A7.}.\footnote{In the proof we will work with \ref{A1}--\ref{A5}.}
    \item[(ii)] $G = c\Gamma$ for some $c \in \R$.
\end{itemize}
\end{theorem}

We have seen in Section \ref{sec:egr.properties} that the family $(c\Gamma)_{n \geq 1}$ satisfies the properties in (i); nothing is changed by multiplying $\Gamma$ by a constant. Given \ref{A1}--\ref{A6}, both \ref{A4}  and \ref{A5} are necessary as shown by the following examples.

\begin{example} { \ } 
\begin{itemize}
\item[(i)] The functional $\check{\Gamma}(\bpi, \bR) = \mathrm{Var}_{\bpi}(\br)$ in Example \ref{eg:chain.rule} satisfies \ref{A1}--\ref{A4} but not \ref{A5}.
\item[(ii)] Consider $\hat{\Gamma}(\bpi, \bR) \coloneqq c \Gamma(\bpi, \bR) + c H(\bpi)$, where $H$ is the Shannon entropy and $c \in \mathbb{R} \setminus \{0\}$. Using the chain rule of the Shannon entropy,
\[
H(\bpi \circ {\bf p}) = H(\bpi) + \sum_{i \in \supp(\bpi)} \pi_i H({\bf p}^i),
\]
we see that $\hat{\Gamma}$ satisfies \ref{A1}--\ref{A6} and \ref{A5} but not \ref{A4}.
\end{itemize}
\end{example}

The proof of the implication (i) $\Rightarrow$ (ii) is quite delicate, and we will adopt the following strategy:
\begin{enumerate}
\item[1.] Use num\'{e}raire invariance (from \ref{A5} and \ref{A4}) to reduce $(G: \cD_n \rightarrow \R)_{n \geq 1}$ to an equivalent family $(\sG: \cA_n \rightarrow \R)_{n \geq 1}$ of functions that take simplex-valued arguments.
\item[2.] Using algebraic relations between the excess growth rate and the relative entropy (Lemma \ref{lem:link.rel.entr}), as well as an axiomatic characterization of the latter (Proposition \ref{prop:Interior.Char.Rel.Entropy}), derive the characterization when the domain of the reduced functions from Step 1 is restricted to $\Delta_n^\circ \times \Delta_n^\circ$.
\item[3.] Extend the characterization to the full domain.
\end{enumerate}

{\it Step 1.} Recall the set $\cA_n$ defined by \eqref{eqn:An}. Clearly, if $(G: \cD_n \rightarrow \R)_{n \geq 1}$ is num\'{e}raire invariant, it is characterized by its restriction to $\cA_n$ (see \eqref{eqn:egr.renormalize}). We cast Theorem \ref{thm:characterization.rel.entr} in the following equivalent form.

\begin{theorem} \label{thm:characterization.rel.entr.reduced}
Let $(\sG:\cA_n\to\R)_{n \geq 1}$ be a family of functions. The following are equivalent:
\begin{itemize}
\item[(i)] The family satisfies \ref{B1}--\ref{B4} where:
\begin{enumerate}[label=\textup{(B\arabic*)},leftmargin=*]
\item\label{B1} $\sG(\bpi,\br)$ is (jointly) Lebesgue measurable.
\item\label{B2}
$\sG(\bpi\sigma,\br\sigma)=\sG(\bpi,\br)$ for every permutation $\sigma$.
\item\label{B5} $\sG(\bpi,\br)=\sG(\bpi,\cC_{\bpi}[\br])$.
\item\label{B3}
$\sG(\bpi,\br)=0$ if $\cC_{\bpi}[\br]=\overline{\be}_{\bpi}$, where $\overline{\be}_{\bpi}$ is the {\it barycenter} on the support of $\bpi \in \Delta_n$ defined by $\overline{\be}_{\bpi} \coloneqq \cC_{\bpi}[\barE_n]$.%
\item\label{B4}
For $(\bpi,\ba)\in \cA_n$, $k_i\ge1$, $(\bp^i,\br^i)\in \cA_{k_i}$, $\bp=(\bp^1,\dots,\bp^n)$ and $\br = (\br^1,\dots,\br^n)$, we have
\[
\sG(\bpi\circ \bp,\ba\circ\br)
=
\sG(\bpi, \cC_{\bpi}[\ba \llangle \bp, \br \rrangle])+\sum_{i=1}^n \pi_i\,\sG(\bp^i,\br^i).
\]
\end{enumerate}
\item [(ii)] $\sG = c\Gamma$ for some $c \in \R$.
\end{itemize}
\end{theorem}

It is easy to see that Theorem \ref{thm:characterization.rel.entr} and Theorem \ref{thm:characterization.rel.entr.reduced} are equivalent. Given a family $(G : \cD_n \rightarrow \R)_{n \geq 1}$ that satisfies \ref{A1}--\ref{A5}, let $\sG$ be the restriction of $G$ to $\cA_n$ (for each $n \geq 1$). For this choice, it can be easily verified that \ref{B1}--\ref{B4} hold. Similarly, if $(\sG: \cA_n \rightarrow \R)_{n \geq 1}$ satisfies \ref{B1}--\ref{B4}, we can define $G(\bpi,\bR) = \sG(\bpi,\cC_{\bpi}[\bR])$ and recover \ref{A1}--\ref{A5}. The most subtle equivalence is \ref{A1}$\iff$\ref{B1}. When $n=1$ it is immediate. For $n\geq 2$, the map $(\bpi,\bR)\mapsto \left(\bpi,\cC[\bR],\sum_i R_i\right)$ is a smooth diffeomorphism from $\Delta_n^\circ\times(0,\infty)^n$ onto $\Delta_n^\circ\times\Delta_n^\circ\times(0,\infty)$, with inverse $(\bpi,\br,s)\mapsto(\bpi,s\br)$. On this domain, $G(\bpi,\bR)=\sG(\bpi,\cC[\bR])$ by numeraire invariance or the definition of the extension. Lemma \ref{lem:meas.chart}(a) (applied after identifying $\Delta_n^\circ$ with an open subset of $\mathbb{R}^{n-1}$) therefore shows that $G$ is jointly Lebesgue measurable on $\Delta_n^\circ\times(0,\infty)^n$ if and only if $\sG$ is jointly Lebesgue measurable on $\Delta_n^\circ\times\Delta_n^\circ$. These domains differ from $\cD_n$ and $\cA_n$, respectively, by Lebesgue null sets, which proves the equivalence.

\medskip

{\it Step 2.} We use the following link between the excess growth rate and the relative entropy. It is a slight extension of \cite[Lemma 2]{pal2020multiplicative}.

\begin{lemma}[Excess growth rate as relative entropy]\label{lem:link.rel.entr}
For $(\bpi, \br) \in \cA_n$ we have
\begin{equation} \label{eqn:egr.relative.entropy}
\Gamma(\bpi, \br) = H\divg{ \bpi}{ \bpi \oplus_{\bpi} \br} \quad \text{and} \quad  \Gamma(\bpi, \br\ominus_{\bpi} \bpi) = H\divg*{ \bpi }{ \mathcal{C}_{\bpi}[\br]}.
\end{equation}
\end{lemma}
\begin{proof}
Write $\bpi \oplus_{\bpi} \br = (\frac{1}{Z} \pi_i r_i)_{1 \leq i \leq n}$, where $Z = \sum_{j\in\supp(\bpi)}\pi_jr_j$ is a normalizing constant. We verify the first identity by a direct computation:
\begin{align*}
H\divg{\bpi}{\bpi \oplus_{\bpi} \br} &= \sum_{i\in\supp(\bpi)} \pi_i \log \left( \frac{ \pi_i }{ \frac{1}{Z} \pi_i r_i} \right) \\
&= \log Z \sum_{i\in\supp(\bpi)} \pi_i - \sum_{i\in\supp(\bpi)} \pi_i \log r_i \\
&= \log \left( \sum_{i\in\supp(\bpi)} \pi_i r_i \right) - \sum_{i\in\supp(\bpi)} \pi_i \log r_i\\
&= \Gamma(\bpi, \br).
\end{align*}
To prove the second identity, write $\tilde{Z}\coloneqq\sum_{i\in\supp(\bpi)}r_i/\pi_i$. We have
\begin{align*}
\Gamma(\bpi, \br \ominus_{\bpi} \bpi) &= \log \left( \sum_{i\in\supp(\bpi)} \pi_i\frac{r_i/\pi_i}{\tilde{Z}} \right) - \sum_{i\in\supp(\bpi)} \pi_i \log\left(\frac{r_i/\pi_i}{\tilde{Z}}\right)\\
&= \log \left( \sum_{i\in\supp(\bpi)} r_i \right) - \sum_{i\in\supp(\bpi)} \pi_i \log\left(\frac{r_i}{\pi_i}\right) \\
&= \sum_{i\in\supp(\bpi)} \pi_i \log\left(\frac{\pi_i}{r_i/\sum_{j\in\supp(\bpi)}r_j}\right) \\
&=H \divg{\bpi }{ \mathcal{C}_{\bpi}[\br]}. \qedhere 
\end{align*}
\end{proof}

Lemma \ref{lem:link.rel.entr} suggests that if a family $(\sG: \cA_n \rightarrow \R)_{n \geq 1}$ satisfies \ref{B1}--\ref{B4}, then $I\divg{\bp}{\bq} \coloneqq \sG(\bp, \bq \ominus_{\bp} \bp)$, defined for $(\bp, \bq) \in \cA_n$ and $n \geq 1$, is essentially a constant multiple of the relative entropy. To this end, we will make use of the following characterization of the relative entropy on the \emph{interior} of the simplex. It is a slight variant of existing characterizations of relative entropy (see~\cite[Section 3.5]{leinster2021entropy}). Since it differs slightly in its domain and aspects of its assumptions, we provide a proof and a technical discussion in Appendix \ref{app:proof.of.int.char.rel.entr}.

\begin{proposition}[Characterization of relative entropy on the open simplex]\label{prop:Interior.Char.Rel.Entropy}
Let $(I\divg{\cdot}{\cdot}:\Delta_n^\circ\times\Delta_n^\circ\to \R)_{n \geq 1}$ be a family of functions. The following are equivalent:
\begin{itemize}
\item[(i)] The family satisfies \ref{C1}--\ref{C4} where:
    \begin{enumerate}                   [label=\textup{(C\arabic*)},leftmargin=*]
    \item\label{C1} For almost every $\bp\in\Delta_n^\circ$, the map  $\bq\mapsto I\divg{\bp}{\bq}$ is Lebesgue measurable on $\Delta_n^\circ$, and for almost every $\bq\in\Delta_n^\circ$, the map $\bp\mapsto I\divg{\bp}{\bq}$ is Lebesgue measurable on $\Delta_n^\circ$.
    \item\label{C2}
    $I\divg{\bp\sigma}{\bq\sigma}=I\divg{\bp}{\bq}$ for every permutation $\sigma$.
    \item\label{C3}
    $I\divg{\bp}{\bp}=0$ for all $\bp\in\Delta_n^\circ$;
    \item\label{C4} 
    For $(\bp,\bq)\in \Delta_n^\circ\times \Delta_n^\circ$, and $(\bmu^i,\bnu^i)\in \Delta_{k_i}^\circ\times \Delta_{k_i}^\circ$ for $k_i\ge1$, $i=1,\dots,n$, the following chain rule holds:
    \begin{align} \label{eqn:relative.entropy.chain.rule}
    I\divg{\bp \circ \bmu}{\bq\circ\bnu}
    =
    I\divg{\bp}{\bq}+\sum_{i=1}^n p_iI \divg{\bmu^i}{\bnu^i}.
    \end{align}
    \end{enumerate}   
\item[(ii)] $I\divg{\cdot}{\cdot} = c H\divg{\cdot}{\cdot}$ for some $c \in \R$.
\end{itemize} 
\end{proposition}

We note that condition \ref{C1} is implied by both separate and joint Lebesgue measurability (see Lemma \ref{lem:meas.chart}(b) for the latter).

With this characterization in mind, we establish a link between \ref{B1}--\ref{B4} and \ref{C1}--\ref{C4} through the function $\sG(\bp, \bq\ominus_{\bp} \bp)$. For $(\bpi, \br) \in \cA_n$, we let
\[
m_{\bpi}(\br)\coloneqq\sum_{i\in\supp(\bpi)}r_i>0
\]
for the mass that $\br$ puts on the support of $\bpi$.

\begin{lemma}\label{lem:implied.conditions.G}
Suppose that $(\sG: \cA_n \to \mathbb{R})_{n \geq 1}$ satisfies \ref{B1}--\ref{B4}. Define $(I\divg{\cdot}{\cdot}: \cA_n \rightarrow \R)_{n \geq 1}$ by $I\divg{\bp}{\bq} = \sG(\bp, \bq\ominus_{\bp} \bp)$. Then the family $(I\divg{\cdot}{\cdot}: \cA_n \rightarrow \R)_{n \geq 1}$ satisfies \ref{C1}, satisfies \ref{C2}--\ref{C3} with $\Delta_n^{\circ} \times \Delta_n^{\circ}$ replaced by $\cA_n$, and satisfies the following version of the \emph{chain rule}:
\begin{enumerate}                   [label=\textup{(C\arabic*$'$)},leftmargin=*]
\setcounter{enumi}{3}
    \item\label{B4'} For $(\bp,\bq)\in \cA_n$, $k_i\ge1$, $(\bmu^i,\bnu^i)\in \cA_{k_i}$, $\bmu=(\bmu^1,\dots,\bmu^n)$, and $\bnu = (\bnu^1,\dots,\bnu^n)$, we have
    \begin{equation}\label{eq:modified.chain.rule.I}
    I\divg{\bp \circ \bmu}{ \bq\circ\bnu} =\sG(\bp,(\bq\ominus_{\bp} \bp) \oplus_{\bp}  \bh_{\bmu}(\bnu))+\sum_{i=1}^n p_iI\divg{\bmu^i}{\bnu^i},
    \end{equation}
    where
    \[
    \bh_{\bmu}(\bnu) \coloneqq \left(\frac{m_{\bmu^i}(\bnu^i)}{\sum_{j=1}^nm_{\bmu^j}(\bnu^j)}\right)_{1 \leq i \leq n}\in\Delta_n^\circ.
    \]
    In particular, when $\bmu$ and $\bnu$ are chosen so that $\supp(\bmu^i)=\supp(\bnu^i)$ for all $i=1,\dots,n$, then \eqref{eq:modified.chain.rule.I} reduces to \eqref{eqn:relative.entropy.chain.rule}.
\end{enumerate}
\end{lemma}

\begin{proof} We treat each property in turn.\footnote{Note that \ref{B5} is not used here but will be needed in Step 3 below.}
 
    \ref{C1} The case $n=1$ is immediate. For $n\geq 2$, the map $(\bp,\bq)\mapsto (\bp,\bq\ominus\bp)$ is a smooth diffeomorphism of $\Delta_n^\circ\times\Delta_n^\circ$ onto itself. Hence the joint measurability asserted in \ref{B1} taken together with Lemma \ref{lem:meas.chart}(a) gives joint Lebesgue measurability of $I\divg{\bp}{\bq}$ on $\Delta_n^\circ\times\Delta_n^\circ$ which implies \ref{C1}.
    
    \ref{C2}: It is easy to check that $(\bq\sigma\ominus_{\bp\sigma} \bp\sigma) = (\bq\ominus_{\bp} \bp)\sigma$ for any $(\bp, \bq) \in \cA_n$ and permutation $\sigma$. 
Using this with \ref{B2} we obtain
\[
I\divg{\bp\sigma}{ \bq\sigma}
=\sG\bigl(\bp\sigma,(\bq\sigma\ominus_{\bp\sigma} \bp\sigma)\bigr)
=\sG\bigl(\bp\sigma,(\bq\ominus_{\bp} \bp)\sigma\bigr)
=\sG\bigl(\bp,(\bq\ominus_{\bp} \bp)\bigr)
=I\divg{\bp}{\bq}.
\]

\ref{C3}: Observe that $\bp \ominus_{\bp} \bp = \overline{\be}_{\bp}$. Therefore, by \ref{B3} we have
\[
I\divg{\bp}{\bp} = \sG (\bp , \bp\ominus_{\bp} \bp) = \sG(\bp, \overline{\be}_{\bp}) = 0.
\]

\ref{B4'}: Consider
\[
I\divg{\bp \circ \bmu}{ \bq \circ \bnu} = \sG(\bp \circ \bmu, (\bq \circ \bnu)\ominus_{\bp \circ \bmu} (\bp \circ \bmu)).
\]
In order to invoke the chain rule \ref{B4}, we express $(\bq \circ \bnu)\ominus_{\bp \circ \bmu} (\bp \circ \bmu)$ as a composite distribution. Let $Z$ be the normalizing constant in the definition of $(\bq \circ \bnu)\ominus_{\bp \circ \bmu}  (\bp \circ \bmu)$:
\begin{align*}
Z =\sum_{(i, j) \in \supp(\bp \circ \bmu)} \frac{(\bq \circ \bnu)_{i, j}}{(\bp \circ \bmu)_{i, j}} 
=\sum_{i \in \supp(\bp)} \frac{q_i}{p_i} \sum_{j \in \supp(\bmu^i)} \frac{\nu^i_j}{ \mu^i_j} = \sum_{i \in \supp(\bp)} \frac{q_i}{p_i} Z_i,
\end{align*}
where $Z_i = \sum_{j \in \supp(\bmu^i)} \nu_j^i / \mu_j^i$ is strictly positive since $(\bmu^i,\bnu^i)\in \cA_{k_i}$.

Write
\begin{align*}
\left((\bq \circ \bnu)\ominus_{\bp \circ \bmu}  (\bp \circ \bmu)\right)_{i,j}
&= \frac{1}{Z} \frac{q_i\nu_j^i}{p_i \mu_j^i} \mathds{1}_{\supp(\bp)}(i) \mathds{1}_{\supp(\bmu^i)}(j) \\
&= \left( \frac{\frac{q_i}{p_i}Z_i}{Z} \mathds{1}_{\supp(\bp)}(i) \right) \cdot \left( \frac{\nu_j^i/\mu_j^i}{Z_i} \mathds{1}_{\supp(\bmu^i)}(j) \right) \\
&= (\brho \circ \bxi)_{i, j},
\end{align*}
where $\brho = (\rho_1, \ldots, \rho_n) \in \cA_n(\bp \mid\cdot)$ with
\[
\rho_i = \frac{\frac{q_i}{p_i}Z_i}{Z} \mathds{1}_{\supp(\bp)}(i), \quad i = 1, \ldots, n,
\]
and $\bxi = (\bxi^1, \ldots, \bxi^n)$ with
\[
\bxi^i = \bnu^i \ominus_{\bmu^i} \bmu^i \in \cA_{k_i}(\bmu^i \mid \cdot), \quad i = 1, \ldots, n.
\]

Therefore, we may apply \ref{B4} to obtain,
\begin{align*}
    I\divg{\bp \circ \bmu}{ \bq \circ \bnu} &= \sG(\bp \circ \bmu, (\bq \circ \bnu)\ominus_{\bp \circ \bmu}  (\bp \circ \bmu))\\
    &= \sG(\bp \circ \bmu, \brho \circ \bxi)\\
    &= \sG(\bp, \cC_{\bp}[\brho \llangle \bmu, \bxi \rrangle])  + \sum_{i=1}^n p_i \sG(\bmu^i,\bxi^i)\\
    &= \sG(\bp, \cC_{\bp}[\brho \llangle \bmu, \bxi \rrangle])  + \sum_{i=1}^n p_i I\divg{\bmu^i}{\bnu^i},
\end{align*}
where the last equality follows from the definitions of $\bxi^i$ and $I_{k_i}$:
\[
\sG (\bmu^i, \bxi^i) = \sG(\bmu^i, \bnu^i \ominus_{\bmu^i} \bmu^i) = I\divg{\bmu^i}{\bnu^i}.
\]

It remains to show that $\cC_{\bp}[\brho \llangle \bmu, \bxi \rrangle] = (\bq\ominus_{\bp} \bp) \oplus_{\bp} \bh_{\bmu}(\bnu)$. To simplify the notation, for non-zero $\bx, \by \in [0, \infty)^n$ we write $\bx \propto \by$ if $\by = c \bx$ for some $c > 0$. Clearly, $\bx, \by \in \Delta_n$ are equal if and only if $\bx \propto \by$. Since
\begin{equation*}
\begin{split}
\cC_{\bp}[\brho \llangle \bmu, \bxi \rrangle] &\propto \left( \rho_i \langle \bmu^i, \bxi^i \rangle \right)_i \\
&\propto \left(  \frac{q_i}{p_i} Z_i \mathds{1}_{\supp(\bp)}(i) \left( \sum_{j = 1}^{k_i} \mu_j^i \xi_j^i \right) \right)_i \\
&= \left( \frac{q_i}{p_i}  Z_i\mathds{1}_{\supp(\bp)}(i) \sum_{j = 1}^{k_i} \mu_j^i \frac{(\nu_j^i/\mu_j^i) \mathds{1}_{\supp(\bmu^i)}(j)}{\sum_{\ell = 1}^{k_i} (\nu_{\ell}^i/\mu_{\ell}^i) \mathds{1}_{\supp(\bmu^i)}(\ell)} \right)_i \\
&= \left( \frac{q_i}{p_i} \mathds{1}_{\supp(\bp)}(i) m_{\bmu^i}(\bnu^i) \right)_i \\
&\propto (\bq \ominus_{\bp} \bp) \oplus_{\bp} \bh_{\bmu}(\bnu),
\end{split}
\end{equation*}
the claim is proved and we have the chain rule in \eqref{eq:modified.chain.rule.I}. Finally, note that if $\supp(\bmu^i)=\supp(\bnu^i)$ then $\bh_{\bmu}(\bnu)=\barE_n$. Hence \eqref{eq:modified.chain.rule.I} reduces to \eqref{eqn:relative.entropy.chain.rule}.  
\end{proof}

\begin{lemma}[Characterization on $\Delta_n^{\circ} \times \Delta_n^{\circ}$] \label{rem:interior.char}
Theorem \ref{thm:characterization.rel.entr.reduced} holds if the domain $\cA_n$ of $\sG$ and \ref{B1}--\ref{B4} is replaced by $\Delta_n^{\circ} \times \Delta_n^{\circ}$.\footnote{If $(\bpi, \br) \in \Delta_n^{\circ} \times \Delta_n^{\circ}$, then $\cC_{\bpi}[\br] = \br$. Thus \ref{B5} holds automatically and may be removed.} 
\end{lemma}
\begin{proof}
We only need to show that (i) implies (ii). Given a family $(\sG : \Delta_n^{\circ} \times \Delta_n^{\circ} \rightarrow \R)_{n \geq 1}$ that satisfies \ref{B1}--\ref{B4}, define  $I\divg{\bp}{\bq} = \sG(\bp,\bq\ominus\bp)$ for $(\bp, \bq) \in \Delta_n^\circ\times \Delta_n^\circ$. From Lemma \ref{lem:implied.conditions.G}, $(I: \Delta_n^\circ \times \Delta_n^\circ \rightarrow \R)_{n \geq 1}$ satisfies \ref{C1}--\ref{C4}. By Proposition \ref{prop:Interior.Char.Rel.Entropy}, there exists $c\in\R$ such that for all $n \geq 1$ we have
\[
\sG(\bp, \bq\ominus\bp) = c H\divg{\bp}{\bq}, \quad (\bp,\bq)\in\Delta_n^\circ\times\Delta_n^\circ.
\]
Since $\ominus$ is invertible on $\Delta_n^\circ$, we get (by Lemma \ref{lem:link.rel.entr})
\[
\sG(\bp, \bq) = c H\divg{\bp}{\bp\oplus\bq} = c\Gamma(\bp,\bq), \quad (\bp,\bq)\in\Delta_n^\circ\times\Delta_n^\circ.\qedhere
\]
\end{proof}

{\it Step 3.} We extend the characterization from $\Delta_n^{\circ} \times \Delta_n^{\circ}$ to all of $\cA_n$. To do so, we need two auxiliary results that address the boundary values. For $\bp \in [0, \infty)^n$ with support $\supp(\bp)=\{j_1,\dots,j_d\} \neq \emptyset$ (ordered according to increasing index $j_1<j_2<\cdots<j_d$), 
we define the \emph{coordinate projection operator} $\Pi_{\bp} : [0, \infty)^n \to [0, \infty)^d$, $d = |\supp(\bp)|$, by
\begin{equation*}
  \big(\Pi_{\bp}[\bq]\big)_i \coloneqq q_{j_i}, \qquad i=1,\dots,d.
\end{equation*}
In words, $\Pi_{\bp}[\bq]$ restricts $\bq$ to the coordinates in $\supp(\bp)$. Note that
\[
(\bp, \bq) \in \cA_n \Rightarrow \Pi_{\bp}[\cC_{\bp}[\bq]] \in \Delta_{|\supp(\bp)|}^{\circ}. 
\]
For clarity, in the following we sometimes use $I_k$ and $\sG_k$ to show explicitly the underlying dimension.

\begin{lemma}\label{lem:decomposition.In}
Suppose $(I\divg{\cdot}{\cdot}: \cA_n \rightarrow \R)_{n \geq 1}$ satisfies \ref{C2}--\ref{C3} and \ref{B4'}. Then, there exists a function $\varphi:(0,1]\to\R$ with $\varphi(1)=0$ such that 
for every $(\bp,\bq)\in \cA_n$ we have
\begin{equation}\label{eq:phi.In}
I\divg{\bp}{\bq} 
= \varphi(m_\bp(\bq)) 
+ I_{|\supp(\bp)|}\divg{\Pi_\bp[\bp]}{\Pi_\bp\left[\cC_{\bp}[\bq]\right]}.
\end{equation}
\end{lemma}

\begin{proof}
Fix $(\bp,\bq)\in \cA_n$ and set $d=|\supp(\bp)| \in [n]$. 
Write $\hat\bp=\Pi_\bp[\bp]\in\Delta_d^{\circ}$ and $\hat\bq=\Pi_\bp[\cC_\bp[\bq]]\in\Delta_d^{\circ}$.

{\it Case 1.} $d = n$. Then $\hat{\bp} = \bp$, $\hat{\bq} = \bq$ and $m_{\bp}(\bq) = 1$. The identity \eqref{eq:phi.In} holds by letting $\varphi(1) = 0$. 

{\it Case 2.} $d < n$. After permuting coordinates (using \ref{C2}) if necessary, we may assume $\supp(\bp)=\{1,\dots,d\}$. Let $m=m_\bp(\bq) \in (0, 1]
$. Next, define $\hat{\bq}'\in\Delta_{n-d}$ through $\hat{q}'_i= q_{i+d}/(1-m)$ for $i=1,\dots,n-d$ to account for the renormalized values of $\bq$ off the support of $\bp$. If $m=1$, we may take any arbitrary $\hat{\bq}'\in\Delta_{n-d}$. By construction, we may represent $\bp$ and $\bq$ as the compositions
\[
\bp = (1,0) \circ (\hat\bp,\hat\bq'), \quad \bq = (m,1-m)\circ (\hat\bq,\hat\bq').
\]
Since $((1,0),(m,1-m))\in \cA_2$ and $\supp(\hat{\bp}) = \supp(\hat{\bq}) = [d]$, the special case in \ref{B4'} applies and we obtain
\[
I\divg{\bp}{\bq} = I_2\divg{(1,0)}{(m,1-m)} + I_d\divg{\hat\bp}{\hat\bq}.
\]
Thus \eqref{eq:phi.In} holds with $\varphi(m)\coloneqq I_2\divg{(1,0)}{(m,1-m)}$. That $\varphi(1)=0$ follows from \ref{C3}.
\end{proof}

\begin{lemma}\label{lem:rep.Gn.supp.p}
    Suppose that $(\sG : \cA_n \rightarrow \R)_{n \geq 1}$ satisfies \ref{B1}--\ref{B4} and define $I\divg{\bp}{\bq} = \sG(\bp, \bq\ominus_{\bp} \bp)$ for $(\bp, \bq) \in \cA_n$. Then for $(\bp,\bq)\in \cA_n$ we have
    \[
    I\divg{\bp}{\bq} =I_{|\supp(\bp)|}\divg{\Pi_{\bp}[\bp] }{ \Pi_\bp\left[\cC_{\bp}[\bq]\right]}.
    \]
    In particular, the function $\varphi$ from Lemma \ref{lem:decomposition.In} vanishes identically.
\end{lemma}

\begin{proof}

    Since $(\sG: \cA_n \rightarrow \R)_{n \geq 1}$ satisfies \ref{B1}--\ref{B4}, Lemma \ref{lem:implied.conditions.G} implies that $(I : \cA_n \rightarrow \R)_{n \geq 1}$ satisfies \ref{C2}--\ref{C3} and \ref{B4'}. By Lemma \ref{lem:decomposition.In}, there exists a function $\varphi(\cdot)$ on $(0,1]$ with $\varphi(1)=0$ satisfying \eqref{eq:phi.In}. We claim that $\varphi(\cdot) \equiv 0$.
    
    Fix $n$ and $\bp \in \Delta_n$ such that $|\supp(\bp)| < n$. For $\alpha \in (0, 1]$, let $\bq^{(\alpha)} \in \cA_n(\bp \mid\cdot)$ be such that $\cC_{\bp}[\bq^{(\alpha)}]=\bp$ and $m_{\bp}(\bq^{(\alpha)})=\alpha$. Such a $\bq^{(\alpha)}$ can always be constructed by multiplying the coordinates of $\bp$ by $\alpha$ and distributing the remaining mass $1-\alpha$ arbitrarily on $[n]\setminus \supp(\bp)$.
 Then, by Lemma \ref{lem:decomposition.In} and \ref{C3},
    \begin{align*}
        I\divg{\bp}{\bq^{(\alpha)}} &= \varphi(m_{\bp}(\bq^{(\alpha)})) +I_{|\supp(\bp)|}\divg*{\Pi_{\bp}[\bp]}{ \Pi_\bp\left[\cC_{\bp}[\bq^{(\alpha)}]\right]}\\
        &= \varphi(\alpha) + I_{|\supp(\bp)|}\divg{\Pi_{\bp}[\bp]}{ \Pi_\bp[\bp]} = \varphi(\alpha).
    \end{align*}
    On the other hand, $I \divg{\bp}{\bq^{(\alpha)}} = \sG(\bp,\bq^{(\alpha)}\ominus_{\bp} \bp)$, and a direct calculation shows that $\bq^{(\alpha)}\ominus_{\bp} \bp = \overline{\be}_{\bp}$,
    the uniform distribution on $\supp(\bp)$.  By \ref{B3}, $\sG(\bp,\overline{\be}_{\bp})=0$, hence $I\divg{\bp}{ \bq^{(\alpha)}}=0$. Therefore, $\varphi(\alpha)=0$ for all $\alpha \in (0,1]$, and so $\varphi(\cdot)\equiv 0$.
\end{proof}

We are now ready to complete the proof of Theorem \ref{thm:characterization.rel.entr.reduced} (and therefore, of Theorem \ref{thm:characterization.rel.entr} as well).

\begin{proof}[Proof of Theorem \ref{thm:characterization.rel.entr.reduced}]
We have seen that $(c\Gamma)_{n \geq 1}$ satisfies \ref{B1}--\ref{B4}. On the other hand, suppose that the collection $(\sG : \cA_n \rightarrow \R)_{n \geq 1}$ satisfies \ref{B1}--\ref{B4}. By Lemma \ref{rem:interior.char}, there exists a $c\in\mathbb{R}$ such that
\[
\sG(\bp,\bq) = c \Gamma(\bp,\bq), \quad \text{for all $n$ and } (\bp,\bq)\in\Delta_n^\circ\times\Delta_n^\circ.
\]
Next, observe that for any $(\bp,\bq)\in \cA_n$, 
\[
\left(\Pi_{\bp}[\bp],\Pi_{\bp}\left[\cC_{\bp}[\bq]\right]\right)\in\Delta_{|\supp(\bp)|}^\circ\times\Delta_{|\supp(\bp)|}^\circ,
\]
and moreover, $(\Pi_\bp\left[\cC_{\bp}[\bq]\right]\ominus \Pi_{\bp}[\bp]) \in \Delta_{|\supp(\bp)|}^{\circ}$. Therefore, by Lemma \ref{lem:rep.Gn.supp.p} (and writing $I\divg{\bp}{\bq} = \sG(\bp, \bq\ominus_{\bp} \bp)$),
\begin{align*}
\sG(\bp,\bq\ominus_{\bp}\bp) &= I\divg{\bp}{\bq} \\
&=I_{|\supp(\bp)|}\divg{\Pi_{\bp}[\bp]}{ \Pi_\bp\left[\cC_{\bp}[\bq]\right]} \\
&= \sG_{|\supp(\bp)|}\left(\Pi_{\bp}[\bp], \Pi_\bp\left[\cC_{\bp}[\bq]\right]\ominus \Pi_{\bp}[\bp] \right)\\
&= c\Gamma_{|\supp(\bp)|}\left(\Pi_{\bp}[\bp], \Pi_\bp\left[\cC_{\bp}[\bq]\right]\ominus \Pi_{\bp}[\bp] \right).
\end{align*}
One readily checks that
\[\Gamma_{|\supp(\bp)|}\left(\Pi_{\bp}[\bp], \Pi_\bp\left[\cC_{\bp}[\bq]\right]\ominus \Pi_{\bp}[\bp] \right) = \Gamma\left(\bp, \bq\ominus_{\bp} \bp \right),\]
and hence
\[
\sG(\bp,\bq\ominus_{\bp}\bp) = c \Gamma\left(\bp, \bq\ominus_{\bp} \bp \right), \quad (\bp,\bq) \in \cA_n.
\]
Unwinding by writing $\cC_{\bp}[\bq] = (\bq\oplus_{\bp}\bp)\ominus_{\bp}\bp$, we see that this implies
\[
\sG(\bp,\cC_{\bp}[\bq]) = c \Gamma\left(\bp, \bq \right).\]
Finally, we invoke \ref{B5} to obtain
\[
\sG(\bp,\bq) = \sG(\bp,\cC_{\bp}[\bq]) = c \Gamma\left(\bp, \bq \right), \quad (\bp,\bq)\in \cA_n,
\]
which completes the proof.
\end{proof}

\subsection{Via Jensen gap} \label{sec:Jensen.gap}
In this subsection, we characterize the excess growth rate
\[
\Gamma(\bpi, \bR) = \log \left( \sum_{i \in \supp(\bpi) } \pi_i R_i \right) - \sum_{i \in \supp(\bpi)} \pi_i \log R_i, \quad (\bpi, \bR) \in \cD_n,
\]
where $n \geq 2$ is {\it fixed}, as the gap in Jensen's inequality with respect to the logarithm which is strictly concave.  

We say that $g: \cD_n \rightarrow \R$ is a {\it gap function} if there exists $\varphi: (0, \infty) \rightarrow \R$ (which may be neither convex nor concave) such that 
\begin{equation} \label{eqn:gap.function}
g(\bpi, \bR) = \varphi\left(\sum_{i \in \supp(\bpi)} \pi_i R_i \right) - \sum_{i \in \supp(\bpi)} \pi_i \varphi(R_i), \quad (\bpi, \bR) \in \cD_n.
\end{equation}
We call $\varphi$ the {\it generator} of $g$. From \eqref{eqn:EGR.financial}, $\Gamma$ is the gap function with generator $\varphi = \log$. Note that since $R_i > 0$ for $i \in \supp(\bpi)$, $\varphi$ only needs to be defined on $(0, \infty)$.

\begin{lemma}[Uniqueness of generator] \label{lem:generator.uniqueness}
The generator of a gap function is unique up to the addition of an affine function. More precisely, if $g$ is a gap function and $\varphi$ and $\tilde{\varphi}$ are generators (that is, \eqref{eqn:gap.function} holds for both $\varphi$ and $\tilde{\varphi}$), then $\varphi(R) - \tilde{\varphi}(R) \equiv aR + b$ for some $a, b \in \R$.
\end{lemma}
\begin{proof}
Let $0 < \underaccent{\bar}{R} < \bar{R}$ be given, and let
\[
\bR = (\underaccent{\bar}{R}, \bar{R}, 1, \ldots, 1) \in (0, \infty)^n.
\]

For $R \in [\underaccent{\bar}{R}, \bar{R}]$, consider
\[
\bpi(R) \coloneqq \left( \frac{\bar{R} - R}{\bar{R} - \underaccent{\bar}{R}}, \frac{R - \underaccent{\bar}{R}}{\bar{R} - \underaccent{\bar}{R}}, 0, \ldots, 0 \right) \in \Delta_n.
\]
Then $\langle \bpi(R), \bR\rangle = R$ and we have
\begin{equation*}
\begin{split}
g(\bpi, \bR) &
= \varphi(R) - \frac{\bar{R} - R}{\bar{R} - \underaccent{\bar}{R}} \varphi(\underaccent{\bar}{R}) - \frac{R - \underaccent{\bar}{R}}{\bar{R} - \underaccent{\bar}{R}} \varphi(\bar{R}) \\
  &= \tilde{\varphi}(R) - \frac{\bar{R} - R}{\bar{R} - \underaccent{\bar}{R}} \tilde{\varphi}(\underaccent{\bar}{R}) - \frac{R - \underaccent{\bar}{R}}{\bar{R} - \underaccent{\bar}{R}} \tilde{\varphi}(\bar{R}).
\end{split}
\end{equation*}
It follows that $\varphi(R) - \tilde{\varphi}(R)$ is affine in $R$ on $[\underaccent{\bar}{R}, \bar{R}]$. Since $\underaccent{\bar}{R}, \bar{R}$ are arbitrary (and the intercept and slope remain the same upon extension of the domain), $\varphi - \tilde{\varphi}$ is affine on $(0, \infty)$ and the lemma is proved.
\end{proof}

Our goal is to characterize the excess growth rate among the family of gap functions. For $\bR \in [0, \infty)^n \setminus \{\mathbf{0}\}$ (where ${\bf 0} = (0, \ldots, 0)$ is the zero vector), we define
\[
\cD_n(\cdot \mid \bR) \coloneqq \{ \bpi \in \Delta_n : (\bpi, \bR) \in \cD_n\}
\]
to be the {\it slice} of $\cD_n$ given the second slot. Consider the following assumptions on $g: \cD_n \rightarrow \R$:
\begin{enumerate}[label=\textup{(D\arabic*)},leftmargin=*]
\item\label{D1} For every $\bR \in [0, \infty)^n \setminus \{\mathbf{0}\}$, the map $\bpi\mapsto g(\bpi,\mathbf{R})$ is concave on $\cD_n(\cdot \mid \bR)$.
\item\label{D2} $g(\bpi,\mathbf{R})=0$ if $\mathbf{R}$ is constant on $\supp(\bpi)$.
\item\label{D3} $g(\bpi,\alpha \mathbf{R})=g(\bpi,\mathbf{R})$ for all  $(\bpi,\mathbf{R})\in \cD_n$ and $\alpha>0$.
\item\label{D4} For $m \in (0, \infty)$ and $\bR \in [0, \infty)^n \setminus \{\mathbf{0}\}$, let $C_{m,\mathbf{R}} \subset \cD_n(\cdot \mid \bR)$ be the constant mean set (which is convex) defined by
\[
C_{m,\mathbf{R}}\coloneqq\{\bpi\in \cD_n(\cdot \mid \bR): \langle \bpi, \bR\rangle =m\}.
\]
Then, for any $m$ and $\bR$, the map $\bpi \mapsto g(\bpi, \bR)$ is affine on $C_{m, \bR}$:
\begin{equation} \label{eqn:constant.mean.affine}
g(\bpi, \bR) = \langle \ba(\bR), \bpi\rangle + b(m), \quad \bpi \in C_{m, \bR},
\end{equation}
for some gradient $\ba(\bR) \in \R^n$ that depends only on $\bR$ and is Lebesgue measurable in $\bR$, and the intercept $b(m) \in \R$ depends only on $m$ and is Lebesgue measurable in $m$.
\end{enumerate}

Note that \ref{D3} encodes num\'{e}raire invariance. In Section \ref{sec:relative entropy}, num\'{e}raire invariance allows us to restrict the domain to the simplex; the main argument is then driven by the chain rule. Here, num\'{e}raire invariance is the key property that distinguishes the excess growth rate (again up to a multiplicative constant) among other gap functions. To motivate \ref{D4}, assume that $g$ is a gap function whose generator $\varphi$ is Lebesgue measurable. For definiteness, set $\varphi(0)=0$. On the constant mean set $C_{m, \bR}$, we have
\begin{equation} \label{eqn:gap.function.constant.mean}
\begin{split}
g(\bpi, \bR) = \varphi(m) - \sum_{i = 1}^n \pi_i \varphi(R_i),
\end{split}
\end{equation}
which is affine in $\bpi$. We may take $\ba(\bR) = (-\varphi(R_i))_{1 \leq i \leq n}$ and $b(m) = \varphi(m)$ which are Lebesgue measurable in $\bR$ and $m$ respectively.

\begin{theorem} [Characterization II] \label{thm:characterization.Jensen}
Let $n \geq 2$ and let $g: \cD_n \rightarrow \R$ be (jointly) Lebesgue measurable. 
\begin{itemize}
    \item[(i)] $g$ is a gap function if and only if it satisfies \ref{D2} and \ref{D4}. In this case, the generator $\varphi$ (which is unique up to an affine function by Lemma \ref{lem:generator.uniqueness}) is Lebesgue measurable, and is concave if and only if \ref{D1} holds.
    \item[(ii)] $g$ satisfies \ref{D2}--\ref{D4} if and only if $g = c \Gamma$ for some $c \in \R$. In this case, $c \geq 0$ if and only if \ref{D1} holds.
\end{itemize}
\end{theorem}

Despite the importance of Jensen's inequality, we have not been able to locate axiomatic characterizations of its {\it gap} in the literature.  
Before proving Theorem \ref{thm:characterization.Jensen}, we compare it with known results about the {\it quasiarithmetic mean} of which the {\it exponential mean} (also called the (weighted) {\it log--sum--exp} in machine learning)
\[
\log \left( \sum_{i \in \supp(\bpi)} \pi_i e^{r_i} \right) = \gamma(\bpi, \br) + \sum_{i \in \supp(\bpi)} \pi_i r_i
\]
is a member. For further details, see \cite[Chapter 5]{leinster2021entropy} and \cite[Chapter 4]{GMM09}.

 Let $\phi: I \rightarrow J$ be a homeomorphism between real intervals. Following \cite[Definition 5.1.1]{leinster2021entropy}, we define the {\it quasiarithmetic mean} on $I$ generated by $\phi$ to be the family $(M_{\phi}: \Delta_n \times I^n \rightarrow I)_{n \geq 1}$, where
\begin{equation} \label{eqn:quasi.arithmetic.mean}
M_{\phi}(\bpi, \br) \coloneqq \phi^{-1} \left( \sum_{i = 1}^n \pi_i \phi(r_i) \right), \quad (\bpi, \br) \in \Delta_n \times I^n.
\end{equation}
Taking $\phi = \exp(\cdot): I = \R \rightarrow J = (0, \infty)$ recovers the exponential mean. The following result, which characterizes the (unweighted) exponential mean, can be found in \cite[Theorem 4.15]{GMM09}:

\begin{proposition}[Characterization of unweighted exponential mean]
Fix $n \geq 1$ and let $\mathsf{M}: \R^n
\rightarrow \R$ be an unweighted quasiarithmetic mean, i.e., $\mathsf{M}(\cdot) = M_{\phi}(\barE_n, \cdot)$ for some $\phi: \R \rightarrow J$. The following are equivalent:
\begin{itemize}
\item[(i)] $\mathsf{M}$ is difference scale invariant, in the sense that
\begin{equation} \label{eqn:scale.invariant}
\mathsf{M}(\br + s\mathbf{1}) = \mathsf{M}(\br) + s, \quad s \in \R.
\end{equation}
\item[(ii)] $\mathsf{M}(\cdot) = M_{\phi}(\barE_n, \cdot)$ where $\phi(x) = e^{\alpha x}$ for some $\alpha \in \R \setminus \{0\}$ or $\phi(x) = x$.
\end{itemize}
\end{proposition}

Difference scale invariance \eqref{eqn:scale.invariant}, when expressed in terms of $\bR = e^{\br}$, corresponds to the num\'{e}raire invariance of the excess growth rate; see the role of \ref{D3} in Theorem \ref{thm:characterization.Jensen}(ii). Also, see \cite[Theorem 4.10]{GMM09} which provides a list of properties which characterize the (unweighted) quasi-arithmetic mean (for some $\phi$) as a family $(\mathsf{M} : \R^n \rightarrow \R)_{n \geq 1}$ of functions. Together, these two results characterize the (unweighted) exponential mean. The theory of generalized means, or more generally the theory of {\it aggregation functions} (see \cite{GMM09}) and {\it value} (as in \cite[Chapter 7]{leinster2021entropy}), answers the question ``what is the value of the whole in terms of its parts.'' There, properties such as {\it monotonicity} ($\bx \leq \by \Rightarrow \mathsf{M}(\bx) \leq \mathsf{M}(\by)$) are natural and crucial. On the other hand, the excess growth rate, as the {\it difference} between the exponential and arithmetic means (see \eqref{eqn:egr.small.r}), focuses on how the returns differ from each other (so monotonicity no longer plays a role). This perspective distinguishes our study from that of generalized means. %

\begin{proof}[Proof of Theorem \ref{thm:characterization.Jensen}]
(i) Suppose $g$ is a gap function. Clearly, it satisfies \ref{D2}. We show that the generator $\varphi$ is measurable. Then, \ref{D4} follows from the paragraph above Theorem  \ref{thm:characterization.Jensen}.

By Lemma \ref{lem:meas.chart}(b) the section $\bpi \mapsto g(\bpi,\bR)$ is Lebesgue measurable on $\Delta_n^\circ$ for almost every $\bR\in(0,\infty)^n$.
Applying \eqref{eqn:gap.function} we have that
\begin{equation}
\label{eq:phi.rep}
\varphi(\langle \bpi,\bR\rangle) = g(\bpi, \bR) + \sum_i \pi_i \varphi(R_i).
\end{equation}
Fix an integer $k\geq2$ and consider the open set $\{\bR\in(0,\infty)^n: R_1<1/k, R_n>k\}$. Since it has positive Lebesgue measure on $(0,\infty)^n$ there is a choice of $\bR$ in this set such that $\bpi \mapsto g(\bpi,\bR)$ is Lebesgue measurable.
The transformation $\bpi \mapsto (\langle \bpi,\bR\rangle, \pi_2,\dots, \pi_{n-1})$ is invertible and the first coordinate ranges over the \emph{nonempty} interval $(\min_i R_i, \max_i R_i)\supset (1/k,k)$. Hence, in view of the representation in \eqref{eq:phi.rep} and Lemma \ref{lem:meas.chart}(a), $\varphi$ is measurable on $(1/k,k)$ and, as $k$ is arbitrary, also on $(0,\infty)$.

Next, suppose that $g$ satisfies \ref{D2} and \ref{D4}. Then, there exist measurable functions $\ba: \R^n \rightarrow \R^n$ and $b: (0, \infty) \rightarrow \R$ such that
\[
g(\bpi, \bR) = \langle\ba(\bR), \bpi \rangle + b(\langle \bpi, \bR\rangle), \quad (\bpi, \bR) \in \cD_n.
\]
Define $\varphi = b$, which is measurable. Letting $\bpi = \be_i$ be the $i$th-basis vector, we have, for $R_i>0$,
\[
0 = g(\bpi, \bR) = a_i(\bR) + b(\langle \bpi, \bR \rangle) = a_i(\bR) + \varphi(R_i),
\]
where the first equality holds by \ref{D2}. It follows that $a_i(\bR) = -\varphi(R_i)$ when $R_i>0$, and since $\supp(\bpi)\subseteq\supp(\bR)$ when $(\bpi,\bR)\in \cD_n$ we have
\[
g(\bpi, \bR) = \varphi(\langle \bpi, \bR \rangle) - \sum_{i = 1}^n \pi_i \varphi(R_i) = \varphi\left(\sum_{i \in \supp(\bpi)}\pi_i R_i \right) - \sum_{i \in \supp(\bpi)} \pi_i \varphi(R_i).
\]
Thus, $g$ is a gap function, and its generator is measurable.

Given that $g$ is a gap function, it is clear that its generator $\varphi$ is concave if and only if $\bpi \mapsto g(\bpi, \bR)$ is concave on $\cD_n(\cdot \mid \bR)$ for every $\bR \in [0, \infty)^n \setminus \{\mathbf{0}\}$.

(ii) Suppose $g = c \Gamma$ for some $c \in \R$, so that $g$ is a gap function with generator $\varphi = c \log$. From (i), $g$ satisfies \ref{D2} and \ref{D4}. That $g$ satisfies \ref{D3} is a consequence of Proposition \ref{prop:numeraire.invariance} (num\'{e}raire invariance).

Now, suppose $g$ satisfies \ref{D2}--\ref{D4}. From (i), $g$ is a gap function with a measurable generator $\varphi$. We aim to use \ref{D3} to show that $\varphi$ is equal to $c \log$ plus an affine function, for some $c \in \R$. If so, we have $g = c\Gamma$.

For $\alpha > 0$, consider the function $k_{\alpha}: (0, \infty) \rightarrow \R$ defined by
\[
k_{\alpha}(u) \coloneqq \varphi(\alpha u) - \varphi(u), \quad u \in (0, \infty).
\]
Also define $h : (0, \infty) \rightarrow \R$ by
\[
h(\alpha) \coloneqq \varphi(\alpha) - \varphi(1) = k_{\alpha}(1).
\]
Note that $h(1) = 0$. Our strategy is to derive functional equations for $k_{\alpha}$ and $h$.

{\it Step 1: $k_{\alpha}$ is affine.} Fix $0 < u < v$. For $t \in [0, 1]$, consider
\[
\bpi = (1 - t, t, 0, \ldots, 0) \in \Delta_n \quad \text{and} \quad  \bR = (u, v, 1, \ldots, 1) \in (0, \infty)^n.
\]
By \ref{D3}, which is the homogeneity property $g(\bpi, \alpha \bR) = g(\bpi, \bR)$, we have
\[
\varphi \left( (1 - t) \alpha u + t \alpha v \right) - (1 - t) \varphi(\alpha u) - t \varphi(\alpha v) = \varphi( (1 - t) u + t v) - (1 - t) \varphi(u) - t \varphi(v).
\]
Writing this in terms of $k_{\alpha}$ gives
\[
k_{\alpha}( (1 - t) u + t v) = (1 - t) k_{\alpha} (u) + t k_{\alpha}(v).
\]
Thus, $k_{\alpha}$ is affine on $(0, \infty)$, and there exist unique $a_{\alpha}, b_{\alpha} \in \R$ such that
\[
k_{\alpha}(u) = a_{\alpha} u + b_{\alpha}, \quad u \in (0, \infty).
\]

{\it Step 2: $a_{\alpha}$ is affine in $\alpha$.} Observe that
\[
k_{\alpha}(u) = \varphi(\alpha u) - \varphi(u) = [h(\alpha u) + \varphi(1)] - [h(u) + \varphi(1)] = h(\alpha u) - h(u).
\]
On the other hand, we have
\[
k_{\alpha}(1) = a_{\alpha} + b_{\alpha} = \varphi(\alpha) - \varphi(1) = h(\alpha).
\]
And so,
\[
h(\alpha u) = h(u) + h(\alpha) + a_{\alpha}(u - 1), \quad \alpha, u \in (0, \infty).
\]
Swapping $\alpha$ and $u$ gives
\[
h(u \alpha) = h(\alpha) + h(u) + a_u(\alpha - 1).
\]
Equating the two expressions gives
\[
0 = a_{\alpha} (u - 1) - a_u (\alpha - 1).
\]
Thus, for any $\alpha, u \in (0, \infty) \setminus \{1\}$ we have
\[
\frac{a_{\alpha}}{\alpha - 1} = \frac{a_u}{u - 1}.
\]
We conclude that there is a constant $r \in \R$ such that
\[
a_{\alpha} = r (\alpha - 1).
\]

{\it Step 3: Cauchy's functional equation for $\tilde{h}(u) \coloneqq h(u) - ru + r$.} From Step 2, we have
\[
h(\alpha u) = h(u) + h(\alpha) + r(\alpha - 1)(u - 1), \quad \alpha, u \in (0, \infty).
\]
Rearranging yields
\[
(h(\alpha u) - r \alpha u + r) = (h(\alpha) - r\alpha + r) + (h(u) - ru + r).
\]
Letting $\tilde{h}(u) = h(u) - ru + r$, we have the functional equation
\begin{equation} \label{eqn:h.functional.equation}
\tilde{h}(\alpha u) = \tilde{h}(\alpha) + \tilde{h}(u), \quad \alpha, u \in (0, \infty).
\end{equation}
If we make the exponential change of variables $\alpha = e^x$, $u = e^y$, $x, y \in \R$, and let $\psi(x) \coloneqq \tilde{h}(e^x)$, then \eqref{eqn:h.functional.equation} is equivalent to {\it Cauchy's functional equation}
\begin{equation} \label{eqn:Cauchy}
\psi(x + y) = \psi(x) + \psi(y), \quad x, y \in \R.
\end{equation}
Since $\psi$ is measurable, there exists $c \in \R$ such that $\psi(x) \equiv cx$, see \cite[Theorem 1.1.8]{leinster2021entropy}. Unwinding the transformations, we have
\[
\tilde{h}(u) = c \log u, \quad u > 0.
\]
(Alternatively, we may apply directly \cite[Corollary 1.1.11]{leinster2021entropy} to \eqref{eqn:h.functional.equation}.) It follows that
\begin{equation} \label{eqn:varphi.solution}
\varphi(u) = h(u) + \varphi(1) = \tilde{h}(u) + ru - r + \varphi(1) = c \log u + (\varphi(1)  - r) + ru.
\end{equation}
Thus, $\varphi$ is equal to $c \log u$ plus an affine function. Finally, we note that $\varphi$ given by \eqref{eqn:varphi.solution} is concave if and only if $c \geq 0$.
\end{proof}

\subsection{Via logarithmic divergence and cross-entropy} \label{sec:log.divergence}
In this subsection, we consider the excess growth rate as the divergence $\Gamma_{\bpi}\divg{\bY}{\bX}$ (see Definition \ref{def:egr.divergence}). We fix $n \geq 2$ and, for simplicity, restrict $\bpi \in \Delta_n^{\circ}$ so that $\bX, \bY \in (0, \infty)^n$. By num\'{e}raire invariance, we may replace $\bX$ and $\bY$ by $\bp = \cC[\bX]$ and $\bq = \cC[\bY]$ respectively, and hence regard $\Gamma_{\bpi}\divg{\cdot}{\cdot}$ as a divergence on $\Delta_n^{\circ}$. We characterize the excess growth rate as the unique {\it logarithmic divergence} (Definition \ref{def:log.divergence}) which is {\it perturbation invariant} in the sense of \eqref{eqn:perturbation.invariant} below. In fact, our result can be equivalently stated as a characterization theorem for the {\it cross-entropy}.

To motivate our result, we first prove a characterization of the Mahalanobis distance \eqref{eqn:Mahalanobis}.\footnote{This result is probably known by experts but we are unable to find an exact reference in the literature. The closest result we could locate, proved in \cite{NBN07}, states that the squared Mahalanobis distance is the only Bregman divergence on $\R^n$ which is {\it symmetric} in the sense that $\mathbf{B}_{\phi}\divg{\by}{\bx} = \mathbf{B}_{\phi}\divg{\bx}{\by}$ for all $\bx, \by$. } Recall that the {\it Bregman divergence} \cite{B67} of a differentiable convex function $\phi$ on a convex subset of $\R^n$ is defined by
\begin{equation} \label{eqn:Bregman.divergence}
B_{\phi}\divg{\by}{\bx} \coloneqq \phi(\by) - \phi(\bx) - \nabla_{\by - \bx} \phi(\bx).
\end{equation}
When $\phi$ is strictly convex, we have $B_{\phi}\divg{\by}{\bx} = 0$ only if $\bx = \by$. If $\phi: \R^n \rightarrow \R$ is a quadratic function of the form $\phi(\bx) = \frac{1}{2} \bx^{\top} A \bx + \bb^{\top} \bx + c$ where $A$ is an $n \times n$  positive semidefinite matrix, $\bb \in \R^n$ and $c\in\R$ (we regard $\bx$ as a column vector), then
\begin{equation*} %
B_{\phi}\divg{\by}{\bx} = \frac12(\by - \bx)^{\top} A (\by - \bx), \quad \bx, \by \in \R^n,
\end{equation*}
is a squared Mahalanobis distance (provided $A$ is strictly positive definite). %

\begin{theorem}[Characterization of squared Mahalanobis distance as a Bregman divergence] \label{thm:Mahalanobis}
Let $\phi: \R^n \rightarrow \R$ be $C^2$ (twice continuously differentiable)  and strictly convex.\footnote{It is possible to assume only that $\phi$ is $C^1$. We assume $C^2$ to shorten the proof.} The following are equivalent:
\begin{itemize}
\item[(i)] $B_{\phi}\divg{\cdot}{\cdot}$ is invariant under translation, in the sense that
\begin{equation} \label{eqn:Bregman.translation.invariant}
B_{\phi}\divg{\by + \bz}{\bx + \bz} = B_{\phi}\divg{\by}{\bx}, \quad \bx, \by, \bz \in \R^n.
\end{equation}
\item[(ii)] $\phi(\bx) = \frac{1}{2} \bx^{\top} A \bx + \bb^{\top} \bx + c$ for some strictly positive definite matrix $A \in \R^{n \times n}$, $\bb \in \R^n$ and $c\in\R$.
\end{itemize}
In particular, any translation invariant Bregman divergence is a squared Mahalanobis distance.
\end{theorem}
\begin{proof}
It is clear that (ii) implies (i). Assume $\phi$ satisfies (i). Expanding and rearranging \eqref{eqn:Bregman.translation.invariant}, we have, for $\bx, \by, \bz \in \R^n$,
\[
(\phi(\by + \bz) - \phi(\by)) - (\phi(\bx + \bz) - \phi(\bx)) = \langle \nabla \phi(\bx + \bz) - \nabla \phi(\bx), \by - \bx \rangle.
\]
Differentiating with respect to $\by$ gives
\[
\nabla\phi(\by + \bz) - \nabla \phi(\by) = \nabla \phi(\bx + \bz) - \nabla \phi(\bx),
\]
which is independent of $\by$. Letting $\bz = t \bv$ for $\bv \in \R^n$, dividing both sides by $t \neq 0$ and letting $t \rightarrow 0$ shows that the Hessian $\nabla^2 \phi(\by)$ is a constant matrix $A$. It follows that $\phi$ is quadratic. Since $\phi$ is strictly convex, $A$ is strictly positive definite. Hence (ii) holds and the theorem is proved.
\end{proof}

\begin{remark}[Characterization of relative entropy as a Bregman divergence]
The negative Shannon entropy $\phi(\bp) = -H(\bp)$ is differentiable and strictly convex in $\bp \in \Delta_n^{\circ}$. It is well known (see e.g.~\cite[Chapter 1]{A16}) that the induced Bregman divergence is the relative entropy:
\[
B_{-H}\divg{\bp}{\bq} = H\divg{\bp}{\bq}, \quad \bp, \bq \in \Delta_n^{\circ}.
\]
We are unaware of a characterization of the relative entropy (within the family of Bregman divergences) which is a direct analogue of Theorem \ref{thm:Mahalanobis} or Theorem \ref{thm:egr.L.divergence} below. What we could find is the following result by Amari \cite[Corollary]{A09}: the relative entropy is the unique Bregman divergence which is also an $f$-divergence.
\end{remark}

Our third and last characterization of the excess growth rate is analogous to Theorem \ref{thm:Mahalanobis}, except that we need a different kind of divergence which is closely related to the R\'{e}nyi entropy/divergence (see Example \ref{eg:Renyi.log.divergence}  below). By an {\it exponentially concave function} on $\Delta_n^{\circ}$, we mean a function $\varphi: \Delta_n^{\circ} \rightarrow \R$ such that $\Phi = e^{\varphi}$ is concave on $\Delta_n^{\circ}$. Clearly, if $\varphi$ is exponentially concave then $\varphi$ itself is concave. The following definition is taken from \cite{PW16}. 

\begin{definition}[Logarithmic divergence] \label{def:log.divergence}
Let $\varphi$ be differentiable and exponentially concave on $\Delta_n^{\circ}$. Its logarithmic divergence is the function $L_{\varphi}\divg{\cdot}{\cdot}: \Delta_n^{\circ} \times \Delta_n^{\circ} \rightarrow \R_+$ defined by
\begin{equation} \label{eqn:L.divergence}
L_{\varphi}\divg{\bq}{\bp} \coloneqq \log \left( 1 + \nabla_{\bq-\bp} \varphi(\bp)\right) - \left( \varphi(\bq) - \varphi(\bp) \right), \quad (\bq, \bp) \in \Delta_n^{\circ} \times \Delta_n^{\circ}.
\end{equation}
\end{definition}

The {\it logarithmic divergence} is a logarithmic generalization of the Bregman divergence \eqref{eqn:Bregman.divergence}. To illustrate this point and to see that the logarithmic divergence is well defined, let $\varphi$ be exponentially concave and consider $\Phi = e^{\varphi}$ which is a positive concave function on $\Delta_n^{\circ}$. For $\bp, \bq \in \Delta_n^{\circ}$, concavity of $\Phi$ implies that
\[
\Phi(\bp) + \nabla_{\bq - \bp} \Phi(\bp) \geq \Phi(\bq).
\]
Dividing both sides by $\Phi(\bp) > 0$, we have
\[
1 + \nabla_{\bq - \bp} \varphi(\bp) \geq \frac{\Phi(\bq)}{\Phi(\bp)} = e^{\varphi(\bq) - \varphi(\bp)} > 0.
\]
We obtain the logarithmic divergence by taking the logarithm and then the difference. Exponential concavity of $\varphi$ leads to a logarithmic first-order approximation; it is more accurate than the usual linear approximation since $\log(1 + \nabla_{\bq-\bp}\varphi(\bp)) \leq \nabla_{\bq-\bp}\varphi(\bp)$. See \cite{PW16, PW18, W18, W19, WY19, wong2022tsallis} for in-depth studies of the logarithmic divergence motivated by portfolio theory and information geometry.

\medskip

Following \cite{PW16} (also see \cite[Example 3.1.6]{F02}), the excess growth rate can be expressed as a logarithmic divergence. Recall that the {\it cross-entropy} $H^{\times}\divg{\bp}{\bq}$ is defined for $\bp, \bq \in \Delta_n^{\circ}$ by
\begin{equation} \label{eqn:cross.entropy}
H^{\times}\divg{\bp}{\bq} \coloneqq -\sum_{i = 1}^n p_i \log q_i.
\end{equation}

\begin{proposition}[Excess growth rate as a logarithmic divergence] \label{prop:egr.as.L.divergence}
For $\bpi = (\pi_1, \ldots, \pi_n) \in \Delta_n^{\circ}$, the function $\varphi(\cdot) = -H^{\times}\divg{\bpi}{\cdot}$ is exponentially concave on $\Delta_n^{\circ}$. Moreover, its logarithmic divergence is the excess growth rate:
\begin{equation} \label{eqn:egr.as.L.divergence}
L_{\varphi}\divg{\bq}{\bp} = \Gamma_{\bpi}\divg{\bq}{\bp}, \quad \bp, \bq \in \Delta_n^{\circ}.
\end{equation}
\end{proposition}
\begin{proof}
For completeness, we provide a sketch of the proof. Consider $\varphi(\bp) = -H^{\times}\divg{\bpi}{\bp} = \sum_{i = 1}^n \pi_i \log p_i$. Note that $\Phi(\bp) = e^{\varphi(\bp)} = \prod_{i = 1}^n p_i^{\pi_i}$ is the weighted geometric mean that is concave in $\bp \in \Delta_n^{\circ}$. Hence $\varphi$ is exponentially concave. A direct computation shows that
\[
1 + \nabla_{\bq-\bp} \varphi(\bp) = \sum_{i = 1}^n \pi_i \frac{q_i}{p_i}.
\]
It follows from \eqref{eqn:L.divergence} that 
\begin{equation*}
\begin{split}
L_{\varphi}\divg{\bq}{\bp} &= \log \left( \sum_{i = 1}^n \pi_i \frac{q_i}{p_i} \right) - \sum_{i = 1}^n \pi_i \log \frac{q_i}{p_i} = \Gamma_{\bpi}\divg{\bq}{\bp}. \qedhere
\end{split}
\end{equation*}
\end{proof}

We give another fundamental example of logarithmic divergence that can be expressed in terms of information-theoretic quantities. 

\begin{example}[R\'{e}nyi divergence] \label{eg:Renyi.log.divergence} 
Let $\lambda \in (0, 1)$ and $\alpha = 1/\lambda \in (1, \infty)$. Consider
\begin{equation*} %
\varphi(\bp) = (\alpha - 1) H_{\alpha}(\lambda \odot \bp), \quad \bp \in \Delta_n^{\circ},
\end{equation*}
where $H_{\alpha}(\bp) \coloneqq (1/(1 - \alpha)) \log \left( \sum_{i = 1}^n p_i^{\alpha} \right)$ is the {\it R\'{e}nyi entropy} of order $\alpha$. Then $\varphi$ is exponentially concave and its logarithmic divergence is given by
\begin{equation}
L_{\varphi}\divg{\bq}{\bp} = (\alpha - 1) H_{\alpha}\divg{ \lambda \odot \bq }{\lambda \odot \bp},
\end{equation}
where $H_{\alpha}\divg{\bp}{\bq} 
= (1/(\alpha - 1)) \log \left( \sum_{i = 1}^n p_i^{\alpha} q_i^{1 - \alpha} \right)$ is the {\it R\'{e}nyi divergence} of order $\alpha$ (this is a special case of \eqref{eqn:Renyi.divergence.densities}). The details can be found in \cite[Proposition 2]{W19}. General relationships between the logarithmic divergence and R\'{e}nyi entropy/divergence are developed in \cite{W18, wong2022tsallis}.
\end{example}

We set out to characterize the excess growth rate within the family of logarithmic divergences. To simplify the proof, we impose some regularity conditions on $\varphi$. We say that an exponentially concave function $\varphi: \Delta_n^{\circ} \rightarrow \R$ is {\it regular} if it is $C^4$ on $\Delta_n^{\circ}$ and, for each $\bp \in \Delta_n^{\circ}$ and $\bv \in \R^n \setminus \{0\}$ with $v_1 + \cdots + v_n = 0$ (that is, $\bv$ is tangent to $\Delta_n^{\circ}$), we have
\begin{equation} \label{eqn:Phi.condition}
\left. \frac{\dd^2 }{\dd t^2} \right|_{t = 0} \Phi(\bp + t\bv) < 0, \quad \text{where } \Phi = e^{\varphi}.
\end{equation}
In particular, \eqref{eqn:Phi.condition} implies that $\Phi$ is strictly concave.

\begin{theorem}[Characterization III] \label{thm:egr.L.divergence}
Let $\varphi: \Delta_n^{\circ} \rightarrow \R$ be regular and exponentially concave. The following are equivalent:
\begin{itemize}
\item[(i)] $L_{\varphi}\divg{\cdot}{\cdot}$ is invariant under perturbations, in the sense that
\begin{equation} \label{eqn:perturbation.invariant}
L_{\varphi}\divg{\bq \oplus \bh }{ \bp \oplus \bh}  = L_{\varphi} \divg{\bq}{\bp
} \quad \text{for } \bp, \bq, \bh \in \Delta_n^{\circ},
\end{equation}
\item[(ii)] $\varphi = -H^{\times}\divg{\bpi}{\cdot} + c$ for some $\bpi \in \Delta_n^{\circ}$ and $c \in \R$.
\end{itemize}
\end{theorem}

\begin{figure}[t!] %
\vspace{-1.8cm}
\includegraphics[scale = 0.7]{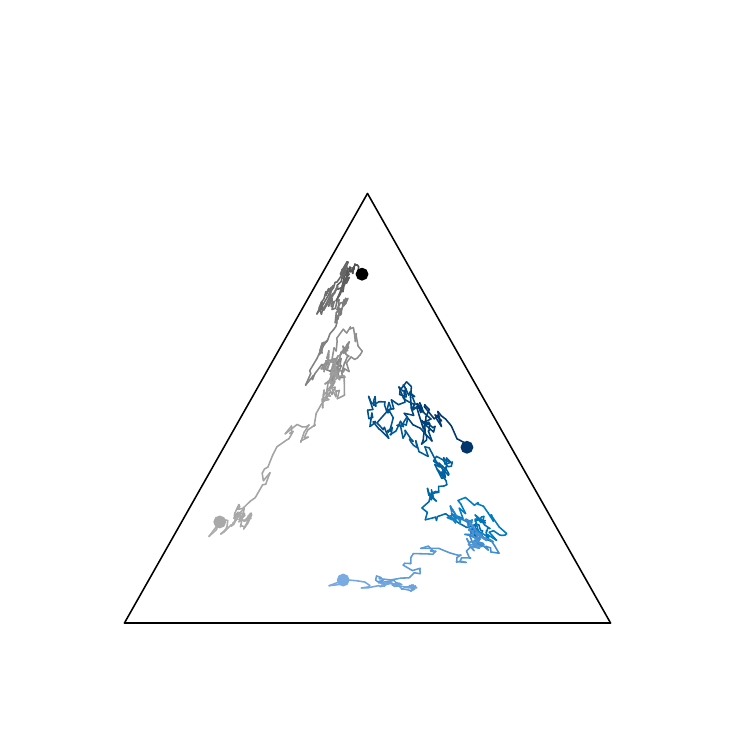}
\vspace{-1.5cm}
\caption{A path $\bp(t)$ in $\Delta_n^{\circ}$ and its perturbation $\bq(t) = \bp(t) \oplus \bh$.}
\label{fig:translation}
\end{figure}

We illustrate the perturbation invariance property \eqref{eqn:perturbation.invariant} in Figure \ref{fig:translation}. Consider a simulated path $(\bp(t))_{t = 0}^T$ in $\Delta_n^{\circ}$.\footnote{Here $n = 3$, $T = 500$ and the path is simulated using a $3$-dimensional Brownian bridge.} For some $\bh \in \Delta_n^{\circ}$, let $
\bq(t) = \bp(t) \oplus \bh$ be a perturbed path. The perturbation appears to be non-linear in the figure, but it is an ordinary translation in the Aitchison vector space $(\Delta_n^{\circ}, \oplus, \odot)$. Now, \eqref{eqn:perturbation.invariant} implies that
\[
\sum_{t = 0}^{T-1} L_{\varphi}\divg{\bp(t + 1)}{\bp(t)} = \sum_{t = 0}^{T-1} L_{\varphi}\divg{\bq(t + 1)}{\bq(t)}.
\]
That is, the two paths have the same cumulative (relative) volatility. Perturbation invariance is closely related to num\'{e}raire invariance. Observe that \eqref{eqn:perturbation.invariant} is equivalent to the identity
\begin{equation} \label{eqn:perturbation.invariant2}
L_{\varphi}\divg{\bp \oplus \br}{\bp} = L_{\varphi}\divg{\bq \oplus \br}{\bq}, \quad \bp, \bq, \br \in \Delta_n^{\circ}.
\end{equation}
(To see this, in \eqref{eqn:perturbation.invariant} replace $\bq$ by $\bp \oplus \br$ and $\bh$ by $\bq \ominus \bp$.) In \eqref{eqn:perturbation.invariant2}, we regard $\br = \cC[\bR]$ as the (normalized) gross return. If $\bp$ is the (normalized) initial price of the assets, then $\bp \oplus \br$ is the (normalized) final prices. The identity \eqref{eqn:perturbation.invariant2} states that the logarithmic divergence depends only on the returns and is independent of the initial prices. Theorem \ref{thm:egr.L.divergence} states that the excess growth rate is the only logarithmic divergence with this property. We believe the regularity conditions can be partially relaxed but do not pursue this further in this paper.

\begin{proof}[Proof of Theorem \ref{thm:egr.L.divergence}]\footnote{This result was claimed in \cite[Example 3.10]{PW18} without proof. We provide a complete argument here.}
We first show that (ii) implies (i). Suppose that $\varphi = -H^{\times}\divg{\bpi}{\cdot} + c$ for some $\bpi \in \Delta_n^{\circ}$ and $c \in \R$. Since $\varphi$ and $-H^{\times}\divg{\bpi}{\cdot}$ only differ by a constant, they induce the same logarithmic divergence. By Proposition \ref{prop:egr.as.L.divergence} (and num\'{e}raire invariance of $\Gamma$), we have $L_{\varphi} \divg{\bq }{ \bp} = \Gamma(\bpi, \bq \ominus \bp)$. Since the perturbation operation is commutative on $\Delta_n^{\circ}$, we have
\begin{equation*}
\begin{split}
L_{\varphi}\divg{\bq \oplus \br }{ \bp \oplus \br} &= \Gamma(\bpi, ((\bq \oplus \br) \ominus (\bp \oplus \br))) \\
&= \Gamma(\bpi, \bq \ominus \bp) \\
&= L_{\varphi}\divg{\bq }{ \bp}.
\end{split}
\end{equation*}
Hence $L_{\varphi}$ is invariant under perturbation.

The proof of the converse is more delicate. We will use tools from stochastic portfolio theory and information geometry, which will be introduced as needed, to derive differential implications of the functional equation \eqref{eqn:perturbation.invariant}. Suppose $L_{\varphi}$ is invariant under perturbation. Define a mapping $\bpi: \Delta_n^{\circ} \rightarrow \R^n$ by
\begin{equation} \label{eqn:fgp}
\bpi_i(\bp) \coloneqq p_i \left(1 + \nabla_{\be_i - \bp} \varphi(\bp)\right), \quad i = 1, \ldots, n,
\end{equation}
where $(\be_1, \ldots, \be_n)$ is the standard basis of $\R^n$. Since $\Phi$ is strictly concave, by \cite[Proposition 6]{PW16} we have that $\bpi(\bp) \in \Delta_n^{\circ}$ for $\bp \in \Delta_n^{\circ}$. Hence, $\bpi$ is a mapping from $\Delta_n^{\circ}$ into itself. We call $\bpi$ the {\it portfolio map} generated by $\varphi$. We claim that $\bpi(\bp)$ is {\it constant} in $\bp \in \Delta_n^{\circ}$. By an abuse of notation, we have $\bpi(\bp) \equiv \bpi$ for some element $\bpi$ of $\Delta_n^{\circ}$. On the other hand, the portfolio map generated by the exponentially concave function $-H^{\times}\divg{\bpi}{\cdot}$ is the constant $\bpi$ \cite[Example 3.1.6]{F02}. Then, by \cite[Proposition 6(i)]{PW16}, we have that $\varphi = -H^{\times}\divg{\bpi}{\cdot} + c$ for some $c \in \R$.\footnote{This is a variant of the classical fact that if two functions have the same gradient on a domain then they differ by a constant.}

To show that $\bpi(\cdot)$ is a constant mapping, we switch to another coordinate system on $\Delta_n^{\circ}$ under which the meaning of invariance under perturbation is more apparent. For $\bq \in \Delta_n^{\circ}$, we define its {\it exponential coordinates} $\boldsymbol{\theta} = (\theta_1, \ldots, \theta_{n-1}) \in \R^{n-1}$ by
\[
\theta_i = \log \frac{q_i}{q_n}, \quad i = 1, \ldots, n - 1.
\]
Similarly, let $\boldsymbol{\phi} = (\phi_1, \ldots, \phi_{n-1})$ be the exponential coordinates of $\bp \in \Delta_n^{\circ}$. Now, it is easy to verify that the exponential coordinates of $\bq \oplus \bp$ are
\[
\log \frac{(\bq \oplus \bp)_i}{(\bq \oplus \bp)_n} = \theta_i + \phi_i, \quad i = 1, \ldots, n - 1.
\]
That is, the exponential coordinate system is an isomorphism between the commutative groups $(\Delta_n^{\circ}, \oplus)$ and $(\R^{n-1}, +)$. 

Let $\widetilde{L}_{\varphi}: \R^{n-1} \times \R^{n-1} \rightarrow \R_+$ be the logarithmic divergence of $\varphi$ written in exponential coordinates:
\[
\widetilde{L}_{\varphi}\divg{\boldsymbol{\theta} }{ \boldsymbol{\phi}} \coloneqq L_{\varphi}\divg{\bq }{ \bp}.
\]
The assumption that $L_{\varphi}$ is invariant under perturbation is equivalent to the condition that $\widetilde{L}_{\varphi}$ is {\it invariant under translation}:
\[
\widetilde{L}_{\varphi} \divg{\boldsymbol{\theta} + \bh}{ \boldsymbol{\phi} + \bh} = \widetilde{L}_{\varphi}\divg{\boldsymbol{\theta} }{ \boldsymbol{\phi}}, \quad \boldsymbol{\theta}, \boldsymbol{\phi}, \bh \in \R^{n-1}.
\]
For $\boldsymbol{\theta} \in \R^{n-1}$, we define
\[
g_{ij}(\boldsymbol{\theta}) \coloneqq - \left. \frac{\partial}{\partial \theta_i} \frac{\partial}{\partial \phi_j} \widetilde{L}_{\varphi} \divg{\boldsymbol{\theta} }{ \boldsymbol{\phi}} \right|_{\boldsymbol{\phi} = \boldsymbol{\theta}}, \quad i, j = 1, \ldots, n - 1.
\]
In information geometry (see \cite[Chapter 6]{A16}), the matrix $(g_{ij}(\boldsymbol{\theta}))$ represents the {\it Riemannian metric} on $\Delta_n^{\circ}$ induced by the divergence $L_{\varphi}$, when expressed under the exponential coordinate system. The assumption that $\varphi$ is regular implies that the matrix $(g_{ij}(\boldsymbol{\theta}))$ is symmetric and strictly positive definite (see \cite[Theorem 4.5]{PW18}). We denote its inverse by $(g^{ij}(\boldsymbol{\theta}))$.

Furthermore, we define (using the $C^4$ condition)
\[
\Gamma_{ijk}(\boldsymbol{\theta}) \coloneqq - \left. \frac{\partial}{\partial \theta_i} \frac{\partial}{\partial \theta_j} \frac{\partial}{\partial \phi_k} \widetilde{L}_{\varphi} \divg{\boldsymbol{\theta}}{ \boldsymbol{\phi}} \right|_{\boldsymbol{\phi} = \boldsymbol{\theta}}, \quad i, j, k = 1, \ldots, n - 1,
\]
and
\[
\Gamma_{ij}^k(\boldsymbol{\theta}) \coloneqq \sum_{\ell = 1}^{n-1} \Gamma_{ij\ell}(\boldsymbol{\theta}) g^{\ell k}(\boldsymbol{\theta}), \quad i, j, k = 1, \ldots, n - 1.
\]
These are the {\it Christoffel symbols} of the so-called {\it primal affine connection} induced by the divergence. 
By \cite[Theorem 4.7]{PW18}, we have the identity
\begin{equation} \label{eqn:Christoffel.symbols}
\Gamma_{ij}^k(\boldsymbol{\theta}) = \delta_{ijk} - \delta_{ik} \bpi_j(\boldsymbol{\theta}) - \delta_{jk} \bpi_i(\boldsymbol{\theta}), \quad \boldsymbol{\theta} \in \mathbb{R}^{n-1}, %
\end{equation}
where $\delta_{ijk}$, $\delta_{ik}$ and $\delta_{jk}$ are Kronecker deltas and $\bpi(\boldsymbol{\theta}) \coloneqq \bpi(\bq)$ is the portfolio map expressed in exponential coordinates. 

The key observation is that since $\widetilde{L}_{\varphi}$ is translation invariant, the Christoffel symbols $\Gamma_{ij}^k(\boldsymbol{\theta})$ are {\it constant} in $\boldsymbol{\theta}$. Differentiating \eqref{eqn:Christoffel.symbols} with respect to $\theta_{\ell}$ gives
\[
-\delta_{ik} \frac{\partial}{\partial \theta_{\ell}} \bpi_j(\boldsymbol{\theta}) - \delta_{jk} \frac{\partial}{\partial \theta_{\ell}} \bpi_i(\boldsymbol{\theta}) = 0, \quad i, j, k, \ell = 1, \ldots, n - 1.
\]
Now, setting $i = j = k$ gives
\[
\frac{\partial}{\partial \theta_{\ell}} \bpi_i(\boldsymbol{\theta}) = 0, \quad i, \ell = 1, \ldots, n - 1.
\]
It follows that $\bpi(\boldsymbol{\theta})$ is constant in $\boldsymbol{\theta}$ (and hence $\bpi(\bp)$ is constant in $\bp$), and the claim is proved.
\end{proof}

\begin{remark}[Excess growth rate and the Fisher--Rao metric] \label{rmk:Fisher.Rao}
In Example \ref{eg:chain.rule} we computed the Taylor approximation of $\gamma(\bpi, \br)$ when $\br \approx 0$. A similar computation, applied to $\Gamma_{\bpi}\divg{\bp + t \bv}{\bp}$ for $\bp \in \Delta_n^{\circ}$ and $\bv \in \R^n$ tangent to $\Delta_n^{\circ}$ (i.e., $\sum_{i = 1}^n v_i = 0$), shows that
\begin{equation} \label{eqn:EGR.expansion2}
\Gamma_{\bpi}\divg{\bp + t\bv}{\bp} = \frac{t^2}{2} \sum_{i, j = 1}^n \frac{\pi_i(\delta_{ij} - \pi_j)}{p_ip_j} v_iv_j  + o(t^2), \quad \text{as } t \rightarrow 0.
\end{equation}
In information-geometric language (see \cite[Chapter 6]{A16}), this expansion defines the Riemannian metric induced by the divergence $\Gamma_{\bpi}\divg{\cdot}{\cdot}$. Letting $\bpi = \bp$ gives
\[
\Gamma_{\bp}\divg{\bp + t\bv}{\bp} = \frac{t^2}{2} \sum_{i = 1}^n \frac{v_i^2}{p_i} + o(t^2).
\]
Thus, we recover the {\it Fisher--Rao metric} $\|\bv\|_{\bp}^2 \coloneqq \sum_{i = 1}^n v_i^2/p_i$ at $\bp \in \Delta_n^{\circ}$. Further details can be found in \cite[Section 2.6]{PW16} and \cite{PW18}.
\end{remark}

\section{Optimization} \label{sec:optimization}

In this section, we study how the portfolio weights may be chosen to maximize the excess growth rate. For a random log-return vector $\br$ with values in $\R^n$, write $\bR\coloneqq e^{\br}$ and $\bm\coloneqq\mathbb{E}[\br]$.  The expected excess-growth objective is
\begin{equation}\label{eqn:obj.expected.egr}
J(\bpi)\coloneqq\mathbb{E}\big[\gamma(\bpi,\br)\big]
=\mathbb{E}\big[\log\langle\bpi,\bR\rangle\big]-\langle\bpi,\bm\rangle,
\qquad \bpi\in\Delta_n.
\end{equation}
Thus $J$ is the part of expected log-wealth growth that remains after subtracting the portfolio-weighted average of the assets' expected log returns.  The corresponding optimization problem is
\begin{equation}\label{eqn:expected.egr.objective}
\max_{\bpi\in\Delta_n}\ J(\bpi).
\end{equation}
If the coordinates of $\bm$ are all equal, then maximizing $J$ agrees with log-wealth maximization and leads to the classical growth-optimal portfolio problem discussed below in Remark~\ref{rmk:growth.optimal}. 

As noted in Appendix \ref{sec:egr.applications}, the log-wealth of a constant rebalanced portfolio admits a decomposition\footnote{The decomposition is more complex for portfolios whose holdings change over time; see \cite{PW13} and \cite[Corollary 1.1.6]{F02} for the continuous-time analogue.} in which the excess growth rate captures a rebalancing premium arising from the cross-sectional dispersion of asset returns, while the remaining term is the weighted average log return of the assets. Thus the objective $J$ targets the contribution to log-growth attributable to diversification and rebalancing, rather than the component driven by forecasts of individual expected returns, which are difficult to estimate. This perspective is central in stochastic portfolio theory, where excess growth is closely related to volatility harvesting and the benefits of diversification; see \cite[Chapters 1--2]{F02}. For these reasons, optimization of the excess growth rate is both practically appealing and theoretically informative. Related optimization problems for the excess growth rate have also been studied in \cite{ding2023optimization, maeso2020maximizing}. 

The deterministic problem
\begin{equation}\label{eqn:deterministic.egr.objective}
\max_{\bpi\in\Delta_n}\ \gamma(\bpi,\br)
\end{equation}
is the special case in which the log-return vector is deterministic.  Although this case is degenerate from the viewpoint of one-period portfolio choice, it is analytically useful: it has an explicit solution and explains the geometry behind the objective.  We first solve this deterministic problem, record a short variational consequence, and then return to the expected objective \eqref{eqn:expected.egr.objective}.

\subsection{The deterministic problem}\label{subsec:max-gamma}

Consider the Lagrangian for the deterministic concave maximization problem \eqref{eqn:deterministic.egr.objective}:
\[
\mathcal L(\bpi,\lambda,\bmu)
=
\log\!\left(\sum_{i=1}^n \pi_i e^{r_i}\right)
-\sum_{i=1}^n \pi_i r_i
-\lambda\left(\sum_{i=1}^n\pi_i-1\right)
+\sum_{i=1}^n \mu_i \pi_i,
\]
where the multiplier $\lambda\in\R$ enforces $\sum_{i = 1}^n\pi_i=1$ and the multipliers $\mu_i\ge 0$ enforce $\pi_i\ge 0$. The feasible set $\Delta_n$ is nonempty, compact, and convex. Moreover, {\it Slater's condition} holds: there exists a strictly feasible point for the inequality constraints that satisfies the equality constraint, e.g. $\barE_n=(1/n,\ldots,1/n)$.
As a result, strong duality holds (see \cite[Section 5.2.3]{boyd2004convex}) and the Karush--Kuhn--Tucker (KKT) conditions are necessary and sufficient for optimality.

At any maximizer, the KKT conditions require
\[
\frac{\partial \mathcal L}{\partial \pi_i}
=
\frac{e^{r_i}}{\sum_{j=1}^n \pi_j e^{r_j}}
-r_i-\lambda+\mu_i
=0,
\]
together with
\[
\mu_i\ge 0,\quad \pi_i\ge 0,\quad \mu_i\pi_i=0\ \text{(complementary slackness)},\quad \sum_{i=1}^n\pi_i=1.
\]
In particular, if $\pi_i>0$ then $\mu_i=0$ and
\[
\frac{e^{r_i}}{\sum_{j=1}^n \pi_j e^{r_j}}-r_i=\lambda,
\]
while for $\pi_i=0$ we have
\[
\frac{e^{r_i}}{\sum_{j=1}^n \pi_j e^{r_j}}-r_i \le \lambda.
\]
This leads us to the following structural characterization of any optimizer. In the sequel we will see that as long as $\br$ has distinct entries the optimizer is unique.

\begin{lemma}\label{lem:support-exactly-two}
If $\br\in\mathbb{R}^n$ has $n\geq2$ distinct coordinates, any maximizer $\bpi^\star$ of \eqref{eqn:deterministic.egr.objective} is supported on \emph{exactly} two indices; i.e., $\lvert\supp(\bpi^\star)\rvert=2$. In particular, $\bpi^\star$ has support on the maximum and minimum of $\br$, $r_{(n)}$ and $r_{(1)}$.
\end{lemma}

\begin{proof}
    Let $Z\coloneqq\sum_{j=1}^n \pi_j^\star e^{r_j}>0$ and define
\[
h_\lambda(x)\coloneqq\frac{e^{x}}{Z}-x-\lambda, \quad \lambda \in \R.
\]
Then $h_\lambda''(x)=e^{x}/Z>0$, so $h_\lambda$ is strictly convex and has at most two distinct zeros. From the KKT condition, for every $i$ in $\supp(\bpi^\star)$ we must have $h_\lambda(r_i)=0$. Since the $r_i$ are assumed to be distinct, we conclude $\lvert\supp(\bpi^\star)\rvert\leq 2$.

On the other hand, $\lvert\supp(\bpi)\rvert=1$ gives $\gamma(\bpi,\br)=0$, which is suboptimal whenever $\br$ is not constant, because for any $i\neq j$ and $\bpi' \coloneqq \frac{1}{2}(\be_i+\be_j)\in\Delta_n$,
\[
\gamma(\bpi',\br)=\log\!\Bigl(\tfrac{e^{r_i}+e^{r_j}}{2}\Bigr)-\tfrac{r_i+r_j}{2}>0.
\]

Finally, to see that $\pi_i^\star>0$ if and only if $r_i\in\{r_{(1)},r_{(n)}\}$ we observe that the KKT conditions imply that $h_{\lambda}(r_i)\leq 0$ for all $i$. %
By strict convexity and the fact that $\lim_{x\to\pm\infty}h_\lambda(x)=+\infty$ the sublevel set $\{x:h_\lambda(x)\le 0\}$ is a compact interval $[u,v]$ with $h_\lambda(u)=h_{\lambda}(v)=0$. Moreover,
\[
h_\lambda(x)<0 \ \text{for all } x\in(u,v), \qquad
h_\lambda(x)>0 \ \text{for all } x\notin[u,v].
\]
As $h_\lambda(r_i)\le 0$ for \emph{every} $i$, $r_i\in[u,v]$ for all $i$. Since the $r_i$ are distinct, $(u,v)$ is non-empty. In particular,
\[
u\le \min_i r_i = r_{(1)} \quad\text{and}\quad \max_i r_i = r_{(n)} \le v.
\]
But $u$ and $v$ themselves are necessarily entries of $\br$ as we have already established that $\br$ contains the two unique roots $\{x:h_\lambda(x)=0\}$. Thus we necessarily have
\[
u=r_{(1)} \quad\text{and}\quad v=r_{(n)}.
\]
Consequently, the only indices with $h_\lambda(r_i)=0$ are those achieving the minimum and maximum of $\br$ as claimed.
\end{proof}

This result allows us to obtain an explicit equation for the maximum and the optimal allocation $\bpi^\star$. The intuition is very much in the spirit of ``volatility harvesting.'' By comparison to the {\it wealth maximizing} strategy (which would invest solely in the stock with the highest return), the optimal strategy for the {\it excess growth rate} allocates capital to the most extreme returns, both the largest \emph{and} the smallest. %

\begin{theorem}\label{thm:pairwise-max}
Suppose $\br\in\R^n$ has $n\geq2$ distinct coordinates. Then
\[
\max_{\bpi\in\Delta_n}\gamma(\bpi,\br)
=
\log\!\Bigl(\tfrac{e^{r_{(n)}}-e^{r_{(1)}}}{r_{(n)}-r_{(1)}}\Bigr)
-\frac{e^{r_{(n)}}r_{(1)}-e^{r_{(1)}}r_{(n)}}{e^{r_{(n)}}-e^{r_{(1)}}}
-1.
\]
Moreover, the unique optimizer $\bpi^\star$ is supported on two points:
\[
\pi_{i^\star}^\star
=
\frac{\,e^{r_{i^\star}}-e^{r_{j^\star}}-(r_{i^\star}-r_{j^\star})e^{r_{j^\star}}\,}
{(r_{i^\star}-r_{j^\star})\,(e^{r_{i^\star}}-e^{r_{j^\star}})},
\quad
\pi_{j^\star}^\star=1-\pi_{i^\star}^\star,
\quad
\pi_k^\star=0\ \ (k\notin\{i^\star,j^\star\}),
\]
where the indices $\{i^\star,j^\star\}$ attain the maximum and minimum of $\br$, respectively; i.e., $r_{i^\star}=r_{(n)}$ and $r_{j^\star}=r_{(1)}$.
\end{theorem}
\begin{proof}
Fix distinct indices $\{i,j\}$ and parameterize $\bpi$ by
\[
\pi_i=t,\quad \pi_j=1-t,\quad t\in[0,1],\qquad \pi_k=0\ \ (k\notin\{i,j\}).
\]
Define the univariate objective
\[
f_{ij}(t)\coloneqq\gamma(\bpi,\br)=\log\!\bigl(te^{r_i}+(1-t)e^{r_j}\bigr)-\bigl(tr_i+(1-t)r_j\bigr).
\]
From Lemma \ref{lem:support-exactly-two} we conclude that
\begin{equation}\label{eqn:red.obj.rep}
\max_{\bpi\in\Delta_n}\gamma(\bpi,\br)
=
\max_{1\le i<j\le n}
\left\{\max_{t\in[0,1]} f_{ij}(t)\right\}=\max_{t\in[0,1]}f_{i^\star j^\star}(t).
\end{equation}
So, it suffices to treat the inner maximization problem for fixed $i\not=j$. One computes
\[
f_{ij}'(t)=\frac{e^{r_i}-e^{r_j}}{te^{r_i}+(1-t)e^{r_j}}-(r_i-r_j),
\qquad
f_{ij}''(t)=-\frac{(e^{r_i}-e^{r_j})^2}{\bigl(te^{r_i}+(1-t)e^{r_j}\bigr)^2}<0.
\]
Thus $f_{ij}$ is strictly concave on $[0,1]$. Since $f_{ij}(0)=f_{ij}(1)=0$ and $f_{ij}(1/2)>0$, the unique critical point is interior to $(0,1)$ and is therefore the constrained maximizer. Solving $f_{ij}'(t)=0$ yields
\[
t_{ij}^\star
=\frac{\,e^{r_i}-e^{r_j}-(r_i-r_j)e^{r_j}\,}{(r_i-r_j)\,(e^{r_i}-e^{r_j})}\in(0,1),
\]
and the corresponding maximal value simplifies to
\[
\max_{t\in[0,1]} f_{ij}(t)
=
\log\!\Bigl(\frac{e^{r_i}-e^{r_j}}{r_i-r_j}\Bigr)
-\frac{e^{r_i}r_j-e^{r_j}r_i}{e^{r_i}-e^{r_j}}
-1
\;>\;0\quad\text{for }r_i\neq r_j.
\]
Combining this with \eqref{eqn:red.obj.rep} completes the proof.
\end{proof}

\begin{remark}\label{rem:ties}
If $\br$ has repeated coordinates, we may aggregate equal entries; the same conclusions hold upon reducing to the list of distinct values. In particular, any maximizer is supported on exactly two distinct values of $\br$ unless $\br$ is constant, in which case any $\bpi\in\Delta_n$ is optimal. Solutions with less sparsity can arise in the optimization of the expected excess growth rate \eqref{eqn:expected.egr.objective} and can be enforced here by adding a regularization term.
\end{remark}

\subsubsection{Variational interpretation} \label{sec:variational.interpretation}

Before returning to the expected problem, we record one variational consequence of Theorem \ref{thm:pairwise-max}.  From Proposition \ref{prop:egr.variational} we obtain
\[
\sup_{(\bpi,\bq)\in\Delta_n\times\Delta_n}\Bigl\{\langle  \bq - \bpi,\br\rangle-H\divg{\bq}{\bpi}\Bigr\}
=
\sup_{\bpi\in\Delta_n}\gamma(\bpi,\br),
\]
and we understand the form of the optimal $\bq$ for fixed $\bpi$. More generally, for $\lambda>0$,
\begin{align}
\sup_{(\bpi,\bq)\in\Delta_n\times\Delta_n}\Bigl\{\langle  \bq - \bpi,\br\rangle-\lambda\,H\divg{\bq}{\bpi}\Bigr\}\nonumber
&=
\sup_{\bpi\in\Delta_n}\left\{\lambda\log\!\left(\sum_{i = 1}^n \pi_i e^{r_i/\lambda}\right)-\sum_{i = 1}^n \pi_i r_i\right\}\\
&=\lambda\,\sup_{\bpi\in\Delta_n}\gamma(\bpi,\br/\lambda).\label{eqn:penalized.joint.opt}
\end{align}
So, the maximization of $\gamma(\bpi,\br)$ in Theorem \ref{thm:pairwise-max} provides us with the solution to \eqref{eqn:penalized.joint.opt} which we collect in the next proposition. 

\begin{proposition}
    \label{prop:joint.optim.penalized.sol}
Suppose $\br\in\mathbb{R}^n$ has $n\geq2$ distinct coordinates. Then
    \begin{align*}
    \sup_{(\bpi,\bq)\in\Delta_n\times\Delta_n}&\Bigl\{\langle  \bq - \bpi,\br\rangle-\lambda\,H\divg{\bq}{\bpi}\Bigr\}\\
    &=
    \lambda\log\!\Bigl(\tfrac{e^{r_{(n)}/\lambda}-e^{r_{(1)}/\lambda}}{r_{(n)}-r_{(1)}}\Bigr)
    -\frac{e^{r_{(n)}/\lambda}r_{(1)}-e^{r_{(1)}/\lambda}r_{(n)}}{e^{r_{(n)}/\lambda}-e^{r_{(1)}/\lambda}}
    +\lambda\log\lambda-\lambda.
    \end{align*}
    Moreover, for the (unique) indices $\{i^\star,j^\star\}$ that attain the maximum and minimum of $\br$, respectively, the optimal pair $(\bpi^\star, \bq^\star)$ is given by 
    \[
\pi_{i^\star}^\star
=
\frac{\,e^{r_{i^\star}/\lambda}-e^{r_{j^\star}/\lambda}-\lambda^{-1}(r_{i^\star}-r_{j^\star})e^{r_{j^\star}/\lambda}\,}
{\lambda^{-1}(r_{i^\star}-r_{j^\star})\,(e^{r_{i^\star}/
\lambda}-e^{r_{j^\star}/\lambda})},
\quad
\pi_{j^\star}^\star=1-\pi_{i^\star}^\star,
\quad
\pi_k^\star=0\ \ (k\notin\{i^\star,j^\star\})
\]
and $\bq^\star = \bpi^{\star} \oplus_{\bpi^{\star}} \cC[e^{\br/\lambda}]$.
\end{proposition}

Financially, we can interpret this problem as finding the two portfolios $\bpi,\bq$ whose holdings differ maximally in their average log returns when subject to a relative entropy penalization. The optimal pair has support on the two most extreme returns and $\bq^\star$ tilts away from $\bpi^\star$ towards the rescaled return profile $\lambda^{-1}\br$. By a related argument, we can also link the excess growth rate $\gamma(\bpi,\cdot)$ to a constrained optimization problem through a {\it perspective transformation}.\footnote{The \emph{perspective} of a function $f(\br)$ is given by $\mathfrak{p}(\lambda,\br)\coloneqq\lambda f(\br/\lambda)$ for $\lambda\in(0,\infty)$. When $f$ is convex (concave), $\mathfrak{p}:(0,\infty)\times\mathbb{R}^n\to\mathbb{R}$ is \emph{jointly} convex (concave) (cf.~\cite[Section 3.2.6]{boyd2004convex}). The perspective of $\br\mapsto\gamma(\bpi,\br)$ has already appeared in \eqref{eqn:penalized.joint.opt}.} We defer the details to Appendix \ref{app:variational.constrained}.

\subsection{Maximizing the expected excess growth rate}\label{subsec:max-exp-egr}

We now return to the expected objective $J$ in \eqref{eqn:obj.expected.egr}. To ensure that $J$ is well defined and finite on $\Delta_n$, we impose the following mild integrability condition on $\br$.

\begin{assumption}
\label{ass:int}
Let $\br$ be an $\mathbb R^n$-valued random vector such that $\mathbb{E}[|r_i|]<\infty$ for all $1 \leq i \leq n$.
\end{assumption}
Under Assumption \ref{ass:int} the objective is finite (i.e., $J(\bpi)\in\mathbb{R}$) since
\[J(\bpi)\leq \mathbb{E}\left[\max_{i}\{r_i\}\right]-\langle\bpi, \bm \rangle\leq 2\sum_{i=1}^n\mathbb{E}\left[|r_i|\right]<\infty\]
and by choosing any $i\in \supp(\bpi)$,
\[J(\bpi)\geq \log(\pi_i) +\mathbb{E}[r_i]-\sum_{j=1}^n\mathbb{E}\left[|r_j|\right]\geq\log(\pi_i)-2\sum_{j=1}^n\mathbb{E}\left[|r_j|\right]>-\infty.\]

Our main result is the following necessary and sufficient first-order condition.

\begin{theorem}[First-order condition]
\label{thm:wealth.ratio}
Under Assumption \ref{ass:int}, a portfolio $\bpi^\star\in \Delta_n$ maximizes $J(\cdot)$ if and only if for every $\bpi\in\Delta_n$,
\begin{equation}\label{eq:wealth-ratio}
\mathbb{E}\left[\frac{\langle \bpi, \bR\rangle }{\langle \bpi^\star, \bR\rangle}\right]\ \le\ 1+\langle \bpi-\bpi^\star, \bm\rangle,
\end{equation}
with equality in \eqref{eq:wealth-ratio} whenever $\bpi$ is supported on $\supp(\bpi^\star)$.
\end{theorem}
\begin{proof}

Since $\gamma$ is concave in $\bpi$, we see that $J(\bpi)$ is concave in $\bpi$. In Lemma \ref{lem:superdifferential.set} of the Appendix we provide a direct characterization of the superdifferential set $\partial_{\Delta_n}^+ J(\bpi)$ (see Definition \ref{def:superdifferential.J}). Our approach is slightly more technical but avoids imposing additional integrability conditions on $\br$.

By standard convex analysis, we have that $\bpi^\star$ is a maximizer of $J(\cdot)$ if and only if
$\mathbf{0}\in \partial_{\Delta_n}^+ J(\bpi^\star)$ (cf.~\cite[Theorem 27.4]{rockafellar1997convex}). That is, in view of Lemma \ref{lem:superdifferential.set}, there is a $\lambda\in \mathbb{R}$ and a $\bmu\in\mathbb{R}^n_+$ with $\mu_i=0$ on $\supp(\bpi^\star)$ such that 
\begin{equation}\label{eqn:optimality.J}
\mathbb{E}\left[\frac{\bR}{\langle \bpi^{\star}, \bR\rangle }\right]-\bm-\lambda \mathbf{1}+\bmu=\mathbf{0},
\end{equation}
where $\mathbf{0}$ is the zero vector. Consequently, for any optimal $\bpi^\star$, we have the following inequality for all $i$:
    \[\mathbb{E}\left[\frac{R_i}{\langle \bpi^\star, \bR\rangle }\right]-m_i \le \lambda\]
and moreover, equality holds for all $i\in\supp(\bpi^\star)$. Multiplying by the coordinates of an arbitrary portfolio $\bpi\in\Delta_n$ we get
    \[\mathbb{E}\left[\frac{\pi_iR_i}{\langle \bpi^\star, \bR\rangle }\right]-\pi_im_i \le \pi_i\lambda.\]
    Summing over $i$ we get
    \[\mathbb{E}\left[\frac{\langle \bpi, \bR \rangle}{\langle \bpi^\star, \bR\rangle }\right]-\langle\bpi,\bm \rangle\le \lambda.\]
    In particular, if $\supp(\bpi)\subseteq\supp(\bpi^\star)$ then equality holds at all non-zero coordinates and 
    \[\mathbb{E}\left[\frac{\langle \bpi, \bR \rangle}{\langle \bpi^\star, \bR\rangle }\right]-\langle\bpi,\bm \rangle= \lambda, \qquad \supp(\bpi)\subseteq\supp(\bpi^\star).\]
    Taking $\bpi\equiv\bpi^{\star}$ in the above we see that $1-\langle\bpi^\star,\bm \rangle= \lambda$. The necessity of \eqref{eq:wealth-ratio} holding at any optimizer $\bpi^\star$ (with equality if $\supp(\bpi)\subseteq\supp(\bpi^\star)$) follows. 
    
    Sufficiency can be seen by choosing $\bpi=\be_j\in \Delta_n$ for $j=1,\dots,n$ and defining
    \begin{equation}\label{eqn:multipliers.J}
    \mu_j= 1+\langle \be_j-\bpi^\star, \bm\rangle - \mathbb{E}\left[\frac{\langle \be_j, \bR \rangle}{\langle \bpi^\star, \bR\rangle}\right]\geq0, \quad j=1,\dots,n.
    \end{equation}
    Since by hypothesis $\mu_j=0$ if $\supp(\be_j)=\{j\}\subseteq \supp(\bpi^\star)$, we recover \eqref{eqn:optimality.J} with $\lambda=1-\langle\bpi^\star,\bm \rangle$ by using \eqref{eqn:multipliers.J} and evaluating \eqref{eq:wealth-ratio} at $\bpi=\be_j\in \Delta_n$ for $j=1,\dots,n$.
\end{proof}

\begin{remark}[Growth optimal portfolio] \label{rmk:growth.optimal}
    If the linear term is absent, the objective
    \[
    \bpi \mapsto \mathbb{E}\left[\log\langle\bpi,\bR\rangle\right]
    \]
    is the classical \emph{log-wealth (growth rate) maximization} problem (cf.~\cite[Chapter 16]{CT06}). An optimal portfolio for this objective, $\bpi^{G}$, is said to be \emph{growth optimal}. The analogue of Theorem \ref{thm:wealth.ratio} for $\bpi^G$ is given in \cite[Theorems 16.2.1--16.2.2]{CT06}. There, it is shown that
    \[\mathbb{E}\left[\frac{\langle \bpi, \bR\rangle }{\langle \bpi^G, \bR\rangle}\right]\ \le\ 1, \quad \bpi\in\Delta_n,\]
    with equality holding if $\supp(\bpi)\subseteq\supp(\bpi^G)$.
\end{remark}

As a corollary, we obtain the following estimate on the relative growth rates by Jensen's inequality.
\begin{corollary}\label{cor:rel.growth.rate}
Let $\bpi^\star$ maximize $J(\cdot)$. Then, for all $\bpi\in\Delta_n$
\[\mathbb{E}\left[\log\left(\frac{\langle \bpi, \bR\rangle }{\langle \bpi^\star, \bR\rangle}\right)\right]\ \le\ \log\left(1+\langle \bpi-\bpi^\star,\bm\rangle \right).\]
\end{corollary}
\begin{remark}
    From Corollary \ref{cor:rel.growth.rate} we see that if $\bpi^{G}$ is the growth optimal portfolio then
    \[0\leq \mathbb{E}\left[\log\left(\frac{\langle \bpi^{G}, \bR\rangle }{\langle \bpi^\star, \bR\rangle}\right)\right]\ \le\ \log\left(1+\langle \bpi^{G}-\bpi^\star,\bm\rangle \right).\]
    Namely, the growth rate differential between $\bpi^\star$ and $\bpi^G$ is controlled by the deviation of $\bpi^{\star}$ from $\bpi^G$ and the expected returns $\bm$. Indeed, if $\bm$ is a \emph{constant vector} (i.e. all stocks have the same \emph{expected} log return) then $\langle \bpi^{G}-\bpi^\star,\bm\rangle=0$ and $\bpi^\star$ is also growth optimal. We can also use this chain of inequalities to conclude $\langle \bpi^\star,\bm\rangle\leq \langle \bpi^G,\bm\rangle$. That is, the log-optimal portfolio not only has a higher growth rate than $\bpi^\star$, but also a larger weighted expected-log-return component. %
\end{remark}

\section{Conclusion} \label{sec:conclusion}
Beginning with the financial definition of the excess growth rate, we demonstrate its rich connections with information-theoretic quantities, characterize it axiomatically from three complementary perspectives, and study its maximization that modifies the classical growth optimal portfolio. We conclude this paper by highlighting several natural questions related to our work. 
\begin{enumerate}
\item[1.] Motivated again by Leinster's book \cite{leinster2021entropy}, one may ask if there are (financially) meaningful {\it deformations} of the excess growth rate, analogous to how the R\'{e}nyi divergence deforms the relative entropy. If so, a natural follow-up question is to derive axiomatic characterization theorems that show these deformations are, in a sense, {\it canonical}. The divergence in Example \ref{eg:Renyi.log.divergence}, which involves the R\'{e}nyi divergence and corresponds to the so-called {\it diversity-weighted portfolio} in stochastic portfolio theory, seems to be a reasonable candidate.
\item[2.] As shown in \cite{larsson2025numeraire, ramdas2024hypothesis}, intuition and techniques from mathematical finance are instrumental in many modern developments in information theory, statistical inference, and hypothesis testing (Section \ref{sec:IT.literature}). Can the excess growth rate contribute to this development?
\item[3.] Our theoretical study of maximization of the excess growth rate covers only the basic one-period setting. From the practical perspective, it is both interesting and necessary to consider extensions to dynamic (multi-period or continuous-time) settings, as well as transaction costs and model uncertainty. One may also ask if there are relations with (suitable generalizations of) the asymptotic equipartition property and Cover's universal portfolio.
\item[4.] Closely related to the relative volatility (quantified by the excess growth rate) of a stock market is the concept of {\it market diversity}. Market diversity measures the concentration of a stock market. It is high when capital is spread more evenly among the different companies, and it is low when a small number of big companies dominate the entire market. In stochastic portfolio theory, it is typically quantified by the Shannon entropy (see \cite{F99}). Currently (2026), the diversity of the US market is rather low (relative to the historical average) due to the emergence of large technology companies. In fact, changes in market diversity tend to correlate with the performance of active large cap fund managers relative to the market; see \cite[Section 3]{campbell2024macroscopic} for detailed discussions and an empirical study. For a related geometric perspective on market diversity and the capital distribution, see \cite[Section 6]{campbell2025efficient}. A shortcoming of Shannon entropy (and similar quantities) is that it does not take into account the {\it similarities} between stocks. For example, stocks in the same industry sector tend to be more correlated with each other. In \cite{LC12} (also see \cite[Chapter 6]{leinster2021entropy}), generalized diversity measures (and hence entropies) were defined for probability distributions on a finite set equipped with a {\it similarity matrix}. It is interesting to adapt their approach to stock markets and study implications for portfolio selection.\footnote{This problem was suggested to us by Martin Larsson (private communication).}
\end{enumerate}

\appendix

\section{Financial applications of the excess growth rate} \label{sec:egr.applications}
We discuss two financial applications of the excess growth rate. First, the excess growth rate, when accumulated over time, captures the profit of a portfolio gained by {\it rebalancing}. By rebalancing, we mean the adjustment of positions through trading rather than buy-and-hold. The simplest situation, which can be considerably generalized (see \cite{PW13, PW16, W19}), is a portfolio that periodically rebalances to the same set of weights. For concreteness, consider a portfolio of $n$ stocks that rebalances to the weights given by $\bpi \in \Delta_n$ at the beginning of each month. Let $\bR(1), \ldots, \bR(T) \in (0, \infty)^n$ denote the gross returns of the stocks in $T$ months, and let $\br(t) = \log \bR(t)$ be the monthly log return. The total gross return of the rebalanced portfolio over $T$ months is given by compounding the monthly gross returns:\footnote{For simplicity, here we assume implicitly that there are no transaction costs.}
\[
\prod_{t = 1}^T \langle \bpi, \bR(t) \rangle.
\]
The log return, which is additive over time, is given by
\begin{align} \label{eqn:CW.decomposition}
\log \left( \prod_{t = 1}^T \langle \bpi, \bR(t) \rangle \right) &=
\sum_{t = 1}^T  \log \langle \bpi, \bR(t) \rangle   \\
&= \sum_{t = 1}^T \left( \log \langle \bpi, \bR(t) \rangle - \langle \bpi, \br(t) \rangle + \langle \bpi, \br(t) \rangle \right) \\
&= \left\langle \bpi, \sum_{t = 1}^T \br(t) \right\rangle + \sum_{t = 1}^T \Gamma(\bpi, \bR(t)).
\end{align}
That is, the log return of the rebalanced portfolio is the sum of the weighted average log return of the stocks and the accumulated excess growth rate. Since the accumulated excess growth rate is increasing, it contributes {\it positively} to the portfolio's log return. In particular, if we consider two price paths along which the stocks have the same initial and final prices, the rebalanced portfolio earns more over the path that has a larger accumulated excess growth rate. This analysis can be expanded to explain the empirical observation that a systematically rebalanced portfolio often (but not always) outperforms a capitalization-weighted benchmark portfolio over long horizons.\footnote{This phenomenon is sometimes called {\it volatility pumping} or {\it volatility harvesting}, see \cite{BNPS12, BNW15} and the references therein. Also see \cite{RX20} for a recent empirical study.}

\begin{figure}[t!] %
\includegraphics[scale = 0.5]{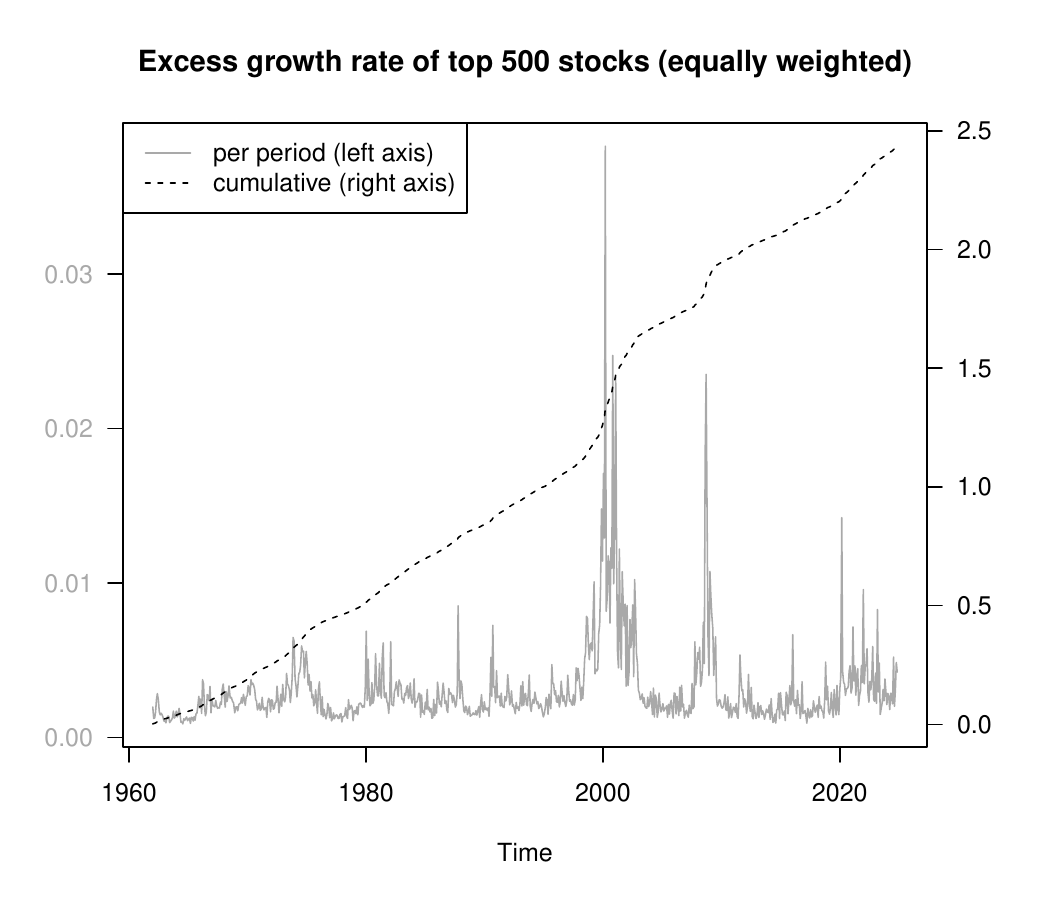}
\vspace{-0.3cm}
\caption{Excess growth rates of the largest $500$ stocks, over consecutive $20$-day periods and equally weighted, of the US stock market from $1962$ to $2024$. We show both the per period excess growth rate and its aggregate through time.}
\label{fig:EGR}
\end{figure}

Second, the excess growth rate, which is invariant under change of num\'{e}raire, quantifies the relative volatility of a stock market. Relative volatility measures how much the stocks' returns differ from each other. Intuitively, a market that is relatively volatile offers more opportunities to construct portfolios that may outperform the market. For precise statements in stochastic portfolio theory, see \cite{FKR18, fernholz2005relative} and the references therein. In fact, it can be argued that a portfolio should rebalance more frequently when the market is relatively more volatile, especially when transaction costs cannot be neglected \cite{PW18, W19}. In Figure \ref{fig:EGR}, which adopts the data-set and methodology used in \cite[Section 4]{campbell2024macroscopic},\footnote{Specifically, we use data from The CRSP US Stock Databases, see \url{https://www.crsp.org/products/research-products/crsp-us-stock-databases}. The data-set can be processed using code in the following repository: \url{https://github.com/stevenacampbell/Macroscopic-Properties-of-Equity-Markets}.} we illustrate the relative volatility of the US stock market. For each (non-overlapping) $20$-day period between 1962 and 2024,\footnote{A year has about 252 trading days, or roughly 21 days per month.} we identify the returns $\bR(t)$ of the $500$ largest stocks of the US stock market, and compute the excess growth rate $\Gamma(\bpi, \bR(t))$ where $\bpi = \barE_{500}$ is equally weighted. Note that due to rank switching (as well as events such as initial public offering and delisting) the set of the top $500$ stocks changes over time. From the figure, we see that the cumulative excess growth rate increases steadily but sometimes abruptly. The series of per period excess growth rates exhibits {\it clustering} of volatility which is typical in financial time series. Periods with high relative volatility can often be identified with major economic events such as the financial crisis in 2008 and COVID-19 in early 2020. For other studies of the excess growth rate in financial economics and portfolio management, we refer the reader to \cite{BFPRS19, FM07, maeso2020maximizing, mantilla2022can, willenbrock2011diversification} and their references.

\section{Proof of Proposition \ref{prop:Interior.Char.Rel.Entropy}}\label{app:proof.of.int.char.rel.entr}

The proof of this result takes three ingredients. The first is a recursion formula for $I\divg{\cdot}{\cdot}$ that is also satisfied by the relative entropy. The second ingredient is a functional equation in four variables that must be satisfied by $B(x,y) \coloneqq I_2\divg{(x,1-x)}{(y,1-y)}$ given the recursion formula. The third is a characterization of the symmetric solutions of this equation that vanish on the diagonal and are Lebesgue measurable in each variable for almost every value of the other. This latter result makes use of the general solution of an auxiliary functional equation due to \cite{ebanks1987generalized}.

\begin{remark}
    We provide here some historical context for our approach and also highlight why the specific structure in Proposition \ref{prop:Interior.Char.Rel.Entropy} necessitates a dedicated analysis. As will be seen, the domain $(\bp,\bq)\in\Delta_n^\circ\times\Delta_n^\circ$ leads to some subtleties that must be carefully checked.
    
    In \cite[Section 3.5]{leinster2021entropy}, Leinster provides the characterization in Proposition \ref{prop:Interior.Char.Rel.Entropy} on the larger domain $(\bp,\bq)\in \cA_n$ and also dispenses with measurability in the first argument. However, his proof (see \cite[Lemma 3.5.3]{leinster2021entropy}) fundamentally requires information about $I\divg{\cdot}{\cdot}$ outside of $\Delta_n^\circ\times\Delta_n^\circ$ and so cannot be directly adapted for our purposes. Instead, our proof follows an alternative line of argumentation that is well understood in the literature (see, for instance, the historical remarks in \cite[Section 3.5]{leinster2021entropy} and Section 2.1 of the survey \cite{csiszar2008axiomatic} where a result of this flavor is attributed to Kannappan and Ng). Indeed, under a related set of assumptions a characterization of relative entropy is proven in Kannappan's book (see \cite[Section 10.2f]{kannappan2009functional}). 
    
    The arguments employed in \cite[Section 10.2f]{kannappan2009functional} leverage Kannappan's work with Ng in \cite{kannappan1983generalized}. Importantly, the paper \cite{kannappan1983generalized} enables a characterization of the solution to the functional equation \eqref{eq:B.function} on the domain $x,u\in[0,1)$, $x+u\in[0,1]$, $y,v,y+v\in (0,1)$. For our purposes we need a characterization on the restricted domain $x,u,y,v,x+u,y+v\in (0,1)$. Fortunately, we may substitute the result of \cite{kannappan1983generalized} with a more general result that was proved a few years later by Ebanks, Kannappan and Ng \cite{ebanks1987generalized}. This result allows us to recover the same characterization using similar arguments, but without including the additional boundary points in the domain of the functional equation. Since we are unable to locate the precise characterization postulated in Proposition \ref{prop:Interior.Char.Rel.Entropy} in the literature, we provide a complete proof here.
\end{remark}

For notational convenience we use the leave-one-out notation $\bp^{(-i)}\in\mathbb{R}^{n-1}$ to denote the deletion of coordinate $i$ from the vector $\bp\in\mathbb{R}^n$ when $n\geq2$.
Here and earlier in the manuscript, we also use the following consequence of \cite[Theorems~2.39 and~2.47(a)]{folland1999real} to address the measure-theoretic subtleties of working with Lebesgue rather than Borel measurability.

\begin{lemma}
\label{lem:meas.chart}
Let $W\subset\R^k$ and
$\Omega\subset\R^m\times\R^{k-m}$ be open, where $1\leq m\leq k$,
and let $\psi:W\to\Omega$ be a $C^1$ diffeomorphism.
Set $\Phi=\mathrm{pr}_1\circ\psi$, where $\mathrm{pr}_1$ is the
projection onto the first $m$ coordinates.
\begin{enumerate}[label=\textup{(\alph*)}]
    \item A function $f:\Phi(W)\to\R$ is Lebesgue measurable if and only if $f\circ\Phi$ is Lebesgue measurable on $W$.
    \item If $m<k$ and $F:W\to\R$ is Lebesgue measurable, then for almost every $z\in\R^{k-m}$, the function $y\mapsto (F\circ\psi^{-1})(y,z)$ is Lebesgue measurable on $\{y\in\R^m:(y,z)\in\Omega\}$.
    Similarly, for almost every $y\in\R^m$, the function $z\mapsto (F\circ\psi^{-1})(y,z)$ is Lebesgue measurable on $\{z\in\R^{k-m}:(y,z)\in\Omega\}$.
\end{enumerate}
\end{lemma}

\begin{proof}
We claim that a function $H$ is Lebesgue measurable on $\Omega$ if and only if $H\circ\psi$ is measurable on $W$. Indeed, as $\psi$ is a $C^1$ diffeomorphism, for any Lebesgue measurable $H:\Omega\to\mathbb R$, by \cite[Theorem~2.47(a)]{folland1999real} we have that $H\circ\psi$ is Lebesgue measurable on $W$. Conversely, arguing identically with $\psi^{-1}$, we have that if $\widehat H=H\circ\psi:W\to\mathbb R$ is Lebesgue measurable then $H$ is Lebesgue measurable on $\Omega$. We will make use of this fact below.

We turn to the proof of (a). For $m=k$ it is clear, so we assume $m<k$. Take $H:\Omega\to\mathbb R$ to be given by $H(y,z)=f(y)$ for $(y,z)\in\Omega$. Our opening claim established that $H$ is measurable on $\Omega$ if and only if $H\circ\psi$ is measurable on $W$. Observe that $H\circ\psi$ is nothing else but $f(\mathrm{pr}_1\psi(w))$ for $w\in W$, which is exactly $f\circ\Phi$. Taken together, we know that $f\circ\Phi$ is measurable on $W$ if and only if $H$ is measurable on $\Omega$. In order to prove (a) it suffices to show that $f$ is measurable on $\Phi(W)$ if and only if $H$ is measurable on $\Omega$.

The forward implication is straightforward since if $f$ is measurable then the strict superlevel set $E_a:=\{y\in\Phi(W):f(y)>a\}$ is Lebesgue measurable and
\[
\{(y,z)\in\Omega:H(y,z)>a\}
=
\Omega\cap(E_a\times\mathbb R^{k-m}),
\]
which is measurable on the completion of the product space. Hence, $H$ is Lebesgue measurable.

For the converse, suppose that $H$ is Lebesgue measurable on $\Omega$ and let $\{U_i\times V_i\}_{i=1}^\infty$ be a countable collection of nonempty open rectangles that covers $\Omega$ with $U_i\times V_i\subseteq\Omega$ for all $i$ (e.g., rectangles with rational endpoints contained in $\Omega$). That is, $\Omega=\bigcup_{i=1}^\infty(U_i\times V_i)$.
Consider now the restriction of $H$ to $U_j\times V_j$ for fixed $j$. By \cite[Theorem~2.39]{folland1999real}, applied to positive and negative parts of $H$, the section $y\mapsto H(y,z)=f(y)$ is measurable on $U_j$ for almost every $z\in V_j$. Each such section equals $f|_{U_j}$, and since $V_j$ is nonempty and open, it has positive measure, so at least one such $z_j$ exists. Since $j$ was arbitrary, $f$ is measurable on $\bigcup_jU_j=\Phi(W)$, which completes the proof.

For (b), as in the opening claim, \cite[Theorem~2.47(a)]{folland1999real} guarantees the measurability of $F\circ\psi^{-1}$. Extend $F\circ\psi^{-1}$ by zero outside $\Omega$ and note that this remains measurable on $\mathbb{R}^m\times\mathbb{R}^{k-m}$ because $\Omega$ is open. Applying \cite[Theorem~2.39]{folland1999real} to its positive and negative parts as in (a), and restricting the resulting sections to the corresponding sections of $\Omega$, completes the proof.
\end{proof}

\begin{lemma}\label{lem:I.recursion}
    If $I_n\divg{\cdot}{\cdot}$ satisfies \ref{C2}--\ref{C4} then for any $(\bp,\bq)\in \Delta_n^\circ\times\Delta_n^\circ$ and $n\geq3$ it satisfies the recursion
    \[
    I_n\divg{\bp}{\bq} = I_2\divg{(p_i,1-p_i)}{ (q_i, 1-q_i)} + (1-p_i) I_{n-1}\divg*{\frac{\bp^{(-i)}}{1-p_i}}{ \frac{\bq^{(-i)}}{1-q_i}},
    \]
    for $i = 1, \ldots, n$.
\end{lemma}

\begin{proof}
We can write
\[\bp = (p_1 , 1 - p_1) \circ \left( (1), \left(\frac{p_2}{1-p_1},\dots,\frac{p_n}{1-p_1}\right)\right) \]
and 
\[\bq = (q_1 , 1 - q_1) \circ \left( (1), \left(\frac{q_2}{1-q_1},\dots,\frac{q_n}{1-q_1}\right)\right). \]
Observe that $(p_1 , 1 - p_1),(q_1,1-q_1)\in \Delta_2^\circ, 1\in \Delta_1^\circ$ and $\frac{\bp^{(-1)}}{1-p_1},\frac{\bq^{(-1)}}{1-q_1}\in \Delta_{n-1}^\circ$.  By the chain rule \ref{C4} and the fact that $I_1((1)\|(1)) =0$ (see \ref{C3})
\begin{align*}I_n\divg{\bp}{\bq} &= I_2((p_1,1-p_1)\|(q_1,1-q_1))+ (1-p_1) I_{n-1}\divg*{\frac{\bp^{(-1)}}{1-p_1}}{\frac{\bq^{(-1)}}{1-q_1}}.
\end{align*}
By permutation invariance \ref{C2}, the index $i=1$ is arbitrary.
\end{proof}

\begin{lemma}\label{lem:I_n.functional.eqn}
    If $I_n\divg{\cdot}{\cdot}$ satisfies \ref{C2}--\ref{C4} then $B(x,y) \coloneqq I_2\divg{(x,1-x)}{(y,1-y)}$ for $x,y\in(0,1)$ satisfies the functional equation
    \begin{equation}\label{eq:B.function}
    B(x,y)+(1-x)B\left(\frac{u}{1-x},\frac{v}{1-y}\right)=B(u,v)+(1-u)B\left(\frac{x}{1-u},\frac{y}{1-v}\right)
    \end{equation}
    on the (open) triangular domain $x,y,u,v,x+u,y+v\in (0,1)$.
\end{lemma}

\begin{proof}
    Applying Lemma \ref{lem:I.recursion} with $n=3$ and $i\neq j\in \{1,2,3\}$ we obtain
    \begin{align*}I_3\divg{\bp}{\bq} &= I_2\divg{(p_i,1-p_i) }{ (q_i, 1-q_i)} \\
    &\quad + (1-p_i) I_{2}\divg*{\left(\frac{p_j}{1-p_i},1-\frac{p_j}{1-p_i}\right)}{ \left(\frac{q_j}{1-q_i},1-\frac{q_j}{1-q_i}\right)}\\
    &=B(p_i,q_i)+(1-p_i)B\left(\frac{p_j}{1-p_i},\frac{q_j}{1-q_i}\right).
    \end{align*}
    Swapping the choice of indices and repeating the argument gives
    \begin{align*}
        I_3\divg{\bp}{\bq}=B(p_j,q_j)+(1-p_j)B\left(\frac{p_i}{1-p_j},\frac{q_i}{1-q_j}\right).
    \end{align*}
    Equating these two expressions recovers \eqref{eq:B.function}.
\end{proof}

\begin{lemma}\label{lem:func.eq.sol.char}
    If $B(\cdot,\cdot)$ is a solution to \eqref{eq:B.function} satisfying 
    \[B(x,y) = B(1-x,1-y) \quad \text{and} \quad B(x,x) =0, \quad x,y\in(0,1),\]
    and $x\mapsto B(x,y)$ is Lebesgue measurable for almost every $y\in (0,1)$ and $y\mapsto B(x,y)$ is Lebesgue measurable for almost every $x\in(0,1)$, then there exists a $c\in\R$ such that
    \[
    B(x,y) = c H\divg{(x,1-x)}{(y,1-y)}, \quad  x,y\in(0,1).\]
\end{lemma}

\begin{proof}
    Let $N_1$ be a null set such that $x\mapsto B(x,y)$ is Lebesgue measurable for \emph{every} $y\in (0,1)
    \setminus N_1$. Fix $y\in (0,1)\setminus N_1$ and $v\in (0,1-y)$. Writing 
    \[k(z) = B\left(z,\frac{y}{1-v}\right), \ g(z) = B\left(z,\frac{v}{1-y}\right),\ f(z) = B(z,y),\ h(z) = B(z,v),\]
    we may rewrite the functional equation for $B$ as
    \[f(x)+(1-x)g\left(\frac{u}{1-x}\right)=h(u)+(1-u)k\left(\frac{x}{1-u}\right)\]
    for $x,u,x+u\in (0,1)$. This is exactly the equation in \cite[Corollary 10.7c]{kannappan2009functional} (see also the original paper \cite{ebanks1987generalized}) for the identity $M(x)\equiv x$. As the identity map is both additive and multiplicative we get by \cite[Corollary 10.7c]{kannappan2009functional} the general solution
    \[
    \begin{aligned}
    f(x) &= x\,L(x) + (1-x)\,L(1-x) + \eta_{3}\,x - \eta_{2}\, (1-x) + \eta_{5},\\[2pt]
    g(x) &= x\,L(x) + (1-x)\,L(1-x) + \eta_{1}\,x + \eta_{2},\\[2pt]
    h(x) &= x\,L(x) + (1-x)\,L(1-x) + \eta_{1}\,x - \eta_{4}\,(1-x) + \eta_{5},\\[2pt]
    k(x) &= x\,L(x) + (1-x)\,L(1-x) + \eta_{3}\,x + \eta_{4},
    \end{aligned}
    \]
    where $L(\cdot)$ is a solution to the {\it logarithmic equation} 
    \begin{equation}\label{eq:logarithmic.eqn}
        L(xy) = L(x)+L(y), \qquad x,y\in(0,1),
    \end{equation} and $\eta_i$, $i=1,\dots,5$ are constants. Note that a priori all of the constants depend on $y,v,y+v\in(0,1)$. 

    Since $y\notin N_1$, $f=B(\cdot,y)$ is Lebesgue measurable, and hence so is
    \[
    S(x)\coloneqq xL(x)+(1-x)L(1-x),\qquad x\in(0,1).
    \]
    As a result, from the form of the general solution, $g,h$ and $k$ are also measurable. As $h=B(\cdot,v)$, we have that $B(\cdot,v)$ is also Lebesgue measurable. Since $y$ was arbitrary in $(0,1)\setminus N_1$ and $v\in(0,1-y)$ we can choose $y$ to conclude that $B(\cdot,v)$ is Lebesgue measurable for every $v\in(0,1)$. But then, the restriction to $y\in(0,1)\setminus N_1$ is unnecessary and we can fix any $y,v\in(0,1)$ with $y+v<1$ throughout.

    The function $S$ is symmetric and satisfies the so-called \emph{fundamental equation of information}:
    \[
    S(x)+(1-x)S\left(\frac{u}{1-x}\right)
    =
    S(u)+(1-u)S\left(\frac{x}{1-u}\right),
    \qquad x,u,x+u\in(0,1).
    \]
    Extending $S$ to $[0,1]$ by $S(0)=S(1)=0$, \cite[Theorem~3.4.22]{aczel1975measures} gives (up to normalization) that
    \[
    S(x)=c\bigl[x\log x+(1-x)\log(1-x)\bigr],
    \qquad x\in(0,1),
    \]
    for some $c=c(y,v)\in\R$. Hence, we can substitute $L(x) = c\log(x)$ in the expressions for $f,g,h$ and $k$.

    Define the binary entropy
    \[\mathcal{E}(x) = -x\log x  - (1-x)\log(1-x), \quad x \in (0, 1).\]
    By emphasizing the dependence of the constants on the parameters $v,y$ and substituting the form of the $f,g,h,k$ in terms of $B$, we arrive at the equations
    \[
    \begin{aligned}
    B(x,y) &=-c(y,v)\mathcal{E}(x) + \eta_{3}(y,v)\,x - \eta_{2}(y,v)\, (1-x) + \eta_{5}(y,v),\\[2pt]
    B\left(x,\frac{v}{1-y}\right) &= -c(y,v)\mathcal{E}(x)  + \eta_{1}(y,v)\,x + \eta_{2}(y,v),\\[2pt]
    B(x,v) &= -c(y,v)\mathcal{E}(x)  + \eta_{1}(y,v)\,x - \eta_{4}(y,v)\,(1-x) + \eta_{5}(y,v),\\[2pt]
    B\left(x,\frac{y}{1-v}\right) &= -c(y,v)\mathcal{E}(x)  + \eta_{3}(y,v)\,x + \eta_{4}(y,v).
    \end{aligned}
    \]
    Isolating the first and third equations we see that:
    \[B(x,y) =-c(y,v)\mathcal{E}(x) + a_1(y,v)\,x + b_1(y,v),\]
    \[B(x,v) = -c(y,v)\mathcal{E}(x)  + a_2(y,v)\,x +b_2(y,v)\]
    where $a_1(y,v) = \eta_{3}(y,v)+\eta_2(y,v)$, $a_2(y,v) = \eta_{1}(y,v)+\eta_4(y,v)$, $b_1(y,v)=\eta_{5}(y,v)- \eta_{2}(y,v)$ and $b_2(y,v)=\eta_{5}(y,v)- \eta_{4}(y,v)$.
    Fix $y$ and take any $v_1,v_2\in(0,1)$ with $v_1+y,v_2+y\in(0,1)$. Substituting these into the first equation and subtracting gives:
    \begin{align*}
        0&=-(c(y,v_1)-c(y,v_2))\mathcal{E}(x) + (a_1(y,v_1)-a_1(y,v_2))x+(b_1(y,v_1)-b_1(y,v_2))
    \end{align*} 
    for all $x\in(0,1)$. This implies
    \begin{align*}&c(y,v_1)=c(y,v_2), \quad a_1(y,v_1)=a_1(y,v_2), \quad b_1(y,v_1) =b_1(y,v_2).
    \end{align*}
    Repeating this argument with the second equation gives us
    \begin{align*}&c(y_1,v)=c(y_2,v), \quad a_2(y_1,v)=a_2(y_2,v), \quad b_2(y_1,v) =b_2(y_2,v).
    \end{align*}
    Since these hold for all admissible $v_1,v_2$ given $y$ (respectively, all admissible $y_1,y_2$ given $v$) we conclude\footnote{This conclusion is tacitly using that $\{\mathcal{E}(x),x,1\}$ are linearly independent on $(0,1)$.} that $c(y,v)\equiv c$ is constant and that $a_1(y,v), b_1(y,v)$ do not depend on $v$ (respectively, $a_2(y,v),b_2(y,v)$ do not depend on $y$). Thus, we deduce that $B(\cdot,\cdot)$ takes the form
    \begin{equation}\label{eq:B.reduced.form.a.b}
    B(x,y) =-c\mathcal{E}(x) + a(y) x + b(y), \ \ \ x,y\in(0,1),
    \end{equation}
    in terms of two univariate functions $a,b:(0,1)\mapsto\mathbb{R}$.

    The equation \eqref{eq:B.reduced.form.a.b} appears in exactly this form in \cite[Equation (10.50)]{kannappan2009functional}. By exploiting symmetry, we may employ analogous arguments to \cite{kannappan2009functional} in order to recover their \cite[Equation (10.50a)]{kannappan2009functional} and deduce that 
    \begin{equation}\label{eq:B.reduced.form}
    B(x,y) =-c\mathcal{E}(x) + x\ell(y) + (1-x)\ell(1-y), \ \ \ x,y\in(0,1),
    \end{equation}
    where $\ell(\cdot)$ is another solution of the logarithmic equation \eqref{eq:logarithmic.eqn}.

    With this verification complete, we can now make use of \eqref{eq:B.reduced.form} to complete the proof. Let $N_2$ be a null set such that $y\mapsto B(x,y)$ is Lebesgue measurable for \emph{every} $x\in (0,1)\setminus N_2$. Consider the sequence of functions $\mathfrak{b}_n(y) = B(x_n,y)$ for $x_n\uparrow 1$ with $x_n\notin N_2$. Each function is Lebesgue measurable by the measurability of $y\mapsto B(x,y)$ for each such $x$. Passing to the limit we define
    \[\mathfrak{b}_\infty(y)\coloneqq\lim_{n\to\infty} \mathfrak{b}_n(y) = \ell(y), \ \ \ y\in(0,1).\]
    As the pointwise limit of Lebesgue measurable functions, $\mathfrak{b}_\infty(y)=\ell(y)$ is Lebesgue measurable. Hence, $\ell(\cdot)$ is a \emph{Lebesgue measurable} solution of  $\eqref{eq:logarithmic.eqn}$ and so, there exists a $c'\in\mathbb{R}$ such that
    \[\ell(y) = c'\log(y), \ \ \ y\in(0,1).\]
    Using the property $B(x,x)=0$, we arrive at the identity
    \[
    c\mathcal{E}(x)
    =
    c'\left[x\log(x)+(1-x)\log(1-x)\right]
    =
    -c'\mathcal{E}(x),
    \qquad x\in(0,1),
    \]
    from which it necessarily follows that $c'=-c$. Putting this all together,
    \[
    B(x,y)
    =
    -c\left[\mathcal{E}(x)+x\log(y)+(1-x)\log(1-y)\right]
    =
    cH\divg{(x,1-x)}{(y,1-y)}.
    \]
    This completes the proof.
\end{proof}

With these ingredients we readily complete the proof of Proposition \ref{prop:Interior.Char.Rel.Entropy}.

\begin{proof}[Proof of Proposition \ref{prop:Interior.Char.Rel.Entropy}]
That relative entropy satisfies \ref{C1}--\ref{C4} is a standard verification, so we focus on the converse implication. 

First, note that the equality $I_1=cH_1$ trivially holds for $n=1$ as $\Delta_1^\circ=\{1\}$. Indeed, $I_1\divg{1}{1} = cH\divg{1}{1}=0$. Then, by Lemma \ref{lem:I_n.functional.eqn} and Assumption \ref{C1}, $B(x,y) = I_2\divg{(x,1-x)}{(y,1-y)}$ is a solution of the functional equation \eqref{eq:B.function} whose sections are Lebesgue measurable in each variable for almost every value of the other. By the permutation invariance and vanishing properties \ref{C2}--\ref{C3} of $I_2\divg{\cdot}{\cdot}$ we also get
\[B(x,y) = B(1-x,1-y) \quad \text{and} \quad B(x,x)=0, \quad x,y\in(0,1).\]
So, by Lemma \ref{lem:func.eq.sol.char} we conclude that there exists a $c\in\mathbb{R}$ such that
\[I_2\divg{(x,1-x)}{(y,1-y)} =  c H\divg{(x,1-x)}{(y,1-y)}, \qquad x,y\in(0,1).\]
To extend this to general $n\geq2$, we use that relative entropy satisfies the recursion of Lemma \ref{lem:I.recursion}. Applying Lemma \ref{lem:I.recursion} for $n=3$ and using that $I_2=cH$ yields
\[I_3\divg{\bp}{\bq} = c H\divg{\bp}{\bq}, \qquad \bp,\bq\in \Delta_3^\circ.\]
Iterating this recursion for $n=4,5,\dots$ completes the proof for general $n$.
\end{proof}

\section{Optimization Details}

\subsection{Constrained variational problem}\label{app:variational.constrained}
The main text introduces a penalized variational formula, which is naturally related to the constrained problem
\begin{equation}\label{eqn:constr.var.rep}
\sup_{(\bpi,\bq)\in\Delta_n\times\Delta_n:\ H(\bq\;\Vert\;\bpi)\le \eta}\langle \bq-\bpi,\br\rangle, \qquad \eta\geq0.
\end{equation}
We include the details here for completeness. For any fixed radius $\eta\ge 0$ we first define the objective
\begin{equation}\label{eqn:def.Phi}
\Phi_\eta(\mathbf r)
\coloneqq \sup_{\bq\in\Delta_n:\ H(\bq\;\Vert\;\bpi)\le \eta}\, \langle \bq-\bpi,\br\rangle. %
\end{equation}

\begin{lemma}[Perspective duality]\label{lem:perspective}
Fix $n\ge 2$ and $\bpi\in\Delta_n$. For $\eta\ge 0$ and $\mathbf r\in\mathbb R^n$ 
\begin{equation}\label{eq:perspective-identity}
\Phi_\eta(\mathbf r)
= \inf_{\lambda> 0}\Bigl\{\lambda\, \gamma\Bigl(\bpi, \frac{1}{\lambda}\br\Bigr)+\eta\,\lambda\Bigr\},
\end{equation}
with the convention
\begin{equation}\label{eq:perspective-limit}
\lambda\, \gamma\Bigl(\bpi,\frac{1}{\lambda}\br\Bigr)\rightarrow
\max_{i\in\supp(\bpi)} r_i-\langle\bpi,\mathbf r\rangle \quad \text{as} \quad \lambda\downarrow 0.
\end{equation}
\end{lemma}

\begin{proof}
Consider the (partial) Lagrangian
\[
\mathcal L(\bq,\lambda)
\coloneqq\langle \bq-\bpi,\mathbf r\rangle-\lambda\bigl(H(\bq\;\Vert\;\bpi)-\eta\bigr),
\qquad \bq\in\Delta_n,\ \lambda> 0 .
\]
For each fixed $\lambda> 0$, the map $\bq\mapsto \mathcal L(\bq,\lambda)$ is continuous and concave on the compact convex set
$\cD_n(\cdot\mid\bpi)$, while for fixed $\bq\in \cD_n(\cdot\mid\bpi)$ the map
$\lambda\mapsto \mathcal L(\bq,\lambda)$ is affine. Since $\mathcal L(\bq,\lambda)=-\infty$ whenever
$\bq\notin\cD_n(\cdot\mid\bpi)$, both suprema below
may be restricted to the compact convex set
$\cD_n(\cdot\mid\bpi)$. Hence Sion's minimax theorem yields
\[
\Phi_\eta(\br)
=\sup_{\bq\in\Delta_n}\inf_{\lambda> 0}\mathcal L(\bq,\lambda)
=\inf_{\lambda> 0}\sup_{\bq\in\Delta_n}\mathcal L(\bq,\lambda).
\]
If $\lambda>0$, we have
\begin{equation}\label{eqn:intermediate.step.Lagrangian}
\sup_{\bq\in\Delta_n}\mathcal L(\bq,\lambda)
=\lambda\sup_{\bq\in\Delta_n}\Bigl\{\Bigl\langle \bq-\bpi,\frac{\br}{\lambda}\Bigr\rangle-H(\bq\;\Vert\;\bpi)\Bigr\}
+\eta\lambda.
\end{equation}
By Proposition~\ref{prop:egr.variational} the inner supremum in \eqref{eqn:intermediate.step.Lagrangian} equals $\gamma(\bpi,\mathbf r/\lambda)$, giving
\[
\sup_{\bq\in\Delta_n}\mathcal L(\bq,\lambda)=\lambda\, \gamma(\bpi,\br/\lambda)+\eta\lambda .
\]
Combining the case $\lambda>0$ with the well-known fact that $\lambda\log\sum_i \pi_i e^{ r_i/\lambda}$
approximates $\max_{i:\pi_i>0} r_i$ as $\lambda\downarrow 0$ yields
\eqref{eq:perspective-identity} and \eqref{eq:perspective-limit}.
\end{proof}

The maximizer of this constrained optimization can be characterized as follows.

\begin{lemma}\label{lem:minimizer.perspective.duality}
    Set $\mathcal{M}_{\bpi}(\mathbf r)\coloneqq\arg\max_{i\in\supp(\bpi)} r_i$ and define $\overline{\eta}(\cdot)$ and $\bq(\cdot)$ through
\[
\overline{\eta}(\mathbf r)\coloneqq-\log\left(\sum_{j\in \mathcal{M}_{\bpi}(\mathbf r)} \pi_j\right), \quad 
q_i(\br)=\frac{\pi_i \exp(r_i)}{\sum_{j=1}^n \pi_j \exp(r_j)}=\bigl(\bpi \oplus_{\bpi} \cC[e^{\br}]\bigr)_i.
\]
Then:
\begin{enumerate}[label=(\alph*),leftmargin=*]
\item If $0\leq\eta<\overline{\eta}(\mathbf r)$, the infimum in \eqref{eq:perspective-identity} is attained at any $\lambda^\star> 0$ satisfying\footnote{The case $\eta=0$ is solved by setting $\lambda^\star=\infty$ which we identify with $\bq^\star=\lim_{\lambda\uparrow \infty}\bq(\br/\lambda)=\bpi$.} $H\divg{\bq(\br/\lambda^\star)}{\bpi}=\eta$ and the maximizer in \eqref{eqn:def.Phi} is given by $\bq^\star=\bq(\br/\lambda^\star)$.
\item If $\eta\ge \overline{\eta}(\mathbf r)$ then the infimum in \eqref{eq:perspective-identity} is achieved in the limit $\lambda\downarrow 0$, and any
$\bq^\star$ supported on $\mathcal{M}_{\bpi}(\mathbf r)$ with $H(\bq^\star\|\bpi)\le \eta$ is optimal for \eqref{eqn:def.Phi}. In particular, the distribution
\[
q^\star_i=\frac{\pi_i}{\sum_{j\in \mathcal{M}_{\bpi}(\mathbf r)}\pi_j}\,\mathbf 1_{\{i\in \mathcal{M}_{\bpi}(\mathbf r)\}}=\lim_{\lambda\downarrow0} q_i(\br/\lambda)
\]
satisfies
$H(\bq^\star\|\bpi)=\overline{\eta}(\mathbf r)$.
\end{enumerate}
\end{lemma}

\begin{proof}
First, observe that $\bq(\br/\lambda)=\nabla_{\br} \gamma(\bpi,\br/\lambda)+\bpi$ and define for $\lambda>0$
\[
\varphi_{\br,\eta}(\lambda)\coloneqq\lambda \gamma(\bpi, \br/\lambda)+\eta\lambda.
\]
Then from Lemma \ref{lem:perspective} $\Phi_\eta(\br)=\inf_{\lambda>0}\varphi_{\br,\eta}(\lambda)$, and from Proposition \ref{prop:egr.variational},
\begin{align}\label{eq:egr.var.rep}
\gamma(\bpi,\br/\lambda)&=\langle \bq(\br/\lambda)-\bpi,\br/\lambda \rangle-H\divg{\bq(\br/\lambda)}{\bpi}\\
&=\langle \nabla_{\br} \gamma(\bpi,\br/\lambda),\br/\lambda \rangle-H\divg{\bq(\br/\lambda)}{\bpi}.\nonumber
\end{align}
Differentiating and using \eqref{eq:egr.var.rep} gives
\begin{align*}
\varphi'_{\br,\eta}(\lambda)
&= \gamma(\bpi,\br/\lambda)
   -\Bigl\langle \nabla_{\br} \gamma(\bpi,\br/\lambda),\br/\lambda\Bigr\rangle
   +\eta
 = \eta-H\divg{\bq(\br/\lambda)}{\bpi},\\
\varphi''_{\br,\eta}(\lambda)
&= -\partial_\lambda H\divg{\bq(\br/\lambda)}{\bpi}.%
\end{align*}
Since $\gamma(\bpi,\cdot)$ is convex, so is its perspective $g_{\bpi}(\lambda, \br) \coloneqq\lambda \gamma(\bpi,\br/\lambda)$ (cf.\ \cite[Section 3.2.6]{boyd2004convex}). In particular, $\lambda\mapsto \lambda \gamma(\bpi,\br/\lambda)$, and by extension $\varphi_{\br,\eta}(\cdot)$, is convex. We may conclude that
\[-\partial_\lambda H\divg{\bq(\br/\lambda)}{\bpi}=\varphi''_{\br,\eta}(\lambda)\geq0\]
from which it follows that $\lambda\mapsto H\divg{\bq(\br/\lambda)}{\bpi}$ is decreasing.

As $\lambda\downarrow 0$, the distribution $\bq(\br/\lambda)$ concentrates on the set $\mathcal{M}_{\bpi}(\mathbf r)$ with the limit
\[
q_i(\br/\lambda)\rightarrow
\frac{\pi_i}{\sum_{j\in \mathcal{M}_{\bpi}(\mathbf r)}\pi_j}\,\mathbf 1_{\{i\in \mathcal{M}_{\bpi}(\mathbf r)\}}
\]
whereas $q_i(\br/\lambda)\rightarrow\pi_i$ as $\lambda\to\infty$.
Therefore,
\[
\lim_{\lambda\downarrow 0}H\divg{\bq(\br/\lambda)}{\bpi}
= -\log\!\Bigl(\sum_{j\in \mathcal{M}_{\bpi}(\mathbf r)} \pi_j\Bigr)=:\overline{\eta}(\mathbf r),\qquad
\lim_{\lambda\uparrow\infty}H\divg{\bq(\br/\lambda)}{\bpi}=0 .
\]

We consider now two cases. If $\br$ is constant on the support of $\bpi$ then $\overline{\eta}(\br)=0$ and $\bq(\br/\lambda) \equiv \bpi$. In particular, $\varphi'(\lambda)=\eta\geq0$ for all $\lambda>0$ and the minimum can be attained by sending $\lambda\downarrow 0$. 

Suppose instead that $\br$ is not constant on the support of $\bpi$ so that $\overline{\eta}(\br)>0$. Because $\varphi'(\lambda)=\eta-H\divg{\bq(\br/\lambda)}{\bpi}$ and $H\divg{\bq(\br/\lambda)}{\bpi}\in[0,\overline{\eta}(\br)]$ decreases in $\lambda$, there exists a $\lambda^\star\geq 0$ with $\varphi'(\lambda^\star)=0$ (equivalently, $H\divg{\bq(\br/\lambda^\star)}{\bpi}=\eta$) if and only if $0\le \eta<\overline{\eta}(\br)$. If $\eta\ge \overline{\eta}(\mathbf r)$ then $\varphi'(\lambda)\ge 0$ for all $\lambda>0$ and the infimum of $\varphi$ is achieved in the limit $\lambda\downarrow 0$. In this case, $\Phi_\eta(\mathbf r)=\max_{i\in\supp(\bpi)} r_i-\langle\bpi,\br\rangle$, and any $\bq^\star$ supported on $\mathcal{M}_{\bpi}(\mathbf r)$ with
$H\divg{\bq^\star}{\bpi}\le \eta$ is optimal.
\end{proof}

\begin{remark}
    It is not hard to check that the solution $\lambda^\star$ to $H\divg{\bq(\br/\lambda^\star)}{\bpi}=\eta$ in Lemma \ref{lem:minimizer.perspective.duality}(a) is unique as long as $\br$ is not constant on $\supp(\bpi)$. 
\end{remark}

We can now use the connections developed in Lemmas \ref{lem:perspective}--\ref{lem:minimizer.perspective.duality} with the excess growth rate maximization in Theorem \ref{thm:pairwise-max} to solve the constrained optimization problem in \eqref{eqn:constr.var.rep}.  We formalize this final link in the following proposition.

\begin{proposition}
Let $\br\in\mathbb{R}^n$ have $n\geq2$ distinct coordinates. Then for any $\eta\geq0$
\begin{align}
 &\sup_{(\bpi,\bq)\in\Delta_n\times\Delta_n:\ H\divg{\bq}{\bpi}\le \eta}\langle \bq-\bpi,\br\rangle  \nonumber\\
& =\inf_{\lambda>0}\Bigg\{\lambda\log\!\Bigl(\tfrac{e^{r_{(n)}/\lambda}-e^{r_{(1)}/\lambda}}{r_{(n)}-r_{(1)}}\Bigr)
-\frac{e^{r_{(n)}/\lambda}r_{(1)}-e^{r_{(1)}/\lambda}r_{(n)}}{e^{r_{(n)}/\lambda}-e^{r_{(1)}/\lambda}}+\lambda\log\lambda+\lambda(\eta-1)\Bigg\}.\label{eqn:constr.joint.opt}
\end{align}
Moreover, for the (unique) indices $\{i^\star,j^\star\}$ that attain the maximum and minimum of $\br$, respectively, define the pair $(\bpi(\lambda), \bq(\lambda))$ by 
    \[
\pi_{i^\star}(\lambda)
=
\frac{\,e^{r_{i^\star}/\lambda}-e^{r_{j^\star}/\lambda}-\lambda^{-1}(r_{i^\star}-r_{j^\star})e^{r_{j^\star}/\lambda}\,}
{\lambda^{-1}(r_{i^\star}-r_{j^\star})\,(e^{r_{i^\star}/
\lambda}-e^{r_{j^\star}/\lambda})},
\quad
\pi_{j^\star}(\lambda)=1-\pi_{i^\star}(\lambda),
\]
$\pi_k(\lambda)=0$ for $(k\notin\{i^\star,j^\star\})$, and $\bq(\lambda) = \bpi(\lambda) \oplus_{\bpi(\lambda)} \cC[e^{\br/\lambda}]$. For any $\lambda^{\star}$ solving\footnote{We are assured of the existence of at least one solution. If $\eta=0$ we identify the solution with the limit $\lim_{\lambda\uparrow\infty}\bpi(\lambda)=\lim_{\lambda\uparrow\infty}\bq(\lambda)=\frac{1}{2}(\be_{i^\star}+\be_{j^\star})$.}
\[H\divg{\bq(\lambda^\star)}{\bpi(\lambda^\star)}=\eta\]
the choice $\bpi^\star=\bpi(\lambda^\star)$ and $\bq^\star=\bq(\lambda^\star)$ is optimal.
\end{proposition}

\begin{proof}
    By using Lemma \ref{lem:perspective} we can rewrite the constrained joint maximization of \eqref{eqn:constr.var.rep} as
\begin{align*}
    \sup_{(\bpi,\bq)\in\Delta_n\times\Delta_n:\ H(\bq\;\Vert\;\bpi)\le \eta}\left\langle\bq-\bpi,\br\right\rangle&=\sup_{\bpi\in\Delta_n}\sup_{\bq\in\Delta_n:\ H(\bq\;\Vert\;\bpi)\le \eta}\, \langle \bq-\bpi,\br\rangle\\
    &=\sup_{\bpi\in\Delta_n}\inf_{\lambda> 0}\Bigl\{\lambda\, \gamma\Bigl(\bpi, \br/\lambda\Bigr)+\eta\,\lambda\Bigr\}\\
    &=\inf_{\lambda> 0}\Bigl\{\lambda\, \sup_{\bpi\in\Delta_n}\gamma\Bigl(\bpi,\mathbf r/\lambda\Bigr)+\eta\,\lambda\Bigr\}.
\end{align*}
The interchange of $\inf\{\dots\}$ and $\sup\{\dots\}$ is justified by Sion's minimax theorem as $\Delta_n$ is convex and compact, $(0,\infty)$ is convex, $\bpi\mapsto\gamma(\bpi,\br)$ is concave, and $ \lambda \mapsto \lambda\gamma
(\bpi,\mathbf{r}/\lambda)$ is convex (cf.~\cite[Section 3.2.6]{boyd2004convex}).

To see the characterization of the solution, we observe from \eqref{eqn:penalized.joint.opt} and Proposition \ref{prop:joint.optim.penalized.sol} that
\begin{align*}\lambda\sup_{\bpi\in\Delta_n}\gamma(\bpi,\br/\lambda) &= \sup_{(\bpi,\bq)\in\Delta_n\times\Delta_n}\Bigl\{\langle  \bq - \bpi,\br\rangle-\lambda\,H\divg{\bq}{\bpi}\Bigr\}\\
&=  \langle  \bq(\lambda) - \bpi(\lambda),\br\rangle-\lambda\,H\divg{\bq(\lambda)}{\bpi(\lambda)}
\end{align*}
Proposition \ref{prop:joint.optim.penalized.sol} also recovers \eqref{eqn:constr.joint.opt}. 
Define the functions
\[g(\bpi, \lambda)\coloneqq\lambda\gamma\Bigl(\bpi,\br/\lambda\Bigr), \quad f(\lambda)\coloneqq \sup_{\bpi\in\Delta_n}g(\bpi,\lambda)+\eta\,\lambda\]
so that 
\[ \sup_{(\bpi,\bq)\in\Delta_n\times\Delta_n:\ H(\bq\;\Vert\;\bpi)\le \eta}\left\langle\bq-\bpi,\br\right\rangle=\inf_{\lambda> 0}f(\lambda).\]
Observe that since $\lambda\mapsto \lambda\gamma\Bigl(\bpi,\br/\lambda\Bigr)$ is convex for each $\bpi$, $g(\bpi,\cdot)$ is convex and (as a maximum of convex functions) so is $f(\lambda)$. Moreover, since $\br$ has distinct entries we are assured that $\bpi(\lambda)$ is the unique optimizer. Since $\Delta_n$ is compact, by Danskin's Theorem (cf.~\cite[Proposition A.3.2]{bertsekas2009convex})
\[f'(\lambda) = \partial_\lambda g(\bpi(\lambda),\lambda)+\eta.\]
Repeating the argument in Lemma \ref{lem:minimizer.perspective.duality} we see that
\[\partial_\lambda g(\bpi(\lambda),\lambda)=-H\divg{\bq(\lambda)}{\bpi(\lambda)}.\]
Thus, to minimize $f(\lambda)$ we search for a solution $\lambda^\star$ of $f'(\lambda^\star)=0$, which is equivalently given by the solution to
\[H\divg{\bq(\lambda^\star)}{\bpi(\lambda^\star)}=\eta.\]
Moreover, since $f(\lambda)$ is convex we have $f''(\lambda)\geq0$ and we conclude,
\[0\geq-f''(\lambda)=\partial_\lambda H\divg{\bq(\lambda)}{\bpi(\lambda)}.\]
That is, $\lambda \mapsto H\divg{\bq(\lambda)}{\bpi(\lambda)}$ is continuous and decreasing. As $\lambda\uparrow\infty$ we see that $\bpi(\lambda)\to \frac{1}{2}(\be_{i^\star}+\be_{j^\star})$, while as $\lambda\downarrow 0$, $\bpi(\lambda)\to \be_{j^\star}$. In the limit $\lambda\uparrow \infty$ we find that similarly $\bq(\lambda)\to\frac{1}{2}(\be_{i^\star}+\be_{j^\star})$ while as $\lambda\downarrow 0$ we have $\bq(\lambda)\to \be_{i^\star}$. Hence
\[\lim_{\lambda\downarrow0} H\divg{\bq(\lambda)}{\bpi(\lambda)}=\infty, \qquad \lim_{\lambda\uparrow\infty} H\divg{\bq(\lambda)}{\bpi(\lambda)}=0. \]
We conclude that there must be a solution for any $0\leq \eta<\infty$ which completes the proof.
\end{proof}

\subsection{Superdifferential set for the expected excess growth rate}

\begin{definition}\label{def:superdifferential.J}
For $\bpi\in\Delta_n$, the superdifferential set of $J$ at $\bpi$ relative to $\Delta_n$ is
\[
\partial_{\Delta_n}^+ J(\bpi)
\coloneqq\bigl\{\bq\in\R^n:\ J(\bpi')-J(\bpi)\le \langle \bq,\bpi'-\bpi\rangle\ \ \forall\,\bpi'\in\Delta_n\bigr\}.
\]
\end{definition}

\begin{lemma}\label{lem:superdifferential.set} Under Assumption \ref{ass:int} 
    $\bg\in \partial_{\Delta_n}^{+} J(\bpi)$ if and only if there exists $\lambda\in \mathbb{R}$ and a $\bmu\in\mathbb{R}^n_+$ with $\mu_i=0$ on $\supp(\bpi)$ such that\footnote{Note that Assumption \ref{ass:int} is not sufficient to guarantee that the displayed expectation has finite
coordinates. However, its coordinates are always non-negative and therefore well defined.}
\[\bg=\mathbb{E}\left[\frac{\bR}{\langle \bpi, \bR\rangle }\right]-\bm-\lambda \mathbf{1}+\bmu\in \mathbb{R}^n.\]
In particular, $\partial_{\Delta_n}^+ J(\bpi)\not=\emptyset$ if and only if $\mathbb{E}\left[\frac{R_i}{\langle \bpi, \bR\rangle }\right]<\infty$ for all $i\in\{1,\dots,n\}$.
\end{lemma}
\begin{proof}
Consider the normal cone to the simplex at $\bpi\in\Delta_n$
\begin{align*}
N_{\Delta_n}(\bpi)&\coloneqq\{\bv\in\mathbb{R}^n:\ \langle \bv,\bpi'-\bpi\rangle\le 0\ \ \forall\,\bpi'\in \Delta_n\}.
\end{align*}
Any $\bv\in N_{\Delta_n}(\bpi)$ admits a representation $\bv=\lambda \mathbf{1} -\bmu$ where $\lambda\in\mathbb{R}$, $\mu_i=0$ if $i\in\supp(\bpi)$ and $\mu_i\geq0$ otherwise. We observe that if $\bg\in \partial_{\Delta_n}^+ J(\bpi)$ and $\bv\in  N_{\Delta_n}(\bpi)$ then $(\bg-\bv)\in \partial_{\Delta_n}^+ J(\bpi)$.

We now search for a particular element of the supergradient set. We begin with the assumption that $\mathbb{E}\left[\frac{R_i}{\langle \bpi, \bR\rangle }\right]<\infty$ for all $i$ otherwise the claimed form of $\bg$ cannot be a member of $\partial_{\Delta_n}^+ J(\bpi)$. For $x,y>0$, $\log y-\log x\le \frac{y-x}{x}$. With
$x=\langle\bpi,\bR\rangle$ and $y=\langle\bpi',\bR\rangle$ this yields
\[
\log\langle\bpi',\bR\rangle-\log\langle\bpi,\bR\rangle
\le 
\frac{\langle\bpi'-\bpi,\bR\rangle}{\langle\bpi,\bR\rangle}=\left\langle\bpi'-\bpi,\frac{\bR}{\langle\bpi,\bR\rangle}\right\rangle.
\]
Taking expectations and adding the remaining linear term we conclude
\begin{equation}\label{eqn:J.supergrad}J(\bpi')-J(\bpi) \leq \left\langle\bpi'-\bpi,\mathbb{E}\left[\frac{\bR}{\langle \bpi, \bR\rangle }\right]-\bm\right\rangle. 
\end{equation}
Define 
\[\bg^{\star}(\bpi) = \mathbb{E}\left[\frac{\bR}{\langle \bpi, \bR\rangle }\right]-\bm.\]
By \eqref{eqn:J.supergrad} $\bg^\star(\bpi)\in \partial_{\Delta_n}^+ J(\bpi)$ (and also $(\bg^{\star}-\bv)\in \partial_{\Delta_n}^+J(\bpi)$ for $\bv\in N_{\Delta_n}(\bpi)$).

Next, we argue that on the relative interior of any face of $\Delta_n$ where $\pi_i>0$ the $i$th coordinate of $\bg^\star(\bpi)$ defines the partial derivative. Here and in what follows we make regular use of the inequality
\[0\leq \frac{R_i}{\langle\bpi,\bR \rangle}\leq \frac{1}{\pi_i}\]
when $\pi_i>0$. Combining this with the inequality $|\log(1+x)|\leq |x|/(1-|x|)$ for $x\in(-1,1)$ we have for all $h\in\mathbb{R}$ with $0<|h|\leq \pi_i/2$ (since $R_i>0$ and $\langle\bpi,\bR\rangle>0$)
\begin{align*}
    &\left| \frac{\log\langle \bpi + h \be_i, \bR\rangle-\log\langle \bpi, \bR\rangle}{h}\right|
    =\left|\frac{1}{h}\log\left(1+h\frac{R_i}{\langle \bpi,\bR\rangle}\right)\right|\\
    &\leq \frac{\frac{R_i}{\langle \bpi,\bR\rangle}}{1-|h|\frac{R_i}{\langle \bpi,\bR\rangle}}
    \leq \frac{\frac{1}{\pi_i}}{1-|h|\frac{1}{\pi_i}}\leq \frac{2}{\pi_i}<\infty.
\end{align*}
The second inequality follows from the monotonicity of $x\mapsto x/(1-x)$ on $(-\infty,1)$. So, by the dominated convergence theorem,
\[\partial_{\pi_i}J(\bpi) = \mathbb{E}\left[\frac{R_i}{\langle \bpi, \bR\rangle }\right]-m_i.\]

We claim that if $\bg\in \partial_{\Delta_n}^+ J(\bpi)$ then $\bg = \bg^\star(\bpi)-\bv$ for some $\bv\in N_{\Delta_n}(\bpi)$. Since we necessarily have that the coordinates of $\bg$ are finite, if this were true the expectations in $\mathbf{g}^\star(\bpi)$ would also have to be finite. For a set of ``active'' indices $S\subset\{1,\dots,n\}$ we define the face
\[\Delta_{S} \coloneqq\{\bpi'\in \Delta_n: \pi_i'=0 \ \forall i\not\in S\}\]
and the relative interior of the face,
\[\mathrm{ri}(\Delta_{S}) \coloneqq\{\bpi'\in \Delta_n: \pi_i'=0 \ \forall i\not\in S \ \text{and} \  \pi_i'>0 \ \forall i\in S\}.\]
For fixed $\bpi$, choose $S=\supp(\bpi)$ so $\bpi\in \mathrm{ri}(\Delta_{S})$. We define the tangent space to $\Delta_{S}$ (embedded in $\mathbb{R}^n$) at this $\bpi$ as 
\[T_{S}(\bpi)\coloneqq\left\{\bt \in \mathbb{R}^n: t_i=0 \ \forall i\not\in S, \ \sum_{i=1}^nt_i=0\right\}.\]
Let $\bt\in T_{S}(\bpi)$ and $\bg \in \partial_{\Delta_n}^+ J(\bpi)$. Then, for sufficiently small $\epsilon>0$, $\bpi+ \epsilon \bt \in \Delta_S$ and so,
\[J(\bpi + \epsilon\bt) - J(\bpi) \leq \langle \bg , \epsilon\bt \rangle.\]
Dividing by $\epsilon$ and sending $\epsilon\downarrow 0$, we have (by using the differentiability of $J(\bpi)$ on the relative interior),
\[\langle \bg^{\star}(\bpi),\bt\rangle \leq \langle \bg , \bt \rangle.\]
Repeating the argument for $-\bt\in T_{S}(\bpi)$ we have
\[-\langle \bg^{\star}(\bpi),\bt\rangle \leq -\langle \bg , \bt \rangle.\]
Taken together $\langle \bg^{\star}(\bpi)-\bg,\bt\rangle = 0$. But this implies that $\bg-\bg^{\star}(\bpi)$ is orthogonal to every $\bt\in T_S(\bpi)$. In particular, for the coordinates $i\in S$ we must have $g_i= g^{\star}_i(\bpi) - \lambda$ for some $\lambda\in\mathbb{R}$. 

With this characterization of the coordinates in $S$, consider the perturbation $\bt=\be_k-\be_j$ for $k\not\in S$ and $j\in S$. Once more, for sufficiently small $\epsilon>0$ we have that $\bpi+ \epsilon \bt \in \Delta_n$. It follows that
\[J(\bpi + \epsilon\bt) - J(\bpi) \leq \langle \bg , \epsilon\bt \rangle.\]
Dividing by $\epsilon$ and sending $\epsilon\downarrow 0$ we have that
\begin{equation}\label{eqn:lim.sup.est}
\limsup_{\epsilon\downarrow 0}\frac{J(\bpi + \epsilon\bt) - J(\bpi)}{\epsilon} \leq  g_k - g_j = g_k - g_j^{\star}(\bpi) + \lambda.
\end{equation}
At the same time, we may apply the inequality $\log(1+x)\geq x/(1+x)$ for $x\in (-1,\infty)$ to conclude that for sufficiently small $\epsilon>0$
\begin{align*}
\log\langle\bpi+\epsilon\bt,\bR\rangle-\log\langle\bpi,\bR\rangle &= \log\left(\langle\bpi,\bR\rangle+\epsilon (R_k-R_j)\right)-\log\langle\bpi,\bR\rangle\\
&= \log\left(1+\frac{\epsilon (R_k-R_j)}{\langle\bpi,\bR\rangle}\right)\\
&\geq \frac{\frac{\epsilon (R_k-R_j)}{\langle\bpi,\bR\rangle}}{1+\frac{\epsilon (R_k-R_j)}{\langle\bpi,\bR\rangle}}.
\end{align*}
That is,
\[\liminf_{\epsilon\downarrow0}\frac{\log\langle\bpi+\epsilon\bt,\bR\rangle-\log\langle\bpi,\bR\rangle}{\epsilon}\geq \frac{(R_k-R_j)}{\langle\bpi,\bR\rangle}.\]
Moreover since $R_k\geq0$, and $j\in \supp(\bpi)$ we have that
\[\frac{\epsilon (R_k-R_j)}{\langle\bpi,\bR\rangle}\geq -\frac{\epsilon R_j}{\langle\bpi,\bR\rangle}\geq -\frac{\epsilon}{\pi_j}.\]
Specifically, as $x\mapsto x/(1+x)$ is increasing on $(-1,\infty)$ the following inequalities hold for all $\epsilon\leq \pi_j/2$:
\[\frac{1}{\epsilon}\left(\log\langle\bpi+\epsilon\bt,\bR\rangle-\log\langle\bpi,\bR\rangle \right)\geq \frac{\frac{ (R_k-R_j)}{\langle\bpi,\bR\rangle}}{1+\frac{\epsilon (R_k-R_j)}{\langle\bpi,\bR\rangle}}\geq \frac{-\frac{1}{\pi_j}}{1-\frac{\epsilon}{\pi_j}} \geq -\frac{2}{\pi_j}.\]
This supplies a uniform lower bound, so by taking expectations and applying Fatou's lemma we have
\begin{equation}\label{eqn:lim.inf.est}
\liminf_{\epsilon\downarrow 0}\frac{J(\bpi + \epsilon\bt) - J(\bpi)}{\epsilon} \geq \mathbb{E}\left[\frac{(R_k-R_j)}{\langle\bpi,\bR\rangle}\right]-(m_k-m_j)=g^\star_k(\bpi)-g^{\star}_j(\bpi).
\end{equation}
Combining our estimates \eqref{eqn:lim.sup.est} and \eqref{eqn:lim.inf.est} we have\footnote{It is clear here that $\bg^\star(\bpi)$ must be finite since $-m_j\leq g_j^\star(\bpi)\leq 1/\pi_j - m_j$ for all $j\in S$, and $-m_k\leq g_k^\star(\bpi)\leq g_k +\lambda$ for $k\not\in S$.}
\[g^\star_k(\bpi)-g^{\star}_j(\bpi)\leq g_k - g_j^{\star}(\bpi) + \lambda.\]
Equivalently, $g^\star_k(\bpi)-\lambda\leq g_k$.
Letting $\mu_k\coloneqq g_k-(g^\star_k(\bpi)-\lambda)\geq0$ for all $k\not\in S$ and $\mu_j=0$ for $j\in S$ recovers the claimed representation for any $\bg\in \partial_{\Delta_n}^+ J(\bpi)$.
\end{proof}

\section*{Acknowledgment}
S.~Campbell acknowledges support from an NSERC Postdoctoral Fellowship (PDF-599675-2025) and a CDFT Research Grant. T.-K.~L.~Wong acknowledges support from the NSERC Discovery Grant RGPIN-2025-06021. The authors thank Martin Larsson, Johannes Ruf and Ruodu Wang for helpful comments. T.-K.~L.~Wong would also like to thank Soumik Pal with whom many important ideas in this paper, including the first chain rule of the excess growth rate (Proposition \ref{prop:chain.rule.1}), the logarithmic divergence \eqref{eqn:L.divergence} and large deviations of the Dirichlet perturbation (see Remark \ref{rmk:equal.weights}), were first developed.

\bibliographystyle{abbrv}
\bibliography{references}

@book{leinster2021entropy,
  title={Entropy and Diversity: The Axiomatic Approach},
  author={Leinster, Tom},
  year={2021},
  publisher={Cambridge University Press}
}

@article{kannappan1983generalized,
  title={On a generalized fundamental equation of information},
  author={Kannappan, Pl and Ng, CT},
  journal={Canadian Journal of Mathematics},
  volume={35},
  number={5},
  pages={862--872},
  year={1983},
  publisher={Cambridge University Press}
}

@article{ebanks1987generalized,
  title={Generalized fundamental equation of information of multiplicative type},
  author={Ebanks, BR and Kannappan, Pl and Ng, CT},
  journal={Aequationes Math},
  volume={32},
  number={1},
  pages={19--31},
  year={1987}
}

@article{csiszar2008axiomatic,
  title={Axiomatic characterizations of information measures},
  author={Csisz{\'a}r, Imre},
  journal={Entropy},
  volume={10},
  number={3},
  pages={261--273},
  year={2008},
  publisher={Molecular Diversity Preservation International}
}

@book{kannappan2009functional,
  title={Functional Equations and Inequalities with Applications},
  author={Kannappan, Palaniappan},
  year={2009},
  publisher={Springer Science \& Business Media}
}

@article{campbell1965coding,
  title={A coding theorem and {R}{\'e}nyi's entropy},
  author={Campbell, L Lorne},
  journal={Information and Control},
  volume={8},
  number={4},
  pages={423--429},
  year={1965},
  publisher={Elsevier}
}

@article{wong2022tsallis,
  title={{T}sallis and {R}{\'e}nyi deformations linked via a new $\lambda$-duality},
  author={Wong, Ting-Kam Leonard and Zhang, Jun},
  journal={IEEE Transactions on Information Theory},
  volume={68},
  number={8},
  pages={5353--5373},
  year={2022},
  publisher={IEEE}
}

@article{campbell2025efficient,
  title={Efficient convex {PCA} with applications to {W}asserstein {GPCA} and ranked data},
  author={Campbell, Steven and Wong, Ting-Kam Leonard},
  journal={Journal of Computational and Graphical Statistics},
  volume={34},
  number={2},
  pages={540--551},
  year={2025},
  publisher={Taylor \& Francis}
}

@article{pal2020multiplicative,
  title={Multiplicative {S}chr{\"o}dinger problem and the {D}irichlet transport},
  author={Pal, Soumik and Wong, Ting-Kam Leonard},
  journal={Probability Theory and Related Fields},
  volume={178},
  number={1},
  pages={613--654},
  year={2020},
  publisher={Springer}
}

@article{kelly1956new,
  title={A new interpretation of information rate},
  author={Kelly, John L},
  journal={The Bell System Technical Journal},
  volume={35},
  number={4},
  pages={917--926},
  year={1956},
  publisher={Nokia Bell Labs}
}

@book{bertsekas2009convex,
  title={Convex optimization theory},
  author={Bertsekas, Dimitri},
  volume={1},
  year={2009},
  publisher={Athena Scientific}
}

@article{BF92,
  title={Diversification returns and asset contributions},
  author={Booth, David G and Fama, Eugene F},
  journal={Financial Analysts Journal},
  volume={48},
  number={3},
  pages={26--32},
  year={1992},
  publisher={Taylor \& Francis}
}

@article{PW13,
	title={Energy, entropy, and arbitrage},
	author={Pal, Soumik and Wong, Ting-Kam Leonard},
	journal={arXiv preprint arXiv:1308.5376},
	year={2013}
}

@incollection{A94,
  title={Principles of Compositional Data Analysis},
  author={Aitchison, John},
  booktitle={Lecture Notes-Monograph Series},
  pages={73--81},
  year={1994},
  publisher={ Institute of Mathematical Statistics}
}

@article{PW18,
  title={Exponentially concave functions and a new information geometry},
  author={Pal, Soumik and Wong, Ting-Kam Leonard},
  journal={The Annals of Probability},
  volume={46},
  number={2},
  pages={1070--1113},
  year={2018},
  publisher={Institute of Mathematical Statistics}
}

@book{CT06,
  title={Elements of Information Theory},
  author={Cover, Thomas M and Thomas, Joy A},
  edition={second},
  year={2006},
  publisher={John Wiley \& Sons}
}

@article{algoet1988asymptotic,
  title={Asymptotic optimality and asymptotic equipartition properties of log-optimum investment},
  author={Algoet, Paul H and Cover, Thomas M},
  journal={The Annals of Probability},
  pages={876--898},
  year={1988},
  publisher={JSTOR}
}

@article{PW16,
	title={The geometry of relative arbitrage},
	author={Pal, Soumik and Wong, Ting-Kam Leonard},
	journal={Mathematics and Financial Economics},
	volume={10},
	number={3},
	pages={263--293},
	year={2016},
	publisher={Springer}
}

@article{bercher2009source,
  title={Source coding with escort distributions and {R}{\'e}nyi entropy bounds},
  author={Bercher, J-F},
  journal={Physics Letters A},
  volume={373},
  number={36},
  pages={3235--3238},
  year={2009},
  publisher={Elsevier}
}

@book{F02,
	title={Stochastic Portfolio Theory},
	author={Fernholz, {E.} {R.}},
	year={2002},
	publisher={Springer},
}

@article{W18,
  title={Logarithmic divergences from optimal transport and {R}\'{e}nyi geometry},
  author={Wong, Ting-Kam Leonard},
  journal={Information Geometry},
  volume={1},
  number={1},
  pages={39--78},
  year={2018},
  publisher={Springer}
}

@article{larsson2025numeraire,
  title={The numeraire $e$-variable and reverse information projection},
  author={Larsson, Martin and Ramdas, Aaditya and Ruf, Johannes},
  journal={The Annals of Statistics},
  volume={53},
  number={3},
  pages={1015--1043},
  year={2025},
  publisher={Institute of Mathematical Statistics}
}

@inproceedings{breiman1961optimal,
  title={Optimal gambling systems for favorable games},
  author={Breiman, L},
  booktitle={Proceedings of the Fourth Berkeley Symposium on Mathematical Statistics and Probability, Volume 1: Contributions to the Theory of Statistics},
  volume={4},
  pages={65--79},
  year={1961},
  organization={University of California Press}
}

@article{shannon1948mathematical,
  title={A mathematical theory of communication},
  author={Shannon, Claude E},
  journal={The Bell System Technical Journal},
  volume={27},
  number={3},
  pages={379--423},
  year={1948},
  publisher={Nokia Bell Labs}
}

@book{maclean2011kelly,
  title={The Kelly Capital Growth Investment Criterion: Theory and Practice},
  author={MacLean, Leonard C and Thorp, Edward O and Ziemba, William T},
  year={2011},
  publisher={World Scientific}
}

@book{aczel1975measures,
  title={On measures of information and their characterizations},
  author={Acz{\'e}l, J{\'a}nos and Dar{\'o}czy, Zolt{\'a}n},
  series={Mathematics in Science and Engineering},
  volume={115},
  publisher={Academic Press},
  year={1975}
}

@article{qi2023geometric,
  title={Geometric characterization of maximum diversification return portfolio via {R}ao’s quadratic entropy},
  author={Qi, Hou-Duo},
  journal={SIAM Journal on Financial Mathematics},
  volume={14},
  number={2},
  pages={525--556},
  year={2023},
  publisher={SIAM}
}

@article{campbell2024macroscopic,
  title={Macroscopic properties of equity markets: stylized facts and portfolio performance},
  author={Campbell, Steven and Song, Qien and Wong, Ting-Kam Leonard},
  journal={Quantitative Finance},
  volume={25},
  number={9},
  pages={1375--1397},
  year={2025},
  publisher={Taylor \& Francis}
}

@article{FS82,
	title={Stochastic portfolio theory and stock market equilibrium},
	author={Fernholz, Robert and Shay, Brian},
	journal={The Journal of Finance},
	volume={37},
	number={2},
	pages={615--624},
	year={1982},
	publisher={JSTOR}
}

@incollection{W19,
	title={Information Geometry in Portfolio Theory},
	author={Wong, Ting-Kam Leonard},
	booktitle={Geometric Structures of Information},
	pages={105--136},
	year={2019},
	publisher={Springer}
}

@book{A16,
  title={Information Geometry and Its Applications},
  author={Amari, Shun-Ichi},
  year={2016},
  publisher={Springer}
}

@article{egozcue2003isometric,
  title={Isometric logratio transformations for compositional data analysis},
  author={Egozcue, Juan Jos{\'e} and Pawlowsky-Glahn, Vera and Mateu-Figueras, Gl{\`o}ria and Barcelo-Vidal, Carles},
  journal={Mathematical Geology},
  volume={35},
  number={3},
  pages={279--300},
  year={2003},
  publisher={Springer}
}

@article{ramdas2024hypothesis,
  title={Hypothesis testing with e-values},
  author={Ramdas, Aaditya and Wang, Ruodu},
  journal={Foundations and Trends{\textregistered} in Statistics},
  volume={1},
  number={1-2},
  pages={1--390},
  year={2025},
  publisher={Emerald Publishing Limited}
}

@article{fernholz2005relative,
  title={Relative arbitrage in volatility-stabilized markets},
  author={Fernholz, Robert and Karatzas, Ioannis},
  journal={Annals of Finance},
  volume={1},
  number={2},
  pages={149--177},
  year={2005},
  publisher={Springer}
}

@incollection{FK09,
	title = {Stochastic Portfolio Theory: an Overview},
	author = {Fernholz, Robert and Karatzas, Ioannis},
	editor = {Ciarlet, P. G.},
	booktitle = {Handbook of Numerical Analysis},
	publisher = {Elsevier},
	year = {2009},
	volume = {15},
	pages = {89--167}
}

@article{mantilla2022can,
  title={Can the portfolio excess growth rate explain the predictive power of idiosyncratic volatility?},
  author={Mantilla-Garcia, Daniel and Malagon, Juliana and Aldana-Galindo, Julian R},
  journal={Finance Research Letters},
  volume={47},
  pages={102577},
  year={2022},
  publisher={Elsevier}
}

@article{maeso2020maximizing,
  title={Maximizing an equity portfolio excess growth rate: a new form of smart beta strategy?},
  author={Maeso, Jean-Michel and Martellini, Lionel},
  journal={Quantitative Finance},
  volume={20},
  number={7},
  pages={1185--1197},
  year={2020},
  publisher={Taylor \& Francis}
}

@book{boyd2004convex,
  title={Convex optimization},
  author={Boyd, Stephen P and Vandenberghe, Lieven},
  year={2004},
  publisher={Cambridge University Press}
}

@article{willenbrock2011diversification,
  title={Diversification return, portfolio rebalancing, and the commodity return puzzle},
  author={Willenbrock, Scott},
  journal={Financial Analysts Journal},
  volume={67},
  number={4},
  pages={42--49},
  year={2011},
  publisher={Taylor \& Francis}
}

@incollection{MME21,
  title={Distributions on the simplex revisited},
  author={Mateu-Figueras, Gloria and Monti, Gianna S and Egozcue, JJ},
  booktitle={Advances in Compositional Data Analysis: Festschrift in Honour of Vera Pawlowsky-Glahn},
  pages={61--82},
  year={2021},
  publisher={Springer}
}

@book{D09,
  title={Large Deviations: Techniques and Applications},
  author={Dembo, Amir and Zeitouni, Ofer},
  year={2009},
  publisher={Springer}
}

@article{M2018,
  title={On the generalized distance in statistics (Reprint)},
  author={Mahalanobis, Prasanta Chandra},
  journal={Sankhy{\=a}: The Indian Journal of Statistics, Series A},
  volume={80},
  pages={S1--S7},
  year={2018},
}

@article{VH14,
	title={R{\'e}nyi divergence and {K}ullback-{L}eibler divergence},
	author={Van Erven, Tim and Harremos, Peter},
	journal={IEEE Transactions on Information Theory},
	volume={60},
	number={7},
	pages={3797--3820},
	year={2014},
	publisher={IEEE}
}

@inproceedings{MMPE11,
  title={The shifted-scaled {D}irichlet distribution in the simplex},
  author={Monti, Gianna Serafina and Mateu-Figueras, G and Pawlowsky-Glahn, Vera and Egozcue, Juan Jos{\'e}},
  booktitle={Proceedings of the 4th International Workshop on Compositional Data Analysis},
  year={2011}
}

@article{D68,
  title={Three multidimensional-integral identities with {B}ayesian applications},
  author={Dickey, James M},
  journal={The Annals of Mathematical Statistics},
  pages={1615--1628},
  year={1968},
  publisher={JSTOR}
}

@article{BNPS12,
  title={Volatility harvesting: Why does diversifying and rebalancing create portfolio growth},
  author={Bouchey, Paul and Nemtchinov, Vassilii and Paulsen, Alex and Stein, David M},
  journal={The Journal of Wealth Management},
  volume={15},
  number={2},
  pages={26--35},
  year={2012},
  publisher={Portfolio Management Research}
}

@book{folland1999real,
  title={Real analysis: modern techniques and their applications},
  author={Folland, Gerald B},
  year={1999},
  publisher={John Wiley \& Sons}
}

@article{BNW15,
  title={Volatility harvesting in theory and practice},
  author={Bouchey, Paul and Nemtchinov, Vassilii and Wong, Ting-Kam Leonard},
  journal={The Journal of Wealth Management},
  volume={18},
  number={3},
  pages={89-100},
  year={2015},
  publisher={Pageant Media}
}

@article{OJ23,
  title={Tight concentrations and confidence sequences from the regret of universal portfolio},
  author={Orabona, Francesco and Jun, Kwang-Sung},
  journal={IEEE Transactions on Information Theory},
  volume={70},
  number={1},
  pages={436--455},
  year={2023},
  publisher={IEEE}
}

@article{C91,
  author = {Cover, T. M.},
  title = {Universal Portfolios},
  journal = {Mathematical Finance},
  year = {1991},
  volume = {1},
  pages = {1--29},
  number = {1},
}

@article{A09,
  title={$\alpha$-divergence is unique, belonging to both $f$-divergence and {B}regman divergence classes},
  author={Amari, Shun-Ichi},
  journal={IEEE Transactions on Information Theory},
  volume={55},
  number={11},
  pages={4925--4931},
  year={2009},
  publisher={IEEE}
}

@inproceedings{WY19,
  title={Logarithmic divergences: geometry and interpretation of curvature},
  author={Wong, Ting-Kam Leonard and Yang, Jiaowen},
  booktitle={International Conference on Geometric Science of Information},
  pages={413--422},
  year={2019},
  organization={Springer}
}

@book{rockafellar1997convex,
  title={Convex Analysis},
  author={Rockafellar, R Tyrrell},
  year={1997},
  publisher={Princeton University Press}
}

@incollection{EA21,
  title={The information-geometric perspective of compositional data analysis},
  author={Erb, Ionas and Ay, Nihat},
  booktitle={Advances in Compositional Data Analysis: Festschrift in Honour of Vera Pawlowsky-Glahn},
  pages={21--43},
  year={2021},
  publisher={Springer}
}

@inproceedings{TWYZ25,
  title={Maximum likelihood estimation for the $\lambda$-exponential family},
  author={Tian, Xiwei and Wong, Ting-Kam Leonard and Yang, Jiaowen and Zhang, Jun},
  booktitle={International Conference on Geometric Science of Information},
  pages={102--111},
  year={2025},
  organization={Springer}
}

@article{RX20,
  title={The impact of proportional transaction costs on systematically generated portfolios},
  author={Ruf, Johannes and Xie, Kangjianan},
  journal={SIAM Journal on Financial Mathematics},
  volume={11},
  number={3},
  pages={881--896},
  year={2020},
  publisher={SIAM}
}

@article{FKR18,
  title={Volatility and arbitrage},
  author={Fernholz, E Robert and Karatzas, Ioannis and Ruf, Johannes},
  journal={The Annals of Applied Probability},
  volume={28},
  number={1},
  pages={378--417},
  year={2018},
  publisher={JSTOR}
}

@book{PB21,
  title={Statistical Mechanics},
  author={Pathria, Raj Kumar and Beale, Paul D.},
  edition={fourth},
  year={2021},
  publisher={Academic Press}
}

@article{FM07,
  title={The statistics of statistical arbitrage},
  author={Fernholz, Robert and Maguire Jr, Cary},
  journal={Financial Analysts Journal},
  volume={63},
  number={5},
  pages={46--52},
  year={2007},
  publisher={Taylor \& Francis}
}

@article{BFPRS19,
  title={Diversification, volatility, and surprising alpha},
  author={Banner, Adrian and Fernholz, Robert and Papathanakos, Vassilios and Ruf, Johannes and Schofield, David},
  journal={Journal of Investment Consulting},
  volume={19},
  number={1},
  pages={23--30},
  year={2019}
}

@article{B67,
  title={The relaxation method of finding the common point of convex sets and its application to the solution of problems in convex programming},
  author={Bregman, Lev M},
  journal={USSR Computational Mathematics and Mathematical Physics},
  volume={7},
  number={3},
  pages={200--217},
  year={1967},
  publisher={Elsevier}
}

@article{NBN07,
  title={Bregman Voronoi diagrams: properties, algorithms and applications},
  author={Nielsen, Frank and Boissonnat, Jean-Daniel and Nock, Richard},
  journal={arXiv preprint arXiv:0709.2196},
  year={2007}
}

@book{PW25,
  title={Information Theory: From Coding to Learning},
  author={Polyanskiy, Yury and Wu, Yihong},
  year={2025},
  publisher={Cambridge University Press}
}

@inproceedings{R61,
  title={On measures of entropy and information},
  author={R{\'e}nyi, Alfr{\'e}d},
  booktitle={Proceedings of the Fourth Berkeley Symposium on Mathematical Statistics and Probability, volume 1: Contributions to the Theory of Statistics},
  volume={4},
  pages={547--562},
  year={1961},
  organization={University of California Press}
}

@article{LC12,
  title={Measuring diversity: the importance of species similarity},
  author={Leinster, Tom and Cobbold, Christina A},
  journal={Ecology},
  volume={93},
  number={3},
  pages={477--489},
  year={2012},
  publisher={Wiley Online Library}
}

@article{F99,
  title={On the diversity of equity markets},
  author={Fernholz, Robert},
  journal={Journal of Mathematical Economics},
  volume={31},
  number={3},
  pages={393--417},
  year={1999},
  publisher={Elsevier}
}

@book{GMM09,
  title={Aggregation Functions},
  author={Grabisch, Michel and Marichal, Jean-Luc and Mesiar, Radko and Pap, Endre},
  year={2009},
  publisher={Cambridge University Press}
}

@article{ding2023optimization,
  title={An optimization study of diversification return portfolios},
  author={Ding, Chao and Qi, Houduo},
  journal={arXiv preprint arXiv:2303.01657},
  year={2023}
}

@article{LH14,
  title={Online portfolio selection: A survey},
  author={Li, Bin and Hoi, Steven CH},
  journal={ACM Computing Surveys (CSUR)},
  volume={46},
  number={3},
  pages={1--36},
  year={2014},
  publisher={ACM New York, NY, USA}
}

\end{document}